\documentclass[11pt]{article}
\usepackage{float}
\usepackage[margin=1.0in]{geometry}  
\usepackage[utf8]{inputenc}
\usepackage[T1]{fontenc}

\usepackage{textcomp}
\usepackage{xspace}
\usepackage{multirow}

\usepackage{graphicx} 
\usepackage{amsmath}
\usepackage{amsfonts}
\let\oldaa\aa
\usepackage{physics}
\let\aa\oldaa
\usepackage[textsize=tiny]{todonotes}
\usepackage{hyperref}
\usepackage{amsthm}
\usepackage{cleveref}
\usepackage{comment}
\usepackage{amssymb}
\usepackage{cleveref}
\usepackage{physics}
\let\aa\oldaa
\usepackage{colortbl}
\usepackage{array}
\usepackage{enumitem}
\usepackage{complexity}
\usepackage{thm-restate}
\usepackage{thmtools}
\usepackage{bbm}
\usepackage[font=small,labelfont=bf]{caption}

\newcommand{\CAT}{\texttt{CAT}\xspace}

\newcommand{\PARITY}{\texttt{PARITY}\xspace}
\newcommand{\FANOUT}{\texttt{FAN-OUT}\xspace}
\newcommand{\AND}{\texttt{AND}\xspace}
\newcommand{\OR}{\texttt{OR}\xspace}
\newcommand{\NOT}{\texttt{NOT}\xspace}
\renewcommand{\MOD}{\texttt{MOD}\xspace}

\newcommand{\TOFFOLI}{\texttt{TOFFOLI}\xspace}
\newcommand{\MAJORITY}{\texttt{MAJORITY}\xspace}
\newcommand{\TRIBES}{\texttt{TRIBES}\xspace}

\hypersetup{
    colorlinks=true,
    linkcolor=magenta,
    filecolor=magenta,
    urlcolor=cyan,
    citecolor=cyan
}

\renewcommand{\NC}{\textsf{NC}\xspace}
\renewcommand{\AC}{\textsf{AC}\xspace}
\renewcommand{\TC}{\textsf{TC}\xspace}

\newcommand{\ACZ}{\textsf{AC}$^0$\xspace}
\newcommand{\ACZn}{\textsf{AC}$^0$}

\renewcommand{\QAC}{\textsf{QAC}\xspace}
\newcommand{\QACZ}{\textsf{QAC}$^0$\xspace}

\newcommand{\QNCZ}{\textsf{QNC}$^0$\xspace}
\newcommand{\QNCZf}{\textsf{QNC}$^0_f$\xspace}
\newcommand{\QACZf}{\textsf{QAC}$^0_f$\xspace}

\newcommand{\QTCZf}{\textsf{QTC}$^0_f$\xspace}

\newcommand{\calO}{\mathcal{O}}
\renewcommand{\poly}{\text{poly}}

\renewcommand{\R}{{\mathbb{R}}}
\renewcommand{\C}{{\mathbb{C}}}

\newcommand{\N}{{\mathbb{N}}}

\newcommand{\Q}{{\mathbb{Q}}}
\newcommand{\CZ}{\ensuremath{\texttt{CZ}}\xspace}
\newcommand{\B}{{\{0,1\}}}

\renewcommand{\Pr}{\mathop{\bf Pr\/}}
\renewcommand{\E}{\mathop{\bf E\/}}

\newcommand{\eps}{\varepsilon}
\newcommand{\Inf}{{\mathbf{Inf}}}
\newcommand{\one}{{\mathbbm{1}}}
\newcommand{\td}[2]{\mathrm{TD}(#1,#2)}
\newcommand{\dtd}{\mathrm{TD}}

\theoremstyle{plain} 
\newtheorem{theorem}{Theorem}[section]
\newtheorem{lemma}[theorem]{Lemma}

\newtheorem{corollary}[theorem]{Corollary}

\newtheorem{claim}[theorem]{Claim}

\theoremstyle{definition} 
\newtheorem{definition}[theorem]{Definition}

\newtheorem{fact}[theorem]{Fact}

\theoremstyle{remark}

\def\DRAFT{1}

\ifnum\DRAFT=1
\newcommand{\malvika}[1]{\textcolor{purple}{[\textbf{Malvika:} {#1}]}}
\newcommand{\avishay}[1]{\textcolor{orange}{[\textbf{Avishay:} {#1}]}}
\newcommand{\fran}[1]{\textcolor{red}{[\textbf{Fran:} {#1}]}}
\newcommand{\john}[1]{\textcolor{green}{[\textbf{John:} {#1}]}}
\newcommand{\gpt}[1]{\textcolor{blue}{[\textbf{GPT:} {#1}]}}

\fi

\ifnum\DRAFT=0
\newcommand{\malvika}[1]{}
\newcommand{\avishay}[1]{}
\newcommand{\fran}[1]{}
\newcommand{\john}[1]{}
\newcommand{\gpt}[1]{}
\fi

\title{Improved Lower Bounds for QAC$^0$}
\ifnum\DRAFT=1

\author{
  Malvika Raj Joshi\thanks{University of California at Berkeley. \ Email: \url{malvika@berkeley.edu}. \ Supported by NSF grant 2311733 and DOE grant DE-SC0024124.} 
  \and
  Avishay Tal\thanks{University of California at Berkeley. \ Email: \url{atal@berkeley.edu}. \ Supported by an NSF CAREER Award CCF-2145474.}
  \and
  Francisca Vasconcelos\thanks{University of California at Berkeley. \ Email: \url{francisca@berkeley.edu}. \ Supported by the U.S. Department of Energy, Office of Science, under Award No. DE-SC0024124 and the Paul \& Daisy Soros Fellowship for New Americans.}
  \and
John Wright
\thanks{University of California at Berkeley. \ Email:  \url{jswright@berkeley.edu}. \ Supported by an NSF CAREER Award CCF-233971.}
}
\fi
\ifnum\DRAFT=0
\author{
  Anonymous Authors.
}
\fi
\date{}

\usepackage{extra_commands}

\begin{document}

\maketitle

\begin{abstract}
   In this work, we establish the strongest known lower bounds against \QACZ, while allowing its full power of polynomially many ancillae and gates.  Our two main results show that: 
   \begin{enumerate}
       \item Depth~$3$ \QACZ circuits cannot compute \PARITY regardless of size, and require at least $\Omega(\exp(\sqrt{n}))$ many gates to compute \MAJORITY.
       \item Depth~$2$ circuits cannot approximate high-influence Boolean functions (e.g., \PARITY) with non-negligible advantage, regardless of size.
   \end{enumerate} 
  We present new techniques for simulating certain \QACZ circuits classically in \ACZ to obtain our depth~$3$ lower bounds. In these results, we relax the output requirement of the quantum circuit to a single bit (i.e., no restrictions on input preservation/reversible computation), making our depth~$2$ approximation bound stronger than the previous best bound of \cite{rosenthal2021qac0}. This also enables us to draw natural comparisons with classical \ACZ circuits, which can compute \PARITY exactly in depth~$2$ using exponential size. Our proof techniques further suggest that, for boolean total functions, \QACZ circuits do not necessarily provide more power than their classical counterparts. Our third result shows that depth~$2$ \QACZ circuits, regardless of size, cannot exactly synthesize an $n$-target nekomata state (a state whose synthesis is directly related to the computation of \PARITY). This complements the depth~$2$ exponential size upper bound of \cite{rosenthal2021qac0} for approximating nekomata (which is used as a sub-circuit in the only known constant depth \PARITY upper bound). 
  Finally, we argue that approximating \PARITY in \QACZ,  with at least $1/\polylog(n)$ advantage on average, is just as hard as computing it exactly. Thus, extending our techniques to higher depths would also rule out approximate circuits for \PARITY and related problems. 
\end{abstract}

\thispagestyle{empty}
\newpage

\tableofcontents

\thispagestyle{empty}
\newpage
\pagenumbering{arabic} 

\section{Introduction}
In classical computation, the ability to copy information is considered an elementary operation. Every major classically studied circuit class---e.g. $\NC$, $\AC$, and $\TC$---trivially contains the \FANOUT~operation. Quantumly, however, the ability to copy information is more limited and nuanced. For example, the no-cloning theorem explicitly prohibits copying quantum information. Interestingly, even the ability to copy classical information,  via the quantum \FANOUT~gate, appears to offer substantial power in the quantum setting. 

The relative power of the \FANOUT~operation in the classical and quantum settings has largely been formalized through the study of three circuit classes: \ACZ, \QACZ, and \QACZf. \ACZ~is the class of polynomial-sized, constant-depth circuits comprised of unlimited fan-in \AND~and \OR~operations, with \NOT s allowed on the inputs and arbitrary \FANOUT.
Following the seminal works of~\cite{FSS84,Ajtai83,Yao85}, the celebrated work of \cite{hastad1986switch} introduced the switching lemma, which proved tight exponential lower bounds on the size of bounded-depth \ACZ circuits computing or even approximating \PARITY. Later, \cite{lmn1993ac0} used this technique to establish low-degree Fourier concentration of Boolean functions implementable by \ACZ. Beyond profound implications for fields such as cryptography and learning theory, these results demonstrated that high-degree functions cannot be approximated by \ACZ.
Subsequent works refined this picture: \cite{tal2017fourierac0} proved essentially tight bounds on the Fourier spectrum of \ACZ, while  \cite{hastad2017average} established an average-case depth hierarchy theorem, showing that increased depth strictly increases the power of \ACZ even on random inputs.

In 1999, Moore proposed the \QACZ~circuit class as a natural quantum analog of \ACZ~\cite{moore1999qac0} (later published in the work of \cite{green2002acc}). \QACZ~(\QACZf, resp.) is the class of polynomial-sized, constant-depth quantum circuits comprised of arbitrary single-qubit gates and generalized Toffoli gates (with unbounded \FANOUT~gates, resp.). Moore also posed the following fundamental question: 
\begin{quote}
    \centering Is \FANOUT~$\in$ \QACZ? Equivalently, is \QACZ~$=$ \QACZf?
\end{quote}
Note that, quantumly, \FANOUT~is equivalent to \PARITY, up to Hadamard conjugation \cite{moore1999qac0}. Therefore, this question can equivalently be framed as:
\begin{quote}
    \centering Is \PARITY~$\in$ \QACZ?
\end{quote}
Resolving this question would have several profound implications for quantum complexity and quantum computation. In the classical \ACZ setting, even allowing \PARITY~(i.e., $\MOD_{2}$) gates does not yield $\MOD_{m}$ for general $m$~\cite{razborov1987lower,smolensky1987algebraic}. In sharp contrast, if \FANOUT/\PARITY~$\in$~\QACZ, then:
\begin{enumerate}
    \item \QACZ contains \FANOUT, and thus \ACZ and is substantially more powerful than \ACZ \cite{moore1999qac0} (\PARITY~$\notin$~\ACZ~ \cite{hastad1986switch}).
    \item For any integer $m$, $\MOD_m$ can be implemented in \QACZ \cite{green2002acc}. 
    \item $n$-qubit GHZ (cat) states can be prepared in constant depth using only single-qubit gates and a single $\FANOUT/\PARITY$ operation~\cite{moore1999qac0,hoyer2005fanout}.

    \item \QACZ~can perform many powerful computations, including:  majority, threshold[$t$], exact[$t$], counting, sorting, arithmetic, phase estimation, and the quantum Fourier transform \cite{hoyer2005fanout}.
    \item Strong pseudo-random unitaries (PRUs) are implementable in \QACZ~\cite{foxman2025pru}.

    \item The quantum shallow-depth hierarchy collapses, i.e. \QNCZf~= \QACZf~=\QTCZf~and these classes can be characterized by just \FANOUT gates and single qubit unitaries \cite{takahashi2012collapse}. 
\end{enumerate}

Furthermore, since \QACZ~includes gates of unbounded width, standard light-cone techniques are insufficient for proving circuit lower-bounds. Thus, resolving whether \PARITY~$\in$~\QACZ~will likely result in novel techniques for proving more general quantum circuit lower-bounds. In fact, previously-developed techniques have already led to exciting applications. For example, \QACZ~Fourier concentration established by the lower-bound approach of \cite{nadimpalli2024pauli} led to sample-efficient algorithms for learning single-output \QACZ~channels \cite{nadimpalli2024pauli,bao2025learning} and a time-efficient algorithm for average-case learning of \QACZ~unitaries \cite{vasconcelos2025learning}. Additionally, the exponential-size implementation of \PARITY~in \QACZ~proposed by \cite{rosenthal2021qac0} enabled the compression of strong \QACZf~PRUs to weak \QACZ~PRUs by \cite{foxman2025pru}.

Despite substantial effort in proving both upper and lower-bounds \cite{fang2006qlb,pade2020depth2qaccircuitssimulate,rosenthal2021qac0,nadimpalli2024pauli,anshu2025computational,fenner2025tightboundsdepth2qaccircuits}, Moore's question has remained unresolved for nearly three decades.  Prior to this work, the strongest known \PARITY~lower-bounds were 
    either in the setting with limited, \emph{slightly super-linear}, ancillae \cite{anshu2025computational} or with unlimited ancillae, but only up to depth-$2$ \cite{rosenthal2021qac0,fenner2025tightboundsdepth2qaccircuits}.
  As demonstrated by the only known constant-depth upper-bound for \PARITY \cite{rosenthal2021qac0}, and lower-bounds against circuits with limited ancillae \cite{nadimpalli2024pauli,anshu2025computational}, the main power of \QACZ circuits comes from their use of super-linear ancillae to generate entanglement, accounting for the lack of \FANOUT. 

In this work, we introduce novel techniques for proving \QACZ~circuit lower-bounds, enabling us to give the strongest fixed-depth lower-bounds for \QACZ~to-date, while still allowing its full power of polynomial ancillae. In \Cref{sec:d3poly}, we prove the first \emph{depth-$3$} lower-bounds for \QACZ, ruling out computation of exact \PARITY and \MAJORITY with sub-exponential size. For \PARITY specifically, we further prove a size-independent depth-3 lower-bound.  

In \Cref{sec:depth_2_influence}, we also prove a Fourier-tail decay bound for depth-$2$ \QACZ~circuits with unlimited ancillae, demonstrating that they have low total influence. For these results, we treat the output of the circuit as a single bit on a designated register with no constraints on other registers. This makes our depth-$2$ lower bounds stronger than the previously known 
depth-2 approximation bound of \cite{rosenthal2021qac0} which requires the circuit to preserve the state on the input registers. This also allows us to draw analogies with classical \ACZ, where the output is a single bit. Interestingly, our results contrast what is known for classical circuits, since exponential-size \ACZ circuits can compute \PARITY exactly in depth-$2$.  

Finally, in \Cref{sec:generalized_neko}, we show that depth-$2$ \QACZ circuits cannot exactly synthesize a so-called ``nekomata'' state, which is closely related to \FANOUT. For example, \cite{rosenthal2021qac0} achieves a constant-depth upper-bound for \PARITY by first using an exponential-size depth-2 \QAC circuit to approximately synthesize a nekomata state. We show that such a state on $n = \omega(1)$ targets cannot be exactly prepared in depth-$2$, even with unlimited ancillae. This complements the only known $O(1)$-depth upper-bound for approximating nekomatas, due to \cite{rosenthal2021qac0}.

We present new techniques for simulating \QACZ circuits for exact computation of boolean functions in \ACZ. Although it is known that \QACZ (or even \QNCZ) provides more power than \ACZ for search problems \cite{watts2019separation}, and worst-case bounded error regimes for decision problems (\textsf{BQAC$^0$}) \cite{grier2026mathsfqac0containsmathsftc0with}, there are no known decision separations in which \QACZ circuit is \emph{exact}.
We justify our approach in \Cref{sec:aproofs} though a \QACZ reduction from exactly computing \PARITY to approximating \PARITY with any inverse-polylogarithmic advantage on a random input. We conclude that it is sufficient to establish lower bounds against \emph{exact} \QACZ circuits for any of the problems connected to \PARITY described above \cite{green2002acc,hoyer2005fanout,grier2024threshold}, to also rule out these approximate regimes. 

\subsection{Prior Work} \label{sec:prior_work}
\begin{table}[t!]\label{tab:prior}
    \centering
    \footnotesize
    \renewcommand{\arraystretch}{1.4}
    \setlength{\tabcolsep}{4pt}
    \begin{tabular}{|c|c|c|c|c|c|}
        \hline
         \cellcolor{gray!10}Result Type& \cellcolor{gray!10}Paper & \cellcolor{gray!10} \cellcolor{gray!10}Comp. Type & \cellcolor{gray!10}Output Type & \cellcolor{gray!10}depth-& \cellcolor{gray!10}\# Ancillae\\
         \hline
         \hline
         \multirow{2}{*}{\shortstack{\shortstack{\PARITY/Nekomata\\Upper-Bounds}}} & \cite{rosenthal2021qac0} & Approximate & Input-Preserving \PARITY & $d\geq 7$ & $\exp(n^{\calO(1/d)})$ \\
         \cline{2-6}
         & \cite{rosenthal2021qac0} & Approximate & Nekomatas & 2 & $\exp(n^{1+o(1)})$ \\
         \hline
         \hline
         \multirow{3}{*}{\shortstack{Boolean Function\\Lower-Bound Via\\Structural Results}} & \cite{nadimpalli2024pauli} &  Approximate & High-Degree Boolean Funcs & $d=\calO(1)$ & $n^{\Omega(1/d)}$\\
         \cline{2-6}
         & \cite{anshu2025computational} &  Approximate & High-Degree Boolean Funcs & $d=\calO(1)$ & $\Omega\left(n^{1+1/3^{d}}\right)$\\
         \cline{2-6}
         & Thm \ref{thm:inf_lb}/\ref{thm:influence} &  Approximate & \cellcolor{green!15}High-Influence Boolean Funcs & 2 & \cellcolor{green!15}$\infty$ \\
         \hline
         \hline
         \multirow{2}{*}{\shortstack{\shortstack{CAT/Nekomata\\Lower-Bounds}}} &\cite{rosenthal2021qac0} & Approximate & CAT States & 2 & $\infty$ \\
         \cline{2-6}
         & Thm \ref{thm:d2_lb}/\ref{cor:nekolb} & Exact & \cellcolor{green!15}Generalized Nekomatas & 2 & $\infty$ \\
         \hline
         \hline
         \multirow{4}{*}{\shortstack{Restricted-Ancillae\\ \PARITY~Lower-Bounds}} 
         &\cite{bera2011qac0} & Exact & \PARITY & $o(\log n)$ & 0 \\
         \cline{2-6}
         &\cite{fang2006qlb} & Exact & \PARITY & $o(\log n)$ & $n^{1-o(1)}$\\
         \cline{2-6}
         & \cite{nadimpalli2024pauli} &  Approximate & \PARITY & $d=\calO(1)$ & $n^{\Omega(1/d)}$\\
         \cline{2-6}
         & \cite{anshu2025computational} &  Approximate & \PARITY & $d=\calO(1)$ & $\Omega\left(n^{1+1/3^{d}}\right)$\\
         \hline
         \hline
    \multirow{5}{*}{\shortstack{Restricted-Depth\\ \PARITY/\MAJORITY \\Lower-Bounds}} 
    & \cite{rosenthal2021qac0} &  Approximate & Input-Preserving \PARITY & 2 & $\infty$ \\
           \cline{2-6}
    & \cite{rosenthal2021qac0} &  Approximate & \PARITY & $d \geq 1$ & $\Omega(n/d)$ \\
       \cline{2-6}
   & Cor \ref{thm:par_corr_bound}/\ref{thm:corr_proof} &  Approximate & \cellcolor{green!15} \PARITY & 2 & $\infty$ \\
   \cline{2-6}
    & \cite{fenner2025tightboundsdepth2qaccircuits} &  Exact & \PARITY & 2 & $\infty$ \\
    \cline{2-6}
         & Thm \ref{th:d3lb}/\ref{thm:depth_3_inf_anc} &  Exact & \PARITY & \cellcolor{green!15}3 & $\infty$ \\
    \cline{2-6}
         & Thm \ref{th:d3lb}/\ref{cor:d3anclb} &  Exact & \cellcolor{green!15}\MAJORITY & \cellcolor{green!15}3 & \cellcolor{green!15}$\exp(n^{\Omega(1)})$ \\
         \hline
    \end{tabular}
    \caption{Upper- and lower-bounds for computation of Boolean functions and nekomata/CAT states in constant-depth \QAC. The results are grouped together by type. For lower-bound results, with the exception of \cite{bera2011qac0}, the depth column states the values of $d$ for which the corresponding ancilla bounds apply. For prior works, the paper is referenced, whereas for novel results from this work, theorem references (main paper/proof section) are provided. For each result we list the computation type (exact versus approximate), output type, explicit depth, and explicit ancilla count. Key improvements achieved in our work, relative to prior works, are highlighted in green.}
    \label{tab:prior_works}
\end{table}

We will now briefly summarize known \QACZ~upper-bounds, lower-bounds, and structural results prior to this work, as listed in \Cref{tab:prior_works}. We will first discuss \QACZ lower-bounds for \PARITY, which can be split into two main categories: (i) restricted-ancillae and (ii) restricted-depth. We will also describe corresponding \QACZ~low-degree structural results and nekomata/CAT state preparation lower-bounds. Finally, we conclude by describing the only known upper-bound for approximately computing \PARITY~and nekomata states in \QACZ, using exponential ancillae.

\paragraph{Restricted-Ancillae \PARITY Lower-Bounds.}
The size of a \QACZ circuit is closely related to the number of ancillae it uses. 
By definition, the number of ancillae in \QACZ circuits is allowed to be an arbitrary polynomial in $n$.  
The first category of \PARITY lower-bounds \cite{fang2006qlb, bera2011qac0, nadimpalli2024pauli, anshu2025computational} focuses on generic depth-$d$ \QACZ~circuits with ancillae limited to $o(n^2)$.
The proofs of \cite{nadimpalli2024pauli,anshu2025computational} follow by showing that the circuit's Heisenberg-evolved single-qubit ``output'' measurement Pauli/projector can be approximated, to high precision, by low-degree objects. Beyond ruling out the computation of \PARITY, these low-degree structural results enable correlation bounds against generic \emph{high-degree Boolean functions}, such as \MAJORITY~and $\MOD_k$. The key caveat of this low-degree approach, however, is that it only holds for a depth dependent number of ancillae, which \cite{anshu2025computational} pushed to slightly super-linear in $n$ for arbitrary constant-depth \QACZ~circuits. 

\paragraph{Restricted-Depth \PARITY~Lower-Bounds.} The second category of \PARITY~lower-bounds focuses on fixed-depth \QAC~circuits, without any extra constraints on the ancillae. Specifically, \cite{rosenthal2021qac0} established a depth-2 average-case approximate lower-bound, while  \cite{fenner2025tightboundsdepth2qaccircuits} established a depth-2 worst-case exact lower-bound against \PARITY~$\in$ \QACZ~, both with unlimited ancillae. All the known lower bounds that do not impose any restrictions on the number of ancillae beyond the default $\poly(n)$ fall under this category and only go up to depth-$2$. Both the upper and lower bounds of \cite{rosenthal2021qac0} correspond to circuits with an $n$-bit output that preserve the state on the input qubits, and we refer to these circuits as ``input-preserving''.

\paragraph{\CAT/Nekomata Lower-Bounds.} In \cite{moore1999qac0}, Moore proved that there exist reductions between computing \PARITY/\FANOUT~and preparing the $n$-qubit CAT state, $\ket{\Cat_n} = \frac{1}{\sqrt{2}} \left(\ket{0^n}+\ket{1^n}\right)$. In recent work \cite{rosenthal2021qac0}, Rosenthal introduced the notion of an $n$-qubit ``nekomata'' state, of form
\begin{align} \label{eqn:neko}
    \ket{\text{Nekomata}} = \frac{1}{\sqrt{2}} \left(\ket{0^n} \ket{\psi_\alpha} +  \ket{1^n} \ket{\psi_\beta} \right).
\end{align}
The nekomata is similar to $\ket{\Cat_n}$, but allows for each branch to have an arbitrary ancillary state (i.e. the normalized states $\ket{\psi_\alpha}$ and $\ket{\psi_\beta}$). In this case, where the two branches are equally weighted (with probability $1/2$), we refer to the state as a ``balanced'' nekomata. 

\cite{rosenthal2021qac0} showed $O(1)$-depth reductions between preparing $n$-qubit nekomata states and computing \FANOUT/\PARITY, analogous to those for $\ket{\Cat_n}$, and extended all these reductions to the approximate setting. Rosenthal also showed that any depth-$d$ \QACZ circuit that approximates an $n$-qubit nekomata must have  $\Omega(n/(d+1))$ multi-qubit gates acting on the targets. 

\paragraph{\PARITY~and Nekomata Upper-Bounds.} Despite several lower-bound results for \QACZ, there is only one known upper-bound for approximating \PARITY in constant-depth \QAC.
Notably, \cite{rosenthal2021qac0} gave a depth-7 \QAC~circuit for approximating $n$-qubit \PARITY, using $\exp(n^{1-o(1)})$ gates (thereby requiring more resources than permitted in polynomial-sized \QACZ). To achieve this, Rosenthal first gave a depth-$2$ circuit using an approximate $n$-qubit nekomata and then used it to obtain a depth-7 circuit for approximating \PARITY. Due to the recursive nature of \PARITY, for any depth $k=7d$, this implies \QACZ circuits of size roughly $\exp(n^{1/d})$ approximating \PARITY.

\subsection{Our Results} \label{sec:our_results}
In this work we study fixed-depth \QAC~circuits, specifically with depth-$\leq 3$, and do not impose additional restrictions on ancillae. Note that for constant depth circuits, limiting the ancillae also limits the size of the circuit because each qubit can belong to at most $d$ gates. Our motivation for studying fixed-depth \QAC~circuits stems from the large gap between the only known constant-depth upper-bound for \PARITY \cite{rosenthal2021qac0}, which uses an exponential number of gates, and
the best known techniques for arbitrary-depth circuits \cite{fang2006qlb, nadimpalli2024pauli, anshu2025computational}, which fail to rule out even $O(n^2)$-sized circuits for depth-$\geq 2$. 
Our main results are summarized below. 

Our first result shows that depth-$3$ circuits (i) cannot compute \MAJORITY using only \emph{sub-exponential} gates and (ii) cannot compute \PARITY regardless of size. The informal theorem statement is as follows, with the full proof given in \Cref{sec:d3poly} (\Cref{thm:depth_3_inf_anc}). We note that both results apply regardless of the number of ancillae. 

\begin{theorem}[Depth-$3$ \MAJORITY Lower-bound]\label{th:d3majlb}
  Let $C$ be a depth-$3$ \QACZ~circuit $C$ on $n$ inputs with $m \leq 2^{n^{o(1)}}$ gates such that on every input $\x \in \bin^n$, $C$ produces the state $\ket{f(\x)}_t$ on a designated output $t$. Then, $f(\x)$ cannot be the \PARITY or \MAJORITY function.
\end{theorem}

\begin{theorem}[Depth-$3$ \PARITY~Size-Independent Bound]\label{th:d3lb}
Let $C$ be a depth-$3$ \QACZ circuit with $n > 100$ input qubits and an arbitrary number of ancillae and gates, such that on every input $\x \in \bin^n$, the circuit $C$ produces the state $\ket{f(\x)}_t$ on a designated output qubit $t$. Then, $f$ cannot be the \PARITY function.
\end{theorem}

\noindent We prove these depth-$3$ lower-bounds in \Cref{sec:d3poly}, by first showing that, after applying a quantum restriction (that keeps $\Omega(n)$ input bits alive), 
the output of the remaining depth-$(\leq 3)$ \QACZ circuit can be simulated by a classical \ACZ circuit of a slightly larger depth and size (\Cref{thm:d3ac0}). This immediately implies an $\exp(n^{\Omega(1)})$-size lower-bound due to known lower-bounds for \PARITY in \ACZ \cite{hastad1986switch}. Our techniques also apply to other functions that behave  \MAJORITY giving us the first lower-bound against the \MAJORITY~function for \QACZ circuits with polynomial ancillae. 

We observe that at very low depths ($\leq 2$), \QACZ circuits exhibit certain monotonicity properties. For \PARITY we can exploit these properties to strengthen the bound from \Cref{th:d3majlb} to a size-independent bound, \Cref{th:d3lb}, using carefully-designed classical restrictions. These restrictions rely on the property of the \PARITY function being invariant under arbitrary classical restrictions, unlike \MAJORITY, which requires balanced restrictions. Moreover, \PARITY and \MAJORITY are known to be equivalent up to a $O(1)$ factor in depth for \QAC circuits \cite{hoyer2005fanout}. Thus we expect \PARITY lower bounds for higher depths to depend on the size as in \Cref{th:d3majlb}.

Our second main result is a structural result for depth-2 \QACZ circuits. Namely, we show these circuits have \emph{low total influence}, regardless of the number of ancillae. The informal theorem statement is as follows, with the full proof given in \Cref{sec:depth_2_influence} (\Cref{thm:influence}).
\begin{theorem}[Depth-2 Influence Upper-Bound] \label{thm:inf_lb}
Let $C$ be a depth-2 \QACZ circuit with $n$ input qubits and of ancillae.
Consider the function $f_C: \B^n \to [0,1]$ defined by $f_C(x) = \Pr[\text{$C$ accepts $x$}]$. Then, $f_C$ has  total influence $O(\log n)$.
\end{theorem}
\noindent
First, note that this result is \emph{tight}. Specifically, consider the Tribes function $\TRIBES(x) = \lor_{i=1}^{s}\land_{j=1}^{w}x_{i,j}$ that can be exactly implemented by depth-$2$ \QACZ circuits with $s+1$ ancillae and has $\Inf[f] = \Theta(\log n)$ for a specific choice of parameters ($s = \Theta(n/\log n)$ and $w=\Theta(\log n)$).

\Cref{thm:inf_lb} follows from a proof of stronger Fourier tail bounds for functions computable by depth-2 \QACZ~circuits, using novel entropy-based arguments. We also use these improved Fourier tail bounds to prove the following correlation bound against \PARITY (see \Cref{thm:corr_proof}), thus ruling out the approximate computation of \PARITY~in depth-2 \QACZ.
\begin{corollary}
    \label{thm:par_corr_bound}
    $f_C$ has correlation at most $\exp(-\Omega(\sqrt{n}))$  with \PARITY.
\end{corollary}

Our next result is a depth-2 \emph{unlimited ancillae} \QACZ~lower-bound against \emph{exact} preparation of nekomata states. In fact, we are able to rule out a depth-$2$ circuit for creating any entangled state that only has nonzero amplitude in two subspaces, corresponding to the all $0$s, $\ket{\0}_{[n]}$ and all $1$s, $\ket{\1}_{[n]}$ branches on any set of $n > 1$ ``target'' qubits. We call such states, as below,  ``generalized nekomatas'',
\begin{align}
  \ket{\psi} =  \alpha \ket{0^n} \ket{\psi_\alpha}
      + \beta \ket{1^n} \ket{\psi_\beta}, 
    &\qquad \text{with } \alpha \neq 0 \text{ and } \beta \neq 0.
\end{align}
The informal theorem statement is as follows, with the full proof given in \Cref{sec:generalized_neko} (\Cref{cor:nekolb}).
\begin{theorem}[Depth-2 Nekomata Lower-Bound] \label{thm:d2_lb}
A depth-$2$ \QACZ circuit with an arbitrary number of ancillae and gates cannot exactly prepare a ``generalized'' nekomata on $n > 4$ targets.
\end{theorem}
\noindent Note that this bound is also \emph{tight}, since any $\ket{\Cat_n}$ is also a $n$-nekomata and we can construct the state $\ket{\Cat_4}$ in depth-$2$ by constructing $\ket{\Cat_2} = \ket{\mathrm{EPR}}$ in depth-$1$. Rosenthal \cite{rosenthal2021qac0} gave a circuit for approximating an $n$-nekomata with exponential ancillae in depth-2, which is then used as a sub-circuit to obtain a parity circuit. Our result implies that such an approximation cannot be made exact by simply using more ancillae, thus ruling out the exact analog of Rosenthal's parity circuit in the same depth. 

Note that we define these generalized nekomata states only to aid proving our depth-$2$ bound against balanced nekomatas and we do not expect highly unbalanced instances to be useful or powerful. For example, the state $\sqrt{(1-\eps)} \cdot \ket{0^n} + \sqrt{\eps} \cdot \ket{1^n}$, which is a generalized $n$-nekomata, can be simply approximated by $\ket{0^n}$. 

Our final result provides evidence for the robustness of our exact lower-bound techniques by showing that computing \PARITY approximately on a random input in \QACZ is just as hard as computing it exactly.   
The informal statement is as below, 
\begin{theorem}[$1/\polylog(n)$-adv for \PARITY is equivalent to exact in \QACZ]\label{thm:approx_inf}
  Let $C$ be a depth-$d$ \QACZ circuit with $n$ input qubits and $a$ ancillae such that $f_C(\x)$, the acceptance probability of $C$ has correlation at least $1/\log(n)^\delta$ with \PARITY. Then, there exists a depth-$O(d)$ with circuit with $O(a \cdot n^{\delta})$ ancillae that \emph{exactly} computes \PARITY on $n$ qubits.
\end{theorem}
We describe the main ideas behind this reduction in \Cref{sec:reduction_tech} (\Cref{thm:approxtoexactpar}). The formal proof follows though a combination of several standard \QACZ techniques and we include it in \cref{sec:aproofs}.
\section{Techniques}

\subsection{Key Challenges and High-Level Intuition} 
The main challenge in proving lower-bounds for \QACZ~circuits stems from their ability to use multiple ancilla qubits in a single gate. In contrast, the concept of ancillae is foreign to classical \ACZ~circuits, as they have no advantage in using ancillae (e.g., they can be replaced with \FANOUT).

Perhaps a more suitable classical analogue of the use of ancillae in \QACZ~comes from \emph{classical reversible circuits}. These circuits are significantly weaker than \ACZ~and limited only to use of reversible gates (i.e. \TOFFOLI~and \NOT). They are equivalent in power to bounded-read \ACZ formulae, where each variable appears at most $2^d$ times.
In classical reversible circuits, each gate can only spread the influence of a bit to one other bit. Thus, in depth-$d$ we can only compute \FANOUT of size $2^d$. Therefore, depth-$\log n$ is necessary to perform \FANOUT of size $n$.

In the context of decision problems, a \QACZ~circuit without ancillae has roughly the same power as a mere classical reversible circuit. For uniformly random inputs, the state after each layer is maximally mixed (due to unitarity and uniform input distribution).
This was the key insight of \cite{nadimpalli2024pauli}, i.e., gates acting on many qubits are rarely active and can be removed with little error, thus enabling similar light-cone arguments to the classical reversible setting.

In \QACZ~circuits with ancillae, a gate can use multiple ancillae such that all these ancillae become correlated with the input after the gate. However, even in the classical setting, there is a crucial distinction between the notions of \FANOUT in \ACZ circuits and in reversible circuits, in which a \FANOUT operation is required to preserve reversibility.
Simply being entangled with many ancillae is not an indication of the latter since the input cannot be reliably recovered from a single ancilla alone.
Our techniques exploit this distinction. We observe that reversibility puts significant constraints on the type of computation that  \QACZ~circuits can perform. To our knowledge, these constraints cannot be bypassed with only polynomially many ancillae. 
To uncover the true nature of their computation, it is crucial to study these circuits without severely limiting the number of ancillae, i.e., by considering their power with an arbitrary polynomial number of ancillae.

\subsection{Setup}
We define the output of a \QACZ circuit to be the output on a single designated register $t$, as in \Cref{def:qacoutput}. We \emph{do not} require any particular output state for qubits other than the output register $t$.
This is a weaker requirement than ``clean-computation'' or ``dirty-computation'' used in the prior works that establish an approximate depth-$2$ lower bound  \cite{rosenthal2021qac0}, making our results stronger. We find that, for fixed-depth circuits, relaxing the output requirement to a single register reveals more about the structure of the circuit's computation of $f$. 

\begin{restatable*}[Circuit computing classical function $f$]{definition}{qacoutput}\label{def:qacoutput}
Let $C$ be a \QACZ circuit with a designated target register $t$ and associated output basis $\lr{\ket{\mu_0}, \ket{\mu_1}}$
with $\braket{\mu_0|\mu_1} = 0$. We say that $C$ {\emph computes a Boolean function} $f: \bin^n \to \bin$, if 
  for all $\x \in \bin^n$, the output of $C(\x)$ on $t$ is exactly $\ket{\mu_{f(\x)}}$ and unentangled with other qubits. 
  Equivalently, 
    $$\forall \x \in \bin^n \ \ \bra{\mu_{f(\x) \oplus 1}} \cdot C(\x) = 0.$$
\end{restatable*}

\subsection{Depth-3 circuits cannot compute Parity or Majority}
Our proof consists of several components that we detail below.  
\subsubsection{Block Diagonalization of Gates} \label{sec:pregates}
All the multi-qubit classical reversible gates (\AND and \OR, composed with \NOT) can be viewed as reflection gates. As shown by \cite{rosenthal2021qac0}, this lets us define our gate set as reflections about arbitrary separable states. 
For example, a Toffoli gate with controls on qubits in $S$  and a target $t$ is given by \[(I - 2\kb{\1}_S \tens \kb{+}_t).\] 

Although reflection gates have no inherent ``controls'' and ``target'', we can arbitrarily partition the qubits into controls and targets to view these gates as controlled unitaries as follows. For a gate given by $G(S) = (I - 2\kb{\vth}_S)$, and any partitioning of qubits in $G$ into two sets $S = (X,Y)$, $G$ can be diagonalized as follows, 
\begin{align}\label{eq:pre_gatedec}
  G(S) = (I - \kb{\vth}_{X}) \tens I_Y +\kb{\vth}_X \tens (I - 2\kb{\vth})_Y
\end{align}
This can be interpreted as applying a smaller reflection, $(I - 2\kb{\vth})_Y$, on the targets $Y$, in the $\ket{\vth}_X$ subspace on the controls, and doing nothing to $Y$ in the orthogonal subspace on the controls. 
Additionally, these two subspaces are invariant under $G(S)$ because any measurement on the qubits $X$ in the $\ket{\vth}$ basis commutes with such a gate. Therefore, any projector that is either $\kb{\vth}_X$ or orthogonal to $\kb{\vth}_X$ commute with $G$ and induces a unitary on $Y$. 
For example, for any $q \in S$, and $S' = S \setminus q$, conditioning on outcome $\ket{\th_q}$ on $q$ gives, 
\begin{align}
  \kb{\th_q} \cdot G(S) &= \kb{\th_q} \tens I_{S'} - 2\kb{\vth}_S \\
                   &=  \kb{\th_q} \tens (I_{S'} - 2\kb{\vth}_{S'}) \\
                   &= G(S') \tens \kb{\th_q}  
\end{align}
Where $G(S')$ is a valid reflection gate on the smaller subset $S'$. Also, conditioning on outcome $\ket{\th^\perp_{q}}$ gives,
\begin{align}
  \kb{\th^\perp_q}  G(S) &= \kb{\th^\perp_q} \tens I_{S'} - 2 \underbrace{\kb{\th^\perp_q} \cdot \kb{\vth}_S}_{= \braket{\th_q | \th^\perp_q} = 0} \\
&=  \kb{\th^\perp_q} \tens  I_{S'} \label{eq:kill}
\end{align}
Such measurements are unaffected by whether they are performed on the state before the gate or after. More generally, if the state of the qubits on $S$ is either $\ket{\vth}_S$ or some (possibly entangled) state orthogonal to $\ket{\vth}_S$, the gate only adds a global phase to the state and does not create any additional entanglement. The same can be said of any subset $S' \subseteq S$ that is in the state $\ket{\vth}_{S'}$, the gate does not create any entanglement on qubits in $S'$, but it may do so on qubits $S \setminus S'$.  

If a state on a subset of qubits $S$ is denoted as $\ket{\psi}_S$, a pure state (as opposed to a mixed state $\rho_S$), then the qubits in $S$ are not entangled with anything outside of $S$. For example, each single-qubit component of a fully separable state $\ket{\vth_S}$ is a pure state $\ket{\th_q}$ on qubit $q \in S$. 

\subsubsection{Quantum Analogue of Restrictions}
A central building block of our depth-$3$ lower-bound is a new technique, referred to as the ``clean-up step''.  
This step applies ``quantum restrictions" to simplify the first layer of gates in the circuit, such that each of them depends on at most 1 input qubit. We then show that for \emph{any} cleaned-up circuit of depth-$\leq 3$, the circuit's output can be simulated classically by a \ACZ~circuit whose size is polynomial in $n$.

Recall that a depth-$d$ \QACZ circuit consists of $d$ layers of gates on $n$ input qubits, denoted by $[n]$, and ancillae qubits $A$ with $|A| = \poly(n)$. The reversibility property of the circuit enforces that each qubit appears in at most one gate per layer.  As evidenced by classical techniques, such as the Switching Lemma \cite{hastad1986switch}, it is useful to ensure that all coordinates have disjoint light-cones in the bottom layer, i.e, each layer-$1$ gate contains at most one input qubit. Our clean-up step lets us simply convert any \QACZ~circuit computing a Boolean function $f$ that behaves well under restrictions (e.g., \PARITY and \MAJORITY) to one that computes $f$ on a subset of at least $n/3$ coordinates, with the additional guarantee that each coordinate appears in at most one gate. This is reminiscent of clean-ups performed on \ACZ~circuits using random restrictions that simplify the first layer of gates. 

As in \cite{nadimpalli2024pauli}, when analyzing \QACZ for approximate computation, we argue in our approximation lower bound that gates containing many input coordinates can be replaced by identity, while incurring a small error.  However, this argument is not suitable for analysis of circuits with exact output, because the resultant circuit is no longer exact. In classical reversible circuits, however, this type of simplification can be achieved via deterministic restrictions, which preserve the exact computation. For instance, restricting a single coordinate of each layer-$1$ \AND gate to $0$ kills the gate and leaves the rest of the coordinates unrestricted. This means that one can perform such a clean up on classical reversible circuits and leave at least half of the coordinates unrestricted (by restricting only coordinates from gates of width at least $2$).  However, since gates in a \QACZ circuit can be reflections about arbitrary separable states, we cannot always achieve this effect with a classical restriction. 

Our technique, therefore, is a generalization of these deterministic restrictions for quantum circuits. The main idea is to restrict the input to a state lying in a subspace orthogonal to the gate's reflection which, in effect, deactivates the gate. Additionally, we want to be able to apply these restrictions in \QACZ. Finally, to ensure that the circuit still correctly outputs $f(\x)$ after this quantum restriction, the restriction is specifically chosen to be a superposition of the classical restrictions that keep $f(\x)$ invariant. For $f = \PARITY$ these are precisely the states formed by superpositions of classical restrictions of the same parity, and for $f = \MAJORITY$, these are superpositions of classical restrictions with equal number of $0$s and $1$s.
We describe our clean-up step as the following lemma and defer its proof to \Cref{sec:d3poly}. 
\begin{restatable*}[Clean-Up Lemma]{lemma}{cleanup}\label{lem:l1cleanup}
  Let $C$ be a depth-$d$ circuit that computes $f(\x)$, where $f(\x)$ is \PARITY or \MAJORITY, on $n > 2$ coordinates with a separable ancilla starting state. Then, there is a depth-$d$ circuit $C'$ that computes $f(\x)$ on $n/3$ coordinates using a separable ancilla state and satisfies that every layer-$1$ gate of $C'$ contains at most $1$ input qubit.  
\end{restatable*}

It turns out that performing this clean-up step gets us most of the way to an exact depth-$2$ lower bound, stated below as \Cref{cor:d2par}.
We note that \cite{fenner2025tightboundsdepth2qaccircuits} already provide a tighter bound, in terms of $n$, for \PARITY in depth-$2$. Nevertheless, we include our alternative proof of \Cref{cor:d2par} below, which also applies to \MAJORITY and encompasses the main ideas used in \Cref{sec:d3poly} for the depth-$3$ bound.

\begin{corollary}[Depth-2 Exact \PARITY/\MAJORITY Lower-Bound]\label{cor:d2par}
  Let $C$ be a $n$-input depth-$2$ \QACZ circuit with an arbitrary number of gates and ancillae, whose ancillae start in a separable state. Then, $C$ cannot compute \PARITY for $n > 6$ or \MAJORITY for $n > 12$ coordinates. 
\end{corollary}
\begin{proof}
  Suppose $C$ computes $f(\x)$ where $f(\x)$ is \PARITY or \MAJORITY. Then, we can obtain a cleaned-up version of $C$ that computes $f(\x)$ on $n' \geq n/3$ coordinates by applying \Cref{lem:l1cleanup}. Now we will proceed to prove that a \emph{cleaned up circuit} $C$ on $n$ inputs cannot compute \PARITY for $n > 2$ or \MAJORITY for $n > 4$. 
  Let $G(S,t) = (I - 2\kb{\vth}_{S,t})$ be the final gate of $C$, containing the output register $t$. 
  We consider two cases depending on the state on $S$ in $C(\x)$:
\begin{description}
    \item[Case 1:] If the state on $S$ in $C(\x)$ is always $\ket{\vth}_S$, regardless of the input, we can simplify $G(S,t)$ to a single-qubit unitary on $t$. This results in a depth-$1$ circuit in which the light-cone of $t$ contains a single coordinate. This cannot happen unless $n = 1$, because the output of $t$ does not depend on the other coordinates. 
    \item[Case 2:] Otherwise, there is some qubit $q \in S$ such that $\kb{\pth_q} \tens I \cdot C(\x)$ is not always $0$.  Pick $b$ such that, for the two qubit projector, $\Pi_{q,t} :=\kb{\pth_q} \tens  \kb{\mu_b}_t$, the quantity $\Pi_{q,t} \tens I \cdot C(\x)$ is also not always $0$. This is always possible because $\kb{\mu_1} + \kb{\mu_0} = I$.
  Letting $C^1$ denote the depth-$1$ sub-circuit of $C$ we have, 
\begin{align}
  \Pi_{q,t} \tens I \cdot G(s,t) \cdot C^1(\x)  &= \Pi_{q,t} \tens I \cdot C^1(\x)
\end{align}
Then, $\Pi_{q,t}$ has at most two gates in its $C^1$ light-cone and, thus, depends on at most two input coordinates (due to our cleanup). 
Therefore, there exists a  classical restriction $R$ on these $\leq 2$ coordinates such that $\Pi_{q,t} \tens I \cdot C|_R(\y) \neq 0$ on all strings $\y \in \bin^{n - |R|}$. If $C$ computes $f(\x)$ correctly, it must be that $f_R(\y) = b$ is a constant function. 
For \PARITY this cannot happen when $n-|R|> 0$  and for \MAJORITY this cannot happen when $n-|R| > 2$, regardless of the values $R$ chooses for these (at most) $2$ bits. Therefore, $C$ cannot compute \PARITY on $n > 2$ coordinates or \MAJORITY on $n > 4$ coordinates. 
\end{description}
\end{proof}

\subsubsection{Classical Simulation of Shallow Quantum Circuits}

For our depth-$3$ lower bound, we show that parts of the circuit can be simulated classically in the following sense. The circuit can be described as a unitary transformation on the $2^{n + |A|}$ dimensional space of the inputs and ancillae $A$. Since our ancillae start in a fixed state $\ket{\0}_A$, the state of the circuit at any point of the computation is restricted to lie in a smaller, rank $2^n$ subspace. Additionally, for circuits as in \Cref{def:qacoutput}, the final state of the circuit fully lies in one of two orthogonal rank $2^{n-1}$ subspaces, determined by the classical function of the input. Extending this idea to other circuits,   
for any projector $\Pi$ and circuit $C$, we define a classical function $f_{C,\Pi} : \bin^n \to \bin$  that captures whether or not $C(\x)$ has a component in the +1 eigenspace of $\Pi$. We call this the \emph{activation function} of $\Pi$ on $C$.
That is, \[f_{C,\Pi}(\x) = [\Pi \cdot C(\x)  \neq 0]\] where $[\cdot]$ stands for the indicator of an event.

Observe that for any  $C$ that exactly computes a classical function $f(\x)$ on a target register $t$ in the ${\ket{\mu_0}, \ket{\mu_1}}$ basis, $f_{C,\kb{\mu_1}}(\x)$ is precisely $f(\x)$ and $f_{C,\kb{\mu_0}}(\x)$ is $\neg f(\x)$. However, we don't have an analogue of this for intermediate states of the circuit.
In general, it is possible for both $f_{C,\Pi}(\x) = 1$ and $f_{C,(I-\Pi)}(\x) = 1$, and thus $f_{C,\Pi}(\x)$ does not always provide useful information. Our key observations, enabling us to use these activation functions effectively, are as follows. 

At the start of the circuit, the entire input and ancillae state lies inside the subspace $\S_0 = \spn\lr{ \ket{\veta}_A } \tens \mathcal{H}_{[n]}$, where $\mathcal{H}_{[n]}$ is the Hilbert space of our input qubits and $\ket{\veta}_A$ is the ancillae starting state. Using the \QACZ normal form, originally proposed by \cite{rosenthal2021qac0}, each gate $G(S)$ of the \QACZ~circuit is a reflection about a separable state $\ket{\vth}_S$. Observe that any state in the subspace $\S_{\oth}$  that is orthogonal to $\kb{\vth}_S$ is unchanged by $G(S)$. We show that the output of shallow circuits can simulated classically by tracking the evolution of the state in only a $\poly(n)$ number of such subspaces through their activation functions.

Building on these insights, we describe an \ACZ circuit to simulate the activation function of the projector corresponding to the singular output of cleaned-up depth-$3$ \QACZ circuits.  Our classical simulation of depth-$3$ \QACZ circuits with $m$ ancillae produces a depth-$3$ $\poly(mn)$-size \ACZ circuit as stated in the theorem. 
\begin{restatable*}[Depth-3 Classical Simulation]{theorem}{acsimthm} \label{thm:d3ac0}
  Let $C$ be a single-output cleaned-up depth-$3$ \QACZ~circuit on $n$ inputs and $m$ gates that computes the function $f(\x)$ on $n$ coordinates. Then, $f(\x) \in \aczm(O(m^4 \cdot n^4),3)$ 
\end{restatable*}
This is sufficient for a lower-bound against the usual setting of \QACZ, which is limited to $\poly(n)$ gates, by applying known \ACZ lower-bounds of \cite{hastad1986switch}. We also present a stronger result \Cref{thm:depth_3_inf_anc} that rules out depth-$3$ \QACZ circuits for \PARITY on more than $O(1)$ coordinates with \emph{unlimited} number of ancillae and gates.  This is based on our observation that most of the activation functions in the lower layers of the circuit are monotone in the same direction, regardless of the number of gates, which enables us to construct a deterministic (classical) restriction simplifying most of the gates. In contrast, to simplify arbitrary \ACZ circuits, we need to make use of random restrictions, which introduce a dependence on the circuit size. This monotonicity property does not necessarily hold at higher depths, and we expect the lower bounds at higher depths to depend on the size. 

A consequence of \Cref{thm:d3ac0} is that the activation functions at lower levels of \QACZ circuits can be simplified using random restrictions on the corresponding \ACZ circuits. This provides evidence of a \QACZ analog of the {\em Switching Lemma} because the simplified activation functions are either a small junta or a CNF/DNF formula. Then, for polynomial-sized \QACZ circuits, we can potentially simplify the circuit itself using additional random restrictions.

 \subsection{Depth-2 \texorpdfstring{\QACZ}{QAC0} Circuits have \texorpdfstring{$O(\log n)$}{O(log n)} Total Influence}
 
We prove that any depth-2 \QACZ~circuit $C$ on $n$ inputs has total influence $O(\log n)$, regardless of the number of ancillae or gates. This is established by showing exponentially small Fourier tails: for any $\eps>0$, the Fourier weight above level $k = c\log(1/\eps)\log(n/\eps)$ is at most $\eps$, where $c$ is an absolute constant. Since \PARITY~has total influence $n$, this immediately implies an average-case depth-2 lower-bound against \PARITY~with unlimited ancillae, as well as any Boolean function with large total influence.%
\footnote{Furthermore, since $O(\log n)\ll n$ there exists a {\em constant} $n$ such that any depth-$2$ \QACZ circuit cannot even approximately compute \PARITY on $n$ bits.}

\paragraph{Comparison to Prior Work -- Key Challenges and New Ideas.} We compare to prior work by Nadimpalli, Parham, Vasconcelos, and Yuen~\cite{nadimpalli2024pauli}, which shows that any depth-$d$ \QACZ~circuit has small Fourier tails, but only for a restricted class of circuits with $O(n^{1/d})$ many ancillae. Our proof only holds in the more restricted depth-$2$ setting, but also in the more general unlimited ancillae setting. Nadimpalli et al. first proved their results for \QACZ~circuits with \emph{no ancillae}, and then applied a simple reduction to solve the case of $O(n^{1/d})$ many ancillae. We thus discuss the case of no ancillae as it is cleaner and easier to follow.

The total influence measures the average sensitivity of the circuit on a uniformly random input.  On a random input, the initial state is the maximally mixed state, and since the circuit is reversible, the state is also the maximally mixed state after each layer. However, on the maximally mixed state, any \CZ gate with fan-in $\omega(\log n)$ is activated with negligible probability, and can thus be replaced by the identity gate while incurring only a small error in its Fourier tail. This leaves us with a circuit composed only of gates of fan-in $O(\log n)$, meaning that the output depends only on a $O(\log n)^d$ number of input qubits, and thus the total influence is at most $O(\log n)^d$. 

This proof technique fails when we introduce ancillae since the state after each layer is no longer the maximally mixed state.
For example, consider the read-once DNF circuit for the \TRIBES function, where $\TRIBES(x) = \lor_{i=1}^{s} \land_{j=1}^{w} x_{i,j}$ for $w= \Theta(\log n)$ and $s = \Theta(n/\log n)$. Note that, without loss of generality, we can apply controlled-\OR and controlled-\AND gates (as they can be implemented using single-qubit gates before and after a \CZ gate).
Any read-once DNF with $s$ terms can be simulated by a depth-$2$ \QACZ~circuit with $s+1$ ancillae, where in the first layer, each of the first $s$ ancillae is flipped only if the corresponding term is true and in the second layer, the first $s$ ancillae qubits are fed into a controlled-\OR gate that flips the target ancilla qubit, that will contain the value of the DNF. In the case of the \TRIBES function, we get $s+1 = O(n/\log n)$ ancillae.%
\footnote{As this paragraph demonstrates, the \TRIBES function can be implemented by depth-2 \QACZ, and we see that the total influence of such circuits can be $\Omega(\log n)$ as this is the total influence of the \TRIBES function -- proving the tightness of our total influence upper bound.}
Observe though that the fan-in of the second layer is $O(n/\log n)$, and nevertheless the gate is activated with constant probability on a uniformly random input. This is because after the first layer, the ancillae are quite ``biased'' towards $\ket{0}$, unlike in the case of a maximally mixed state. 

Clearly, fan-in is not a good indicator of which gates are active with non-negligible probability in the presence of ancillae.
H\aa{}stad~\cite{hastad1986switch}
and Linial, Mansour, and Nisan~\cite{lmn1993ac0} that use restrictions and the switching lemma to obtain Fourier tails of constant depth circuits. Generalizing the classical proof technique to the quantum setting has remained elusive, as the proofs rely heavily on the discreteness of the classical circuits where bits are either $0$ or $1$ and using encoding arguments (cf. Razborov's proof of the switching lemma in Beame's survey \cite{beame1994switching}).

We therefore need a novel technique.
We observe that the \emph{entropy of the mixed state entering a gate} is a good indicator of which gates are active with non-negligible probability in the presence of ancillae, at least for the case of depth-$2$ circuits. Indeed, if the mixed state has a lot of entropy, and is also separable, then its ``min-entropy'' is large as well, which means that the state has negligible amplitude on any particular basis state. This means that the gate is activated with negligible probability. To use this approach, we need to: (i) reduce to the case of separable states entering a gate at layer $2$, (ii) bound the total influence of the circuit assuming that the mixed state entering a gate has small entropy. We explain how to do this in the following. We believe that this approach can be extended to higher depths, but the main obstacle we face is that even for depth-$3$ circuits, the state entering a gate may not be separable, in which case it is unclear how to connect the entropy and min-entropy measures.

\paragraph{Proof Overview.} The proof proceeds in three main stages. First, we simplify the circuit structure by removing layer-1 gates that depend on too many input qubits. Specifically, any \CZ~gate at layer 1 depending on more than $b = \Theta(\log(n/\eps))$ input qubits can be replaced with the identity gate, incurring only $\eps/2$ error in the Fourier tail. This is because such gates are rarely active on uniformly random inputs, and their removal changes the circuit's behavior by at most $O(2^{-b})$ in $\ell_2$ distance.

Second, we apply a random-valued restriction tailored to the circuit structure. For each remaining layer-1 gate that depends on between $1$ and $b$ input qubits, we randomly keep exactly one of its input qubits alive and fix the rest uniformly at random. Input qubits not involved in any gate remain alive. This restriction keeps alive at least $n/b$ variables and reduces the Fourier tail analysis to structured circuits where each layer-1 gate depends on at most one input qubit (but potentially many ancillae).

Third, we analyze the Fourier tails of these structured circuits. Since the computation is single-output, we focus on the single layer-2 gate $g$ containing the target qubit.
Allowing for the ancilla starting state to be an arbitrary product state, WLOG, $g$ is a \CZ~gate that flips the phase when its input is $\ket{1^m}$. We partition the input qubits of $g$ into disjoint sets $Q_0, Q_1, \ldots, Q_n$, where $Q_i$ contains qubits from the layer-1 gate involving input $x_i$ (if any), and $Q_0$ contains qubits from ancilla-only gates. On input $x$, the mixed state entering $g$ is $\rho^x = \rho_0 \otimes \rho_1^{x_1} \otimes \cdots \otimes \rho_n^{x_n}$, where $\rho_i^b$ is the state of $Q_i$ when $x_i = b$, i.e., $\rho^x$ is a highly separable mixed state. 

The analysis splits into two cases based on the activation probability of gate $g$. Let $\rho = \E_x[\rho^x]$ be the average state over all inputs. 

\textbf{Case 1:} If $\bra{1^m}\rho\ket{1^m} \le \eps/32$, then gate $g$ is almost always inactive. In this case, replacing $g$ with identity changes the circuit's acceptance probability function by at most $\eps/8$ in $\ell_2$ distance. The resulting circuit has only one layer of gates and computes a dictator function (depends on a single input), which has zero Fourier weight above level 1. This implies that the original circuit has at most $\eps/4$ Fourier weight above level $\Theta(\log(1/\eps))$.

\textbf{Case 2:} If $\bra{1^m}\rho\ket{1^m} \ge \eps/32$, i.e., the gate is activated with non-negligible probability, then most $\rho_i^{0}, \rho_i^{1}$ are close to the the all $1$s state. This implies that the trace distance between $\rho_i^0$ and $\rho_i^1$, $\td{\rho_i^0}{\rho_i^1}$ is small, which implies that the influence in direction $i$ is small and overall the total influence is small.

To get the exponential Fourier tail bound, we express the mixed state $\rho$ as a ``matrix Fourier decomposition'': First we express $\rho_i^{x_i} = \rho_i + (-1)^{x_i} \cdot D_i$, where $D_i = (\rho_i^0 - \rho_i^1)/2$ is the derivative with respect to $x_i$ and $\rho_i = \frac{1}{2}(\rho_i^0 + \rho_i^1)$ is the average state. Then, we write the Fourier decomposition of the mixed state $\rho^x$ as 
$$\rho^x = \sum_{R \subseteq [n]} \widehat{\rho}(R) \cdot (-1)^{\sum_{i\in R} x_i}$$
where each $\widehat{\rho}(R)$ is expressed as a product of the $\rho_i$'s (for $i\not\in R$) and $D_i$'s (for $i\in R$).
This decomposition is similar to a Fourier decomposition of Boolean functions except that the coefficients $\hat{\rho}(R)$ are density matrices. 
Then, we relate the Fourier coefficients of the acceptance probability function $f_C(x)$ to those of $\rho^x$. 
We show that for any set $R\subseteq [n]$, we have $|\widehat{f_C}(R)| \le \|\widehat{\rho}(R)\|_1 = \prod_{i\in R} \|D_i\|_1 = \prod_{i\in R} \td{\rho_i^0,\rho_i^1}$. This implies that the Fourier weight at level $\ell$ is at most $(64\ln(8/\eps)/\ell)^\ell$ using Maclaurin's inequality. Setting $\ell = \Theta(\log(1/\eps))$ ensures that $\W^{\ge \ell}[f_C] \le \eps/4$. Combining this with the errors from the simplification and restriction steps, we obtain the desired Fourier tail bound, which implies the $O(\log n)$ total influence bound.

\subsection{Depth-2 \texorpdfstring{\QACZ}{QAC0} Circuits Cannot Prepare Generalized Nekomatas}
The proof techniques so far have focused on lower-bounds for computation of the \PARITY~function, which ultimately maps a variable input state into a single designated output target qubit. We will now describe our lower-bound techniques for preparation of a quantum state, which we refer to as a generalized nekomata, using a depth-2 \QACZ and unlimited ancillae. 

For a \QACZ circuit synthesizing a quantum state such as a nekomatas, all qubits start in the fixed $\ket{0}$ state, meaning there is no input. Thus, input-based restrictions are not possible and instead we ``restrict'' to certain subspaces by post-selecting on the output. Furthermore, rather than just a single output qubit, there are $n$ target output qubits of interest. Whereas the final layer of a \PARITY~circuit trivially depends only on the singular gate acting on the sole output target, in a nekomata circuit $O(n)$ gates can play a non-trivial role in the final layer by acting on all $n$ output qubits and ancillae. Therefore, depth-2 lower bounds against nekomatas appear more challenging than those against \PARITY.

\paragraph{Proof Sketch.} 
Although the state synthesis setting is fundamentally different from the input-output setting, this result follows the same core concepts behind our depth-$3$ bound. The proof consists of two parts.

\textbf{Simplify final layer:} Given a depth-$d$ circuit that outputs a nekomata, we construct a separable state on some subset of qubits, $\ket{\veta}_Q$, such that, inside the eigenspace of $\kb{\veta}_Q$, (1) the final layer is simplified 
(2) the state is still entangled across at least $n/2$ of the targets, and thus still a nekomata. 

\textbf{Structure after a single gate:} Then, we argue that such a state cannot be created by a depth-$2$ circuit by reasoning about the entanglement structure of states created by depth-$1$ circuits. Since all the qubits start in the state $\ket{\0}$, qubits across different gates are unentangled, and we only need to consider a single gate. We observe that the state after a single (non-trivial) gate, $(I - 2\kb{\vth}) \cdot \ket{\0}$, lies in a rank $2$  subspace $\spn\{\ket{\vth}, \ket{\0}\}$ consisting of two non-orthogonal separable states. Additionally, for any partitioning of the qubits into targets $T$ and ancillae $A$, it is the case that the state on $T$ lies in $\spn\{\ket{\vth}_T, \ket{\0}_T\}$, even when the state on $A$ is \emph{restricted} to lie inside some subspace of $\mathcal{H}_A$. The defining property of a generalized $n$-nekomata is that it has an $n$-partite rank $2$ Schmidt decomposition. This is not possible for a state in $\spn\{\ket{\vth}_T, \ket{\0}\}$ unless $|T| = n = 2$ because $\braket{\vth_T | \0} \neq 0$. 

\subsection{Exact to approximate reduction for \PARITY}\label{sec:reduction_tech}
Recently \cite{grier2026mathsfqac0containsmathsftc0with} showed that the previous approximate nekomata constructions, such as \cite{rosenthal2021qac0} can be made exact using the amplitude amplification of \cite{Grover} \cite{Brassard2000QuantumAA}, thus giving a $\poly(n)$ size \QACZ circuit for exact \PARITY/\FANOUT/\MAJORITY on $\polylog(n)$ qubits. We point out this allows us to eliminate the error in \emph{any} \QACZ circuit that computes parity with at least $1/\polylog(n)$ advantage on average, in the same asymptotic depth.  

\begin{restatable}[Approximate-to-Exact Parity in \QACZ]{theorem}{approxtoexpar} \label{thm:approxtoexactpar}
Let $C$ be a depth-$d$ \QACZ circuit with $n$ input qubits and $a$ ancillae.
Consider the function $f_C: \B^n \to [0,1]$ defined by $f_C(\x) = \Pr[\text{$C$ accepts $\x$}]$. Let $\rho \in [-1,1]$ be the correlation of $f_C(\x)$ with $\PARITY(\x)$, such that $\rho = 1/(\log n)^\delta$. Then, there is a depth-$O(d)$ circuit $C'$ using $a\cdot n+n^{O(\delta)}$ ancillae that \emph{exactly} computes $\PARITY(\x)$ on \emph{every} $\x \in \bin^n$. 
\end{restatable}

The main approximate reduction from nekomata to \PARITY, \cite{rosenthal2021qac0}, requires the \PARITY circuit to have error on \emph{most} inputs rather than on average \cite{rosenthal2021qac0}. However, this can be easily remedied by using the ``poor man's cat state'' that can be constructed in \QNCZ \cite{watts2019separation} to map an input $\x$ to a random input of the same parity as $\x$. Formally, this gives us the following. 

\begin{restatable}[Average-to-Worst Case \PARITY ]{claim}{avgtoworst}\label{lem:avgtoworst}
Let $C$ be a depth-$d$ \QACZ circuit with $n$ input qubits and $a$ ancillae and define $f_C(\x)$ to be  $\Pr[\text{$C$ accepts $\x$}]$. Let $\gamma \in [-1,1]$ be the correlation between $f_C(\x)$ and $\PARITY(\x)$. Then, there is a depth-$d' = d + O(1)$ circuit $C'$ with $a' = a + O(n)$ ancillae such that, for \emph{every} input $\x \in \bin^{n}$, the output register of $C'(\x)$ measures to $\ket{\oplus_{\x}}$ with probability $\geq 1/2 + |\gamma|/2$.
\end{restatable}

Then, using \cite{rosenthal2021qac0, grier2026mathsfqac0containsmathsftc0with} to create $\polylog(n)$ many copies of the input register, we can obtain a circuit $C$ with an arbitrarily low constant error. Finally, we use some standard ancilla tricks and then reflect about the state $C \ket{+}^{\tens n}$ to obtain a nekomata.

\begin{restatable}[Approximate-to-Exact nekomata]{claim}{aprexactcat}\label{lem:exactcat}
  Let $C$ be a depth-$d$ \QACZ circuit that constructs a state on $n$ targets $T$ using $a$ ancillae, given by $\ket{\psi}_{T,A} = C \ket{0^{a}}$. Suppose that the qubits $T$ have $\geq 1/4$ probability of measuring to $1^n$ and to $0^n$ each, then, there exists a circuit $C'$ of depth-$d' \leq 3(d+2)$ with $a' = a+1$ ancillae that exactly synthesizes an $n$-nekomata, i.e. 
$$C' \cdot \ket{0^{a'}}  = \frac{1}{\sqrt{2}} \ket{0^n} \ket{\varphi_0} + \frac{1}{\sqrt{2}} \ket{1^n} \ket{\varphi_1}$$
for some states $\ket{\varphi_0}, \ket{\varphi_1}$.
\end{restatable}

We provide the formal proof of \Cref{thm:approxtoexactpar}, which is a simple consequence of the aforementioned works in \Cref{sec:aproofs}. The proofs techniques we use in \Cref{sec:d3poly} and \Cref{sec:generalized_neko} rely on the circuit computation being \emph{exact}. Nevertheless, extending these techniques to a higher depths would be sufficient to rule out approximate circuits for \PARITY, \MAJORITY and all related problems in \QACZ. 
\section{Preliminaries}\label{sec:prelim}

A depth $d$ $\qac$ circuit acting on input qubits  $[n]$ and ancilla  qubits $A$ consists of $d$ layers of multi-qubit Toffoli gates interleaved with layers of arbitrary single-qubit unitaries. The single-qubit unitaries are ``free'' and do not contribute to the depth. 
Each qubit in the circuit can appear in at most one gate per layer, and therefore the layers with multi-qubit gates are associated with a partitioning of the qubits. The inputs start in the standard basis state $\ket{\x}$ for $\x \in \bin^n$ and the ancillae start in the all $\ket{\0}$ state. 

\begin{figure}[ht]
    \centering
    \includegraphics[width=0.6\linewidth]{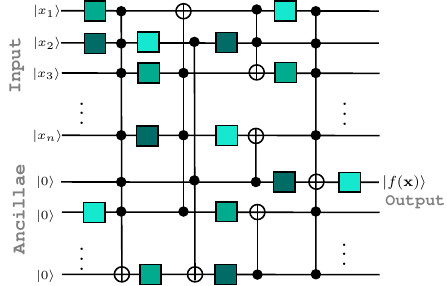}
     \caption{\QACZ circuit computing $f(\x)$. The circuit takes in a $n$-qubit classical input corresponding to the string $\x \in \bin^n$. Each ancilla starts in a fixed state $\ket{0}$. At the end of the circuit, the state on a designated output register contains the answer $\ket{f(\x)}$. Each qubit can belong to at most one gate per layer and there are no locality based constraints on multi-qubit gates, thus, the illustrated circuit has depth $4$.}
    \label{fig:qac0ckt}
\end{figure}

There is an equivalent characterization of \QACZ due to \cite{rosenthal2021qac0} where all the single qubit unitaries are moved to a new layer at the end. The intermediate layers consist only of multi-qubit gates that reflect about product states and there are $d$ such layers. A gate $G$ acting on the subset of qubits $S$ has the form,
\begin{align}
  G(S) = I_S - 2\kb{\vth}_S,
\end{align}
where $\ket{\vth}_S$ is a separable state across all of $S$. Note that this gate is simply a reflection about the state $\ket{\vth}_S$ and is a unitary with eigenvalues $\pm 1$. 
We use the vector labels such as $\ket{\vth}_S$ or  $\ket{\veta}_S$ to denote a product state on $S$, whose component on any qubit $q \in S$ is denoted by $\ket{\th_q}$ or $\ket{\eta_q}$ respectively.  

A depth $d$ \QACZ circuit in this normal form therefore consists of $d$ such layers and we assume that there are no extra single qubit unitaries in between. Additionally, we will consider single-output circuits, in which we allow the ``output'' to be in any basis and therefore do not need to consider the additional layer of single-qubit unitaries at the end of the computation. We define the output of such circuits as below.
\qacoutput

\subsection{Additional Notation} 
We will use the shorthand notation $C(\x)$ to refer to the state $C \ket{\x} \ket{\veta}_A$, the final state of circuit $C$ on input $\x$, when the ancilla starting state $\ket{\veta}_A$ is clear from the context. 
For a classical restriction $R$, $C_{|R}$ denotes the depth $\leq d$ \QAC circuit on $n-|R|$ inputs given by fixing the inputs in $R$ and treating them as ancillae. 

For a distribution $D$ we denote by $x\sim D$ a sample from the distribution. Let $X$ be a finite set. We denote by $x \sim X$ a uniformly random sample from $X$. 

We use the following notation while referring to the complexity of classical functions.
\begin{definition}[\ACZn$(s,d)$]
\ACZn$(s,d)$ refers to the set of all Boolean functions $f : \bin^n \to \bin$ computable by classical \ACZ circuits of depth $d$ using at most $s$ gates. 
\end{definition}

\subsection{Analysis of Boolean Functions}
\label{sec:prelim_aobf}
For every function $f: \B^n \to \R$ there exists a unique Fourier representation 
\[
f(x) = \sum_{S\subseteq [n]}{\hat{f}(S)}\cdot (-1)^{\sum_{i \in S} x_i}
\] where $\hat{f}(S) \in \R$ are called the Fourier coefficients of $f$.
These coefficients satisfy 
$$\hat{f}(S) = \E_{x\sim \{0,1\}^n}[f(x) \cdot \chi_S(x)]$$
for $S\subseteq [n]$, where $\chi_S(x) = (-1)^{\sum_{i \in S}x_i}$ are the Parity functions. 
Indeed, one can see that the existence and uniqueness of the Fourier representation follows from the fact that the $2^n$ parity functions form an orthonormal basis to the space of all functions from $\B^n$ to $\R$ equipped with inner product 
\[\langle{f, g\rangle} = \E_{x\sim \B^n}[f(x) \cdot g(x)].\]
Parseval's identity implies that 
$\E_{x\sim \B^n}[f(x)^2] = \sum_{S\subseteq[n]}\hat{f}(S)^2$
and both sides equal $1$ if $f$ is a Boolean function, i.e., $f:\B^n \to \{\pm1\}$.

We define the total influence of $f$, $\Inf[f]$,  as \[\Inf[f] = \sum_{S\subseteq [n]} \hat{f}(S)^2 \cdot |S|.\]
(In the special case where $f$ is a Boolean function, note that the Fourier coefficients squared of a Boolean function $f$ naturally define a probability distribution over the sets $\{S: S\subseteq[n]\}$ where set $S$ is chosen with probability $\hat{f}(S)^2$. Then, the total influence is the expected size of $S$ under this distribution.)
The total influence also equals the combinatorial quantity $\sum_{i=1}^n \Inf_i[f]$, where $\Inf_i[f]$ is the influence of the $i$-th coordinate on $f$ defined as
\[\Inf_i[f] = \E_{x\sim \B^n}\left[\left(\tfrac{|f(x)-f(x^{\oplus i})|}{2}\right)^2\right] \text{ where  } x^{\oplus i} = (x_1, \ldots, x_{i-1}, \overline{x_i}, x_{i+1}, \ldots, x_n).\]
Note that if $f$ is Boolean, then $\tfrac{|f(x)-f(x^{\oplus i})|}{2}\in \{0,1\}$, and its value indicates whether changing the $i$-th coordinate in $x$ changes the value of $f$.

We say that a coordinate $i$ is \emph{influential} on $f$ if $\Inf_i[f] > 0$. We say that a function is a {\em $k$-junta} if at most $k$ of its coordinates are influential (i.e., if $f$'s value depends on at most $k$ coordinates). In other words, $f$ is a $k$-junta if it can be written as a function $g:\{0,1\}^k \to \R$ applied to some subset of $k$ coordinates $i_1<i_2<\dots<i_k$ as follows: $f(x) = g(x_{i_1}, x_{i_2}, \dots, x_{i_k})$.

We define the Fourier weight at level $k$, $\W^k[f]$, and the Fourier tail at level $k$, $\W^{\ge k}[f]$, as 
\[\W^{k}[f] = \sum_{S:|S|=k} \hat{f}(S)^2,\qquad \W^{\ge k}[f] = \sum_{S:|S| \ge k} \hat{f}(S)^2\]

A {\em restriction} is a partial assignment to the variables of a Boolean function. We denote it by a pair $(J, z)$ where $J \subseteq[n]$ is the set of coordinates that stay alive and $z\in \{0,1\}^{[n]\setminus J}$ is an assignment to the rest. Given a function $f: \B^n \to \R$ and a restriction $(J,z)$, we naturally get the {\em restricted function} $f|_{J,z}:\B^n \to \R$ defined by $f|_{J,z}(x) = f(y)$ where for $i\in [n]$, $y_i = x_i$ if $i\in J$ and $y_i=z_i$ otherwise.

A random restriction is a distribution over restrictions. A {\em random valued restriction} is a distribution over restrictions of a special form -- we first pick $J\subseteq[n]$ from an arbitrary distribution and then pick $z\in \B^{[n]\setminus J}$ uniformly at random.
The expected Fourier coefficients under random valued restrictions are well understood.

\begin{lemma}[Fourier weights under random valued restrictions]\label{lemma:random_valued_restriction}
Let $f: \B^n \to \R$. Let $J$ be a random subset of $[n]$ under some arbitrary distribution $D$ and 	let $z \sim \{0,1\}^{[n]\setminus J }$.
Then, for any set $S\subseteq [n]$, 
\begin{align}
    \E_{J,z}[\widehat{f|_{J,z}}(S)^2] = \sum_{T\subseteq [n]} \widehat{f}(T)^2 \cdot \Pr_{J}[T\cap J = S].
\end{align}
In particular, for any $k\in \N$, 
\begin{align}
    \E_{J,z}[\W^{\ge k}[f|_{J,z}]] = \sum_{T\subseteq [n]} \widehat{f}(T)^2 \cdot \Pr_J[|T\cap J|\ge k].
\end{align}
\end{lemma}
The proof was essentially given in \cite[Prop.~4.17]{o2014analysis}.

\subsection{Trace Distance and its Properties}
For two mixed states (density matrices) $\rho,\sigma$ on the same Hilbert space, the \emph{trace distance} is defined as
\[
\td{\rho}{\sigma} \;:=\; \tfrac{1}{2}\|\rho-\sigma\|_1
\qquad\text{where}\quad
\|X\|_1 := \mathrm{Tr}\sqrt{X^\dagger X}.
\]
This metric has several important properties:

\begin{itemize}
  \item \textbf{Range and equality:} $0 \le \dtd(\rho,\sigma) \le 1$, and $\dtd(\rho,\sigma)=0$ iff $\rho=\sigma$.
  \item \textbf{Unitary invariance:} $\dtd(U\rho U^\dagger, U\sigma U^\dagger)=\dtd(\rho,\sigma)$ for any unitary $U$.
 \item \textbf{Tensoring:} $\dtd(\rho\otimes\tau, \sigma\otimes\tau)=\dtd(\rho,\sigma)$ for any fixed $\tau$.
  \item \textbf{Triangle inequality:} $\dtd(\rho,\tau)\le \dtd(\rho,\sigma)+\dtd(\sigma,\tau)$.

\item \textbf{Contractivity under channels and measurements:}
Let $\Phi$ be any quantum channel. Then for all density matrices $\rho,\sigma$,
\[
\dtd \bigl(\Phi(\rho),\,\Phi(\sigma)\bigr)\;\le\; \dtd(\rho,\sigma).
\]
In particular, if $\dtd(\rho,\sigma)\le \varepsilon$, then the outputs $\Phi(\rho)$ and $\Phi(\sigma)$ are also $\varepsilon$-close in trace distance.
Hence no quantum operation or measurement can increase the statistical distance beyond~$\varepsilon$.
\end{itemize}

\subsection{Activation Functions}\label{sec:actv}
We will refer to \emph{activation functions} of projectors on intermediate states of our circuit, described by the corresponding sub-circuit. These are defined below. 
\begin{definition}[Projector Activation Function]
Given a \QACZ circuit $C$ on $n$ inputs and any projector $\Pi_S$ acting on a subset of qubits $S$, the \emph{activation function of $\Pi_S$ on $C$} is a classical Boolean function, $f_{C,\Pi_S} : \bin^n \to \bin$ defined as.  
$$f_{C,\Pi_S}(\x) := \blr{\Pi_S \cdot C(\x) \neq 0}$$
where the notation $\blr{\cdot}$ to refers to an indicator function. 
\end{definition}
We note that it is possible for both $f_{C,\Pi_S}(\x) = 1$ and $f_{C,(I - \Pi_S)}(\x) = 1$, but it cannot be that both of them are $0$. For a projector $\Pi$, we will use $\eig{\Pi}$ to refer to the $+1$ eigenspace of $\Pi$. Then, the activation function is essentially checking if $C(\x)$ has a component in $\eig{\Pi}$. 

\subsubsection{Projectors and their Eigenspaces}
We give some useful lemmata for manipulating projectors.
\begin{lemma}[Projector Decomposition]\label{lem:projdecomp}
  Let $\Pi$ be a projector on a Hilbert space $\mathcal{H}$ that can be decomposed as $\Pi = \lr{\bigotimes_{i =0}^{n} \Pi^{(i)}_{B_i}}$
  where $B_i$ is a set of qubits.
  Then \begin{enumerate}
    \item $\eig{\Pi} = \bigotimes_{i =0}^{n} \eig{\Pi^{(i)}_{B_i}}$.
    \item $\eig{I-\Pi} = \spn\{\bigcup_{i =0}^{n} \eig{I-\Pi^{(i)}_{B_i}}\}$.\end{enumerate}
    \end{lemma}
    \begin{proof}
      \begin{enumerate}
        \item The first item follows as any +1 eigenvector of $\Pi$ is the tensor product of +1 eigenvectors for $\Pi^{(i)}_{B_i}$.
        \item  All these projectors can be simultaneously diagonalized because they form a set of commuting observables, i.e,  
        $[(I - \Pi^{(i)}_{B_i}), (I - \Pi)] = 0$ and $[(I - \Pi^{(i)}_{B_i}), (I - \Pi^j_{B_j})] = 0$ for all $i,j \in [n]$. Take any eigenvector $\ket{v}$ from this diagonalization. If $\ket{v}$ is a $+1$  eigenvector of $(I - \Pi^{(i)}_{B_i})$ for \emph{some} $i$, then it is a $0$ eigenvector of $\Pi^{(i)}_{B_i}$ and $\Pi$, and thus a $+1$ eigenvector of $(I -\Pi)$. If $\ket{v}$ is a $0$ eigenvector of \emph{every} $(I - \Pi^{(i)}_{B_i})$, then from (1), it must be a $+1$ eigenvector of $\Pi$, and thus a $0$ eigenvector of $(I - \Pi)$. 
    \end{enumerate}
    \end{proof}
    
\begin{lemma}\label{lem:componentinside}
    Suppose $\eig{\Pi} = \spn\{\bigcup_{i\in I} \eig{\Pi^{(i)}\}})$ for some set of projectors $\Pi^{(i)}$. Then,
    any state $\ket{\psi}$ that has a component inside $\eig{\Pi}$ must have a component inside one of the subspaces $\eig{\Pi^{(i)}}$, and vice versa, 
    i.e., $\Pi\ket{\psi} \neq 0$ if and only if $\exists i\in I$ 
    such that $\Pi^{(i)}\ket{\psi} \neq 0$.
    \end{lemma}
    \begin{proof}
      In one direction, suppose that $\Pi^{(i)}\ket{\psi} \neq 0$ for some $i$.
      Then, $\ket{\psi}$ has a non-zero inner product with a +1 eigenvector of $\Pi^{(i)}$ which is also an +1 eigenvector of $\Pi$, so we get that $\Pi\ket{\psi}\neq 0$.

    In the other direction, if $\Pi^{(i)} \ket{\psi} = 0$ for all $i$, then $\ket{\psi}$ is orthogonal to $\eig{\Pi^{(i)}}$ for all $i$ and it must be orthogonal to their span, which is $\eig{\Pi}$, thus $\Pi \ket{\psi} = 0$. 
    \end{proof}

\subsection{Nekomata States}\label{sec:nekodef}
\begin{definition}[Generalized $n$-nekomata]\label{def:neko}
 A state $\ket{\psi}$ is said to be a \emph{generalized $n$-nekomata} if, for some $\alpha \neq 0, \beta \neq 0$, it has the form,  
$$\ket{\psi} = \alpha \cdot \ket{\mu_1}_{t_1} \ket{\mu_2}_{t_2} \dots \ket{\mu_n}_{t_n} \ket{\gamma_0}_A + \beta \cdot \ket{\mu^\perp_1}_{t_1} \ket{\mu^\perp_2}_{t_2} \dots \ket{\mu^\perp_n}_{t_n} \ket{\gamma_1}_A$$
where $\braket{\mu_i | \mu^\perp_i} = 0$ for $i \in [n]$. 
Here $A$ is the set of remaining qubits referred to as the \emph{ancillae} of $\ket{\psi}$ and the qubits $t_1,t_2 \dots t_n$ are as the \emph{targets} of $\ket{\psi}$. 
\end{definition}
A nekomata, as defined in the literature, has $\alpha = \beta = 1/\sqrt{2}$ and $\ket{\mu_i} = \ket{0}$, and is included in the above definition.
\begin{definition}[Generalized nekomata under separable post-selection (GNSP)] \label{def:nekopost}
A state $\ket{\varphi}$ is said to be a \emph{generalized $n$-nekomata under separable post-selection}, if for some (possibly empty) subset of qubits $Q$ and a separable state $\ket{\veta}_Q$, $\kb{\veta}_Q \cdot \ket{\varphi}$ is nonzero and a (un-normalized) $n$-nekomata. In other words, for some set of $n$ target qubits $T$, 
$$\kb{\veta}_Q \cdot \ket{\varphi} =  \alpha \cdot \ket{\mu_1}_{t_1} \ket{\mu_2}_{t_2} \dots \ket{\mu_n}_{t_n} \ket{\gamma_0}_A + \beta \cdot \ket{\mu^\perp_1}_{t_1} \ket{\mu^\perp_2}_{t_2} \dots \ket{\mu^\perp_n}_{t_n} \ket{\gamma_1}_A$$
where $\alpha \neq 0, \beta \neq 0$ and $\braket{\mu_i|\mu^\perp_i} = 0$.
\end{definition}
Then any generalized $n$-nekomata is also a \emph{$k$-GNSP} for any $k \leq n$.

\renewcommand{\P}{\mathcal{P}}
\section{Depth-3 \texorpdfstring{\QACZ}{QAC0} Circuits Cannot Compute \PARITY or \MAJORITY}\label{sec:d3poly}


In this section, we will prove that depth-$3$ \QACZ circuits cannot compute \PARITY or \MAJORITY. We will first show a reduction from circuits computing \PARITY or \MAJORITY to circuits whose first layer is ``cleaned up'' in the sense that each gate depends on at most one input qubit. We will then show that any cleaned-up depth-$3$ \QACZ circuit cannot compute \PARITY or \MAJORITY. To do that, we will show that any cleaned-up depth-$3$ \QACZ circuit can be simulated classically in \ACZ.

\subsection{The Clean-Up Lemma}
We begin with the proof of the clean-up lemma which is restated below.
\cleanup
\begin{proof}
 Note that \PARITY is invariant under any classical restriction and \MAJORITY is invariant under any \emph{balanced classical restriction} with an equal number of $0$s and $1$s among the fixed coordinates.
 Our goal is to convert the layer-$1$ gates of the circuit into ``cleaned up'' gates that depend on at most one input qubit.
 We handle all the other gates as follows:
\paragraph{Handling Gates with exactly 2 input qubits.}
Suppose that there are $m$ such gates, with $2m$ total coordinates. Apply a balanced classical restriction $R$ on $2 \lceil m/2 \rceil \leq m + 1$ coordinates, by fixing one coordinate from each gate, and potentially one extra coordinate in case $m$ is odd (to ensure we fix an even number of coordinates and maintain a balanced restriction). We thus restrict at most $(m + 1)/2m \le 2/3$ fraction of the coordinates. The balanced restriction ensures that the output of the resulting circuit is correct.
\paragraph{Handling Gates with 3 or more input qubits.} Let $G(S)$ be any such gate. We describe our quantum analogue of a restriction that turns two input qubits from $S \cap [n]$ to ancillae and deactivates $G(S)$.
Applying such a restriction for every gate with at least $3$ input qubits deactivates all these gates, while preserving at least $1/3$ fraction of the input qubits.
Pick any two input coordinates in $i,j \in S \cap [n]$, say $(i,j) =(1,2)$ and convert them to ancilla by hard-wiring their state to $\ket{\varphi}_{1,2}$ satisfying (i)
$\ket{\varphi}_{1,2} := \alpha \cdot \ket{01} + \beta\cdot  \ket{10}$ and (ii) $\braket{\vth_{1,2}|\varphi}=0$.
Such a state exists because the dimension of the space of states spanned by $\ket{01}$ and $\ket{10}$ is $2$ and $\braket{\vth_{1,2}|\varphi}=0$ is a linear constraint, so there must exist a non-trivial solution.
More explicitly, we choose $\alpha =\frac{a}{\sqrt{|a|^2 + |b|^2}}$ and $\beta = \frac{-b}{\sqrt{|a|^2 + |b|^2}}$ for $a = \braket{\vth_{1,2}|10}$ and $b = \braket{\vth_{1,2}|01}$ if $(a,b)\neq (0,0)$ and otherwise choose $\alpha = 1$ and $\beta = 0$.


This restriction deactivates $G$ because $\braket{\vth_{1,2}|\varphi} = 0$.
We will first argue that the output of the circuit under this restriction still computes PARITY or MAJORITY (resp.) and then reason about the resulting circuit. Let $\ket{\veta}_A$ denote the original ancilla state, so that, on input $\x \in \bin^{n-2}$, the resulting circuit produces the state $C \cdot \ket{\varphi}_{1,2} \ket{\x} \ket{\veta}_A$. Also, let $\clr{\ket{\mu_0}, \ket{\mu_1}}$ denote the original output basis of $C$. 

\paragraph{$C$ computes \PARITY.} 
Observe that $\ket{\phi_{1,2}}$ is a superposition over two classical restrictions setting the coordinates $(1,2)$ to either $01$ or $10$ respectively, and that both of these restrictions flips the parity of the input.
Thus, for all $\x$, $C \cdot \ket{\varphi}_{1,2} \ket{\x} \ket{\veta}$ correctly outputs \PARITY according to \Cref{def:qacoutput} in the basis $(\ket{\nu_0}, \ket{\nu_1})$ for $\ket{\nu_b} = \ket{\mu_{b\oplus 1}}$ (i.e., we flip what we consider to be $0$ and $1$ in the output basis to accomdoate the fact that we flipped the parity of the input by restricting to $\ket{\phi_{1,2}}$). 
More slowly,  for $b = \PARITY(\x)$, we have that 
\begin{align}
\bra{\mu_b} \cdot  C \ket{\varphi}_{1,2} \ket{\x} \ket{\veta}_A  
  &= \alpha \cdot \bra{\mu_{b}} \cdot C \cdot \ket{01\x} \ket{\veta}_A + \beta \cdot \bra{\mu_{b}} \cdot C \ket{10\x} \ket{\veta}_A 
  = 0.
\end{align}

\paragraph{$C$ computes \MAJORITY.}
Here, we observe that $\ket{\phi_{1,2}}$ is a superposition over two classical restrictions setting the coordinates $(1,2)$ to either $01$ or $10$ respectively, and that both of these restrictions maintains the majority of the input.
Thus, for all $\x$, $C \cdot \ket{\varphi}_{1,2} \ket{\x} \ket{\veta}$ correctly outputs \MAJORITY according to the same basis $(\ket{\mu_0}, \ket{\mu_1})$ as before.
More slowly,  for $b = \neg \MAJORITY(\x)$, we have that 
\begin{align}
  \bra{\mu_b} \cdot  C \ket{\varphi}_{1,2} \ket{\x} \ket{\veta}_A  
  &= \alpha \cdot \bra{\mu_{b}} \cdot C \cdot \ket{01\x} \ket{\veta}_A + \beta \cdot \bra{\mu_{b}} \cdot C \ket{10\x} \ket{\veta}_A 
  = 0.
\end{align}


So far, we have shown that the circuit's output under this restriction still computes \PARITY or \MAJORITY. However, we need to implement the restriction in the standard model with separable initial ancilla state without increasing the depth of the circuit, as we explain next.

\paragraph{Converting back to separable ancillae.}
Let $D$ be the depth-$d$ sub-circuit consisting of gates other than $G(S)$, so that,
$C = D \cdot G(S)$.
The new linear operator on the inputs is given by,  
\begin{align}
 C \ket{\varphi}_{1,2} \ket{\x} \ket{\veta}_A &=  D \cdot G(S) \ket{\varphi}_{1,2} \ket{\x} \ket{\veta}_A \\
 &= D \cdot (I - \underbrace{2\kb{\vth}) \cdot \ket{\varphi}_{1,2}}_{ = 0} \ket{\x} \ket{\veta} \\
 &= D \cdot \ket{\varphi}_{1,2} \ket{\x} \ket{\veta}
\end{align}
Here the ancilla starting state is not necessarily separable because $\ket{\varphi}_{1,2}$ could be entangled. To fix this, we add a new gate $G'(t_1,t_2)$ in place of $G$ to compute $\ket{\varphi}_{1,2}$ from a separable state. One such gate is given by $G' = \cnot(1,2)$ with the starting state $\ket{\mu}_{1} \ket{1}_{2}$, where $\ket{\mu}_{1} = \alpha\ket{0} +\beta\ket{1}$. Then,
\begin{align}
  G' \cdot \ket{\mu}_1 \ket{1}_2 &= \cnot(1,2) \cdot \lr{\alpha \ket{01} +\beta \ket{11}} \\
                                 &= \alpha\ket{01} + \beta \ket{10} 
\end{align}
Since the first layer in $D$ doesn't apply any gate on $(1,2)$, 
the new circuit $C'' = D \cdot G'$ is still a depth-$\leq d$ circuit on the other input coordinates with a separable initial ancilla state $\ket{\veta} \ket{\mu}_1 \ket{1}_2$.

So far, we explained how to clean up one gate with at least 3 input qubits.
Repeating this transformation for all layer-$1$ gates containing $\geq 3$ coordinates in the original circuit preserves at least $1/3$ fraction of these coordinates. 

\begin{figure}[htbp]
  \centering
  \includegraphics[width=0.9\linewidth]{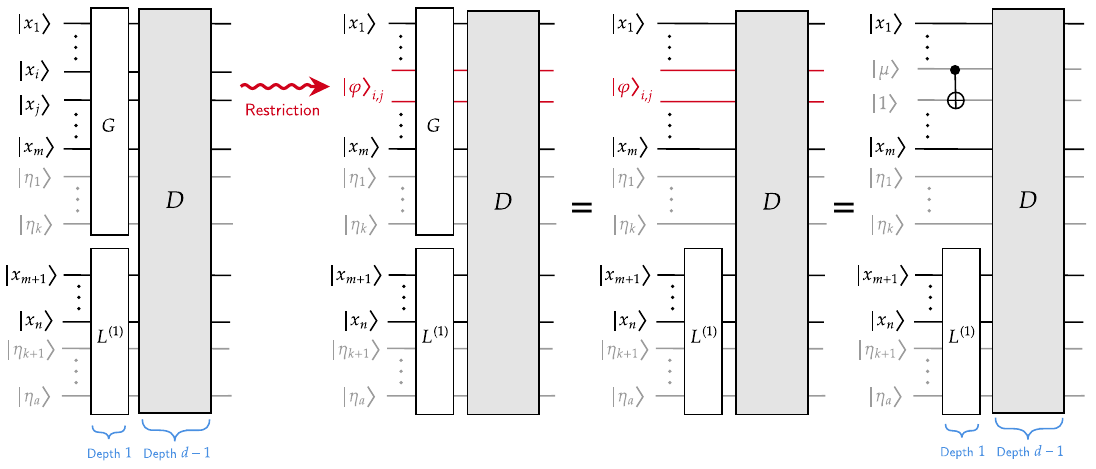}
  \caption{Clean-up of a gate $G$ with more than 3 input qubits. The restriction fixes the state of $\ket{x_i}\ket{x_j}$ to $\ket{\varphi}_{i,j}$ in a way that  deactivates the gate $G$. Then, the restriction to $\ket{\varphi}_{i,j}$ is equivalent to applying a CNOT gate on the separable initial ancillae state $\ket{\mu}\ket{1}$.}
  \label{fig:placeholder}
\end{figure}

The classical and quantum restrictions together fixes at most $2n/3$ input qubits and converts them to ancillae, such that the resulting circuit acts on at least $n/3$ coordinates and has at most one input qubit per layer-$1$ gate, i.e., a cleaned-up circuit.
\end{proof}

From here on, we will assume that the circuits we consider are cleaned-up. Recall that, for a \QACZ circuit $C$ and projector $\Pi$ acting on some of its qubits, the activation function $f_{C,\Pi}(\x)$, defined in \Cref{sec:actv}, checks whether $C(\x)$ has a component in the non-zero eigenspace of $\Pi$ (i.e., $\eig{\Pi}$). We will show that at each layer of the circuit (which is of depth-$3$), the relevant activation functions can be simulated in \ACZ.
We demonstrate this in the following sections.  

\subsection{\texorpdfstring{\ACZ}{AC0} Simulation of Projectors After Layer One}\label{sec:ac0d1}
After the clean-up in \Cref{lem:l1cleanup}, the state after the first layer of the circuit is separable across all input qubits, and has the following form, 
\begin{align}\label{eq:l1state}
   C(\x) &= \ket{\psi_0}_{B_0} \tens \lr{\bigotimes_{i \in [n]}  \ket{\psi_i(x_i)}_{B_i}} 
\end{align}
where $\ket{\psi_i(x_i)}_{B_i}$ is a state on qubits $B_i$ in layer-$1$ gate containing input qubit $i$, that only depends on $x_i$, and $\ket{\psi_0}_{B_0}$ is a fixed state on the remaining ancillae. 
In this section, we show that the activation functions of certain projectors have a simple form that lets us simulate them classically. These consist of two main building blocks: (i) projectors with few gates in their light-cones (``juntas''), described in \Cref{def:juntaproj} and (ii) projectors that are separable across all the $B_i$'s  from \Cref{eq:l1state}.
The first two lemmas show that the activation functions of these two building blocks can be simulated classically in \ACZ depth-$(\leq 2)$.

\begin{definition}[Junta Projector]\label{def:juntaproj}
Given a depth-$1$ \QACZ circuit $C$, a projector $\Pi_S$ acting on a subset of qubits $S$ is a \emph{$k$-junta projector} on $C$ if at most $k$ gates in $C$ act on some qubit from $S$. 
\end{definition}

\begin{lemma}[Junta Projector in \ACZ]\label{lem:d1junta}
  Let $C$ be a cleaned-up depth-$1$ \QACZ circuit on $n$ inputs and $\Pi_S$ be a \emph{$k$-junta projector} on $C$. Then, $f_{C,\Pi_S}(\x)$ is a $k$-junta. 
  In particular, $f_{C,\Pi_S}(\x)$ can be implemented by a width $k$ {\sf DNF} formula with at most $2^k$ clauses. 
\end{lemma}
\begin{proof}
    Let $C'$ be the sub-circuit that only contains the gates in the light-cone of $S$ (i.e., the gates that contain a qubit from $S$). Then, $C'$ acts on at most $k$ coordinates and $f_{C,\Pi_S}(\x) = f_{C', \Pi_S}(\x)$. This shows that $f_{C,\Pi_S}$ is a $k$-junta, and it is straightforward to see that any $k$-junta can be written as a width $k$ {\sf DNF} with at most $2^k$ clauses. (Each clause will check that the value of the $k$ junta variables equals a specific value for which the function should output $1$.)
\end{proof}

\begin{lemma}[Separable Projector in \ACZ]\label{lem:d1ac0}
  Let $C$ be a cleaned-up depth-$1$ \QACZ circuit on $n$ inputs whose ancillae start in $\ket{\0}_A$. Let $\Pi$ be projector that has the form,
    $$\Pi = \Pi^0_{B_0} \tens \lr{\bigotimes_{i =1}^{n} \Pi^{(i)}_{B_i}}$$
    where $B_i$ is the set of qubits in $G_i$, the layer-$1$ gate containing $i \in [n]$ and $B_0 \subseteq A$ is the set of remaining ancillae. 
    (note that $\Pi^{(i)}_{B_i}$ could in particular be $I_{B_i}$ as this is a special case of a projector.)

  Then, the following holds.
  \begin{enumerate}
      \item $f_{C,\Pi}(\x) \in \aczm(n,1)$,  and is either a constant function or given by,
    $$f_{C,\Pi}(\x) =  \andl_{i = 1}^n f_{C,\Pi^{(i)}_{B_i}}(x_i)$$
    
    \item $f_{C,(I - \Pi)}(\x) \in \aczm(n,1)$, and is either the constant function or given by, 
  $$f_{C,(I - \Pi)}(\x) = \orl_{i=1}^n f_{C,(I-\Pi^{(i)}_{B_i})}(x_i).$$
  \end{enumerate}
\end{lemma}

\begin{proof}

First, note that for any projector $\Pi'_{B_0}$ acting only on the ancilla $B_0$, $f_{C,\Pi'_{B_0}}(\x)$ is input independent.
From Item~(1) of \Cref{lem:projdecomp}, $\Pi \cdot C(\x) \neq 0$ iff $\Pi^{(i)}_{B_i} \cdot C(\x) \neq 0$ for every $i$. Therefore,  
  \begin{align}
    f_{C,\Pi}(\x) &= \andl_{i=0}^{n} f_{C,\Pi^{(i)}_{B_i}}(\x)
 \end{align}
 and  from  \Cref{lem:d1junta} each $f_{C,\Pi^{(i)}_{B_i}}(\x)$ for $i \in [n]$ depends only on $x_i$. Hence, this function is either $1$ or given by an \AND gate on the variables $x_1, \dots x_n$ and their negations. 
This completes the first part of the proof.

   We continue to prove the second item.
 From \Cref{lem:projdecomp}, we have that,
  \begin{align}\label{eq:dec1}
    \eig{I-\Pi} &=  \spn\clr{ \eig{(I-\Pi^{(i)}_{B_i}) \tens I_{\overline{B_i}}}}_{i \in \{0,\dots, n\}}
  \end{align}
Then, by \Cref{lem:componentinside}, any state $\ket{\psi}$ that has a component inside $\eig{I-\Pi}$ must have a component inside one of the subspaces $\eig{(I-  \Pi^{(i)}_{B_i}) \tens I_{\overline{B_i}}}$, and vice versa.
This gives,
  \begin{align}
    f_{C,(I - \Pi)}(\x) &= \orl_{i=0}^{n} f_{C,(I - \Pi^{(i)}_{B_i})}(\x)
  \end{align}
and then \Cref{lem:d1junta} shows that for $i \in [n]$, $f_{C,(I - \Pi^{(i)}_{B_i})}(\x)$ is depends only on the variable $x_i$. Hence, this function is either always $1$ or is given by an \OR gate on the variables $x_1, \ldots, x_n$ or their negations.%
\footnote{An alternative argument  goes as follows: 
$f_{C,(I - \Pi)}(\x) = 
\blr{(I-\Pi)\cdot C(\x)\neq 0}
= \neg \blr{(I-\Pi)\cdot C(\x)=0} = \neg \blr{C(\x) \in \eig{\Pi}}$.
Now, for $C(\x)$ to be a +1 eigenstate of $\Pi$ it must be a tensor product of +1 eigenstates of $\Pi_{B_i}^{(i)}$ for $i\in \{0,1, \dots, n\}$. Hence, the condition of $C(\x)\in \eig{\Pi}$ can be rewritten as an AND function on the literals $x_1, \neg x_1 \ldots, x_n, \neg x_n$, and by De Morgan's law, $f_{C,(I - \Pi)}(\x)$ can be written as an OR function on the same literals.}
\end{proof}

\begin{corollary}[Combinations of Separable Projectors \ACZn]\label{cor:d1sepcomb}
  Let $C$ be a depth-$1$ cleaned up \QACZ circuit with ancillae in $\ket{\0}_A$ and $\Pi_T$ be a projector that is given by the product of $k+1$ projectors,  
  $$\Pi_T = \Pi^0_{T_0} \tens \Pi^1_{T_1}  \tens \dots\tens \Pi^k_{T_k}$$ 
  satisfying, $\Pi^0_{T_0} = \kb{\vth}_{T_0}$ and $\Pi^j_{T_j} = I - \kb{\vth}_{T_j}$ for $j \in [k]$ for some separable states $\ket{\vth}_{T_0}, \ldots, \ket{\vth}_{T_k}$.
  Then, $f_{C,\Pi_T}(\x) \in \aczm(O(n^{k+1}),2)$ and is a $n$-DNF formula with $(n+1)^k$ clauses.
\end{corollary}
\begin{proof}
Let $B_i$ be the set of qubits in the layer $1$ gate containing input qubit $i$ and $B_0$ be the set of remaining ancillae. For $i \in \nrng, j \in [k]$, define projector $\Q^{(i,j)}$ acting on $B_i \cap T_j$ as, 
\begin{align}
    \Q\id{i,j} := \begin{cases} 
    0 \ \ \ \text{if } B_i \cap T_j = \emptyset \\
    I - \kb{\vth}_{B_i \cap T_j} \ \ \ \text{otherwise.}
    \end{cases}
\end{align}
Then, since $(I - \Pi^j_{T_j}) = \kb{\vth}_{T_j}$, we can apply \Cref{lem:projdecomp} to get,
\begin{align}\label{eq:spnone}
    \eig{\Pi^j_{T_j}} 
    &= \spn\clr{\eiglr{\Q\id{i,j} \tens I_{T_j \setminus B_i}}}_{i \in \nrng}.
\end{align}
Now, for a tuple $\z \in \nrng^k$, define the following projector,  
\begin{align}
    \P(\z) := \bigotimes_{j \in [k]} \Q\id{z_j,j}.
\end{align}
Combining with \Cref{eq:spnone} gives, 
\begin{align}
   \eig{\Pi^1_{T_1} \tens \dots \tens \Pi^k_{T_k}} &= \spn\clr{ \eig{\P(\z)} }_{\z \in \nrng^k} \\
   \Rightarrow \eig{\Pi_T} &= \spn\clr{ \eig{\Pi^0_{T_0} \tens \P(\z)} }_{\z \in \nrng^k}.
\end{align} 
Therefore, from \Cref{lem:componentinside},
\begin{align}
    f_{C,\Pi_T}(\x) &= \orl_{\z \in \nrng^k}  f_{C, \Pi^0_{T_0} \tens \P(\z)}(\x) 
\end{align}
By definition $\P(\z)$ is separable across all the $B_i$s for $i \in \nrng$ and so is $\Pi^0_{T_0} \tens \P(\z)$. Therefore, from \Cref{lem:d1ac0}, each $f_{C, \Pi^0_{T_0 \tens \P(\z)}}(\x)$ is a width $\leq n$ \AND function and taking an \OR of $(n+1)^k$ of them gives $f_{C,\Pi}(\x) \in \aczm(O(n^{k+1}),2)$. 
\end{proof}

\subsection{\texorpdfstring{\ACZ}{AC0} Simulation of Projectors After Layer Two}
Consider a depth-$2$ circuit $C = L\id{2} \cdot L\id{1}$ whose depth-$1$ sub-circuit is given by $L\id{1}$. A projector $\Pi_S$ that is a $k$ junta-projector on $L\id{1}$ may not necessarily have a small light-cone in $C$, making it challenging to simulate its activation function. In the case that $\Pi_S$ is also a $k$ junta-projector on the depth-$1$ sub-circuit formed by $L\id{2}$, we can break it up into the building blocks from \Cref{sec:ac0d1} to simulate its activation function. As it turns out, this is sufficient for the purpose of simulating the output of depth-$\leq 3$ circuits (see \Cref{sec:ac0simulation} for the proof).

\begin{lemma}[Layer-2 Small Projector]\label{lem:d2ac0small}
  Let $C = L\id{2} \cdot L\id{1}$ be a cleaned up depth-$2$ \QACZ circuit on $n$ inputs and $\Pi_T$ be a projector on a subset of qubits $T$. Suppose that $\Pi_T$ is a $k$-junta projector on both the depth-$1$ circuits $L\id{1}$ and $L\id{2}$, formed by the gates in layers $1$ and $2$ of $C$ respectively. Then, 
  $f_{C,\Pi_T}(\x)$ can be written as a DNF of size at most $O(4^k \cdot n^{k+1})$.
\end{lemma}
\begin{proof}
The depth-$1$ sub-circuit formed by $L\id{1}$ can be partitioned into $D_1 \tens D_2$, where $D_1$ is the sub-circuit formed by $\leq k$ layer-$1$ gates on $T$, and $D_2$ contains the remaining gates. 
Assume without loss of generality that $L\id{2}$ contains only the $\leq k$ gates $G_1(S_1) \dots G_k(S_k)$, in the light-cone of $T$ (as the remaining gates can be replaced with identity, not affecting the output of the circuit).

First, we partition each $S_j$ into $X_j, Y_j$ where $X_j = S_j \setminus T$ and $Y_j = S_j \cap T$ and define,  
\begin{align}
    X := \bigcup_{j \in [k]} X_j, \\
    Y := \bigcup_{j \in [k]} Y_j
\end{align}

Then, $X \cup T$ belongs to the same set of layer-$1$ and layer-$2$ gates as $T$, so we will proceed assuming without loss of generality that $X \subseteq T$ by extending $\Pi_T \to \Pi_T \tens I_{X \setminus T}$ to act on $X \cup T$. 

By the Heisenberg evolution of $\Pi_T \tens I_Y$ acting on $C$ to $\wh{\Pi}_{T,Y} = (L\id{2})^\dag \Pi L\id{2}$ acting on $L\id{1}$, 
\begin{align} 
f_{C,\Pi}(\x) = f_{L\id{1},\wh{\Pi}_{T,Y}}(\x)
\end{align}
Hence, it is sufficient to show that $f_{L\id{1},\wh{\Pi}_{T,Y}} \in \aczm$. We will do so by decomposing $\wh{\Pi}_{T,Y}$ into the building blocks from \Cref{sec:ac0d1}.

For each $G_j(S_j)$, treating $Y_j$ as the controls and $X_j$ as the targets (ref. \Cref{sec:pregates}) gives,
 \begin{align}
    G_j(S_j) &:= (I - 2\kb{\vth}_{S_j}) \\
    &= (I - \kb{\vth}_{Y_j}) \tens I_{X_j} + \kb{\vth}_{Y_j} \tens \underbrace{(I - 2\kb{\vth}_{X_j})}_{G_j(X_j)}\label{eq:d2dec}
 \end{align}
We will describe below the block-diagonalization of $\wh{\Pi}_{T,Y} = (L\id{2})^\dag \Pi L\id{2}$ into $2^k$ subspaces corresponding to the subspaces of $\mathcal{H}_{Y}$. We will then argue that, inside each subspace, $\wh{\Pi}_{T,Y}$ is separable across $T,Y$ and its component on $T$ is a $k$-junta projector, while its component on $Y$ has the form in \Cref{cor:d1sepcomb}.  
Define $\Q^{(1)}_{Y_j} = \kb{\vth}_{Y_j}$ and $\Q^{(0)}_{Y_j} = I - \Q^1_{Y_j}$ for each $j \in [k]$. Then, \Cref{eq:d2dec} becomes, 
\begin{align}
    G_j(S_j) &=  \Q^{(0)}_{Y_j}  \tens I_{X_j} +\Q^{(1)}_{Y_j} \tens G_j(X_j) 
\end{align}
and for both values of $b \in \bin$,
\begin{align}
      [\Q^{(b)}_{Y_j}, G_j(S_j)] &= 0 \tag{i.e., the commutator is $0$.} \\
      \Q^{(b)}_{Y_j} \cdot G_j(S_j) &= \Q^{(b)}_{Y_j} \tens \underbrace{(I - 2\kb{\vth}_{X_j})^b}_{G_j(X_j)^b}\label{eq:d2comm}
\end{align}
where, for a unitary $U$, $U^0 = I$ and $U^1 = U$ and $[O_1, O_2]$ denotes the commutator of two operators $O_1,O_2$. Then, we identify each subspace by a $k$-bit string $\y$. For each $\y \in \bin^k$, the projector onto this subspace, $\Q(\y)_Y \tens I_X$, is given by, 
\begin{align}
  \Q({\y})_{Y} &:= \bigotimes_{j \in [k]} \Q^{y_j}_{Y_j} 
\end{align}
Due to \Cref{eq:d2comm}, $[\Q(\y), L\id{2}] = 0$. Then, $L\id{2}$ inside this subspace is a unitary $U(\y)$ on $X$ given by, 
\begin{align}
    \Q(\y) \cdot L\id{2} &= \Q(\y) \tens \lr{\bigotimes_{j \in [k]} (I - 2\kb{\vth}_{X_j})^{y_j}} \\
    &= \Q(\y) \tens U(\y)_X
\end{align}
Putting it all together, we can block-diagonalize $\wh{\Pi}_{T,Y}$ as, 
\begin{align}
\wh{\Pi}_{T,Y} &= (L\id{2})^\dag \Pi_T L\id{2} \\
&= \bigoplus_{\y \in \bin^k} \underbrace{U(\y)^\dag  \cdot \Pi_T \cdot U(\y)}_{\P(\y)_T}\tens \Q(\y)_Y  \\ 
    &= \bigoplus_{\y \in \bin^k} \P(\y)_T \tens \Q(\y)_Y
\end{align}
where each $\P(\y)$ is a projector only on $T$. Therefore, 
\begin{align}
    f_{C,\Pi_T}(\x) &= f_{D_1 \tens D_2, \wh{\Pi}_{T,Y}}(\x) \\
                  &= \orl_{\y \in \bin^k} f_{D_1 \tens D_2, \P(\y) \tens \Q(\y)}(\x) \\
                  &= \orl_{\y \in \bin^k} \lr{f_{D_1, \P(\y)_T}(\x) \andl f_{D_2, \Q(\y)_Y}(\x)} \label{eq:l2smallact}
\end{align}

Since $\P(\y)_{T}$ only acts on $D_1$, it is a $k$-\emph{junta} projector and can be implemented by a DNF of size at most $2^k + 1$ according to \Cref{lem:d1junta}. Also, from \Cref{cor:d1sepcomb}, $f_{D_2, \Q(\y)}$ can be implemented by a DNF of size at most $O(n^{k+1})$.  
The AND of two DNFs of size $s_1$ and $s_2$ can be written as a DNF of size $s_1\cdot s_2$ whose terms are all pairwise ANDs of all pairs of terms in the original DNFs. Therefore, $f_{C,\Pi}(\x)$ can be written as a DNF of size at most $O(4^k \cdot n^{k+1})$.
\end{proof}

\subsection{\texorpdfstring{\ACZ}{AC0} Simulation of cleaned-up depth-3 circuits}
\label{sec:ac0simulation}

Now we have all the pieces to prove \Cref{thm:d3ac0}. 
\acsimthm
\begin{proof}
  Suppose that $C$ computes a function $f(\x)$ with $n$ coordinates (for large enough $n$) on register $t$ in the $\ket{\mu_b}$ basis. Then, $f_{C,\kb{\mu_1}}(\x) = f(\x)$.
  Let $G(S,t) = (I - 2\kb{\vth}_{S,t})$ be the final (layer-3) gate on $t$ and let $D$ be the depth-$2$ sub-circuit. Since other layer-$3$ gates do not affect the output on $t$, without loss of generality, we can assume that $C = G \cdot D$ (i.e., there are no other layer-$3$ gates other than $G$).

  Let $\oth_S = (I - \kb{\vth}_S)$ and $\ket{\widetilde{\mu_b}} = (I_t - 2\kb{\th_t}) \cdot \ket{\mu_b}$ for $b \in \bin$. Consider the following function, 
  \begin{align}\label{eq:outcases}
    g(\x) = \begin{cases}
      f_{D,\oth_S \tens \kb{\mu_1}}(
      \x), & \text{ if } f_{D,\oth_S}(\x) = 1 \\
      f_{D, \kb{\widetilde{\mu_1}}}(\x), & \text{ otherwise }
    \end{cases}.
  \end{align}
    It is sufficient to argue that: (1) $g(\x) = f_{C,\kb{\mu_1}}(\x) = f(\x)$ and  (2) $g(\x) \in \aczm(O(m^4 \cdot n^4),3)$. 

  \paragraph{Proof of (1).}
  Recall from \Cref{sec:pregates} that we can treat $t$ as the ``target'' and $S$ as the controls in $G$ to decompose it as:
  \begin{align}\label{eq:d3gatedec}
      G(S,t) &= \oth_S \tens I_t + \kb{\vth}_S \tens (I_t - 2\kb{\th_t})
  \end{align}

 We will show that the output of $g(\x)$ is correct for each of the two cases in \Cref{eq:outcases} (i.e., for the case $f_{D,\oth_S}(\x) = 1$ and its complement).
  \paragraph{The Case $f_{D,\oth_S}(\x) = 1$:} In this case $\oth_S \cdot D(\x) \neq 0$. 
  Then, there must exists a $b\in \{0,1\}$ such that $\oth_S \tens \kb{\mu_b} \cdot D(\x) \neq 0$, because $\kb{\mu_0}_t + \kb{\mu_1}_t = I_t$.  
  This means that $f_{D,\oth_S \tens \kb{\mu_b}}(\x) = 1$ for some $b\in \{0,1\}$.
  
 On the other hand, $\oth_S$ deactivates $G$, i.e., $\oth_S \cdot G = \oth_S \tens I_t$, and thus for any $b \in \bin$, 
  \begin{align}
      \oth_S \tens \kb{\mu_b} \cdot C = \oth_S \tens \kb{\mu_b} \cdot D.
  \end{align}
  In particular for $b = \neg f(\x)$ the LHS of the above equation equals $0$ and thus $f_{D,\oth_S \tens \kb{\mu_b}}(\x)=0$.
  Overall, we have shown that $f_{D,\oth_S \tens \kb{\mu_b}}(\x)=1$ for exactly one $b\in \{0,1\}$, and this $b$ must be $f(\x)$. Therefore, in the case $f_{D,\oth_S}(\x) = 1$, we have $f(\x) = f_{D,\oth_S \tens \kb{\mu_1}}(\x)$.
%
%
%
  \paragraph{The Case $f_{D,\oth_S}(\x) = 0$:} 
  In this case, 
  $\oth_S \cdot C(\x) = \oth_S \cdot D(\x) = 0$, since $\oth_S$ deactivates $G$.
  Recall that $f_{C,\kb{\mu_1}}(\x) = f(\x)$.
  When we project $C(\x)$ onto $\kb{\mu_1}_t$ we get 
  \begin{align*}
  \kb{\mu_1}_t \cdot C(\x) &= ((\oth_S + \kb{\vth}_S) \tens \kb{\mu_1}_t) \cdot C(\x)\\ 
 &= (\kb{\vth}_S \tens \kb{\mu_1}_t) \cdot C(\x)\\
  &= (\kb{\vth}_S \tens \kb{\mu_1}_t) (\oth_S \tens I_t + \kb{\vth}_S \tens (I_t - 2\kb{\th_t}))\cdot D(\x)\tag{Using \Cref{eq:d3gatedec}}\\
  &= (\kb{\vth}_S \tens (\kb{\mu_1}_t\cdot (I_t - 2\kb{\th_t}))) \cdot D(\x)
  \end{align*}
  and by denoting $\ket{\widetilde{\mu_1}} = (I_t - 2\kb{\th_t}) \cdot \ket{\mu_1}$,  
  \begin{align*}
(I_t - 2\kb{\th_t})  \cdot \kb{\mu_1}_t \cdot C(\x) &= 
    (\kb{\vth}_S \tens \kb{\widetilde{\mu_1}}_t)\cdot D(\x)\\
    &= 
    (\oth_S+ \kb{\vth}_S) \tens \kb{\widetilde{\mu_1}}_t)\cdot D(\x)\tag{Since $\oth_S \cdot D(\x) = 0$}\\
    &=(I_S \tens \kb{\widetilde{\mu_1}}_t)\cdot D(\x)
  \end{align*}
  and thus in this case, $f(\x) = f_{C,\kb{\mu_1}}(\x) = f_{D,\kb{\widetilde{\mu_1}}}(\x)$.

  \paragraph{Proof of (2).}
By Part (1), we have that \begin{align}
  g(\x) &= \lr{f_{D, \oth_S}(\x) \andl f_{D,\oth_S \tens \kb{\mu_1}}(\x) }      \orl \lr{\neg f_{D, \oth_S}(\x)  \andl f_{D, \kb{\widetilde{\mu_1}}}(\x) }
\end{align}
so it suffices to show that each of the three activation functions can be implemented in \ACZ.

  Let $m_1,m_2$ be the number of layer $1$ and $2$ gates respectively and for $\ell \in [2], j \in [m_{\ell}]$, let $B_{\ell,j}$ denote the set of qubits in the $j$th gate at layer $\ell$. Let $B_{\ell,0}$ denote the set of qubits in layer $\ell$ without any gates. To apply \Cref{lem:d2ac0small}, we will first partition the qubits of $S$ into subsets that belong to at most one gate per layer.  
  For $0 \leq i \leq m_1$ and $0 \leq j \leq m_2$, define $S_{ij}$ as,  
  \begin{align}
    S_{ij} = B_{1,i} \cap B_{2,j} \cap S 
  \end{align}
  Recall that $\oth_S = (I - \kb{\vth}_S)$, and let $\oth_{S_{ij}} := (I - \kb{\vth}_{S_{ij}})$. For $S_{ij} = \emptyset$, we define $\oth_{S_{ij}}$ to be the $0$ projector. Then, from \Cref{lem:projdecomp}
  \begin{align}
    \eig{\oth_S} &= \spn{\clr{\eig{\oth_{S_{ij}} \tens I_{S \setminus S_{ij}}}}}_{i \in \blr{0,m_1},\ j \in \blr{0,m_2}}\\
    \eig{\oth_S \tens \kb{\mu_b}_t} &= \spn{\clr{\eig{\oth_{S_{ij}} \tens \kb{\mu_b}_t \tens I_{S\setminus S_{ij}}}}}_{i \in \blr{0,m_1},\ j \in \blr{0,m_2}}
  \end{align}
  From \Cref{lem:componentinside}, 
  \begin{align}
    f_{D, \oth_S}(\x) &= \orl_{i \in \blr{0,m_1},\ j \in \blr{0,m_2}} f_{D,\oth_{S_{ij}}}(\x) \\
    f_{D,\oth_S \tens \kb{\mu_b}}(\x) &= \orl_{i \in \blr{0,m_1},\ j \in \blr{0,m_2}} f_{D,\oth_{S_{ij}} \tens \kb{\mu_b}}(\x).
  \end{align}
  By definition $\oth_{S_{ij}}$ contains qubits from at most one gate per layer. 
  Then, due to \Cref{lem:d2ac0small}, both these functions can be implemented by DNFs of size $O(n^2 m^2)$.
  Finally, the function $f_{D, \kb{\widetilde{\mu_1}}}(\x)$ can be implemented by a DNF of size $O(n^{2})$ since it is a $1$-junta projector on both the depth-$1$ and depth-$2$ sub-circuits. Overall,
  \begin{align}\label{eq:g}
    g(\x) &= \lr{f_{D, \oth_S}(\x) \andl f_{D,\oth_S \tens \kb{\mu_1}}(\x) }      \orl \lr{\neg f_{D, \oth_S}(\x)  \andl f_{D, \kb{\widetilde{\mu_1}}}(\x) }
  \end{align}
  can be implemented as a depth-$4$ \ACZ circuit of $O(m^2 \cdot n^2)$ size. 
  A slightly more careful analysis, explained next, shows that the depth is actually $3$. 

  We show how to implement \Cref{eq:g} by an OR-AND-OR circuit of size $O(m^4 n^4)$. As the AND of two DNFs of size at most $O(n^2 m^2)$, the left hand side $\lr{f_{D, \oth_S}(\x) \andl f_{D,\oth_S \tens \kb{\mu_1}}(\x) }$ can be implemented as a DNF of size $O(m^4 n^4)$, which is a special case of OR-AND-OR circuit.

  As for the right hand side, the expression $\neg f_{D, \oth_S}(\x)$ can be implemented by a CNF of size $O(n^2 m^2)$ due to De Morgan's law, which is an AND of $O(n^2 m^2)$ clauses. Furthermore,
  the expression $f_{D, \kb{\widetilde{\mu_1}}}(\x)$ can be implemented by a DNF of size $O(n^2)$.

  Finally, the AND of a CNF of size $s_1$ and a DNF of size $s_2$ can be written as a OR-AND-OR circuit of size $O((s_1 +n) \cdot s_2)$ as follows.
  Suppose $\phi = C_1 \land \dots \land C_{s_1}$ and $\psi = T_1 \lor \dots \lor T_{s_2}$ are CNF and DNF formulas respectively where each clause $C_i$ is an OR of literals and each term $T_j$ is an AND of literals. Then, the AND of $\phi$ and $\psi$ can be written as a OR-AND-OR circuit of size $O(s_1 \cdot s_2)$ as follows.
  \begin{align*}
    \phi \land \psi &= (C_1 \land \dots \land C_{s_1}) \land (T_1 \lor \dots \lor T_{s_2}) \\
    &= \orl_{i=1}^{s_2} (C_1 \land \dots \land  C_{s_1} \land T_i)
  \end{align*}
  which is a OR-AND-OR circuit of size $O(s_1 \cdot s_2)$
  as $(C_1 \land \dots \land  C_{s_1} \land T_i)$ is a CNF with at most $s_1 + n$ clauses as the term $T_i$ is an AND of at most $n$ literals, a special case of a CNF of size $n$.
\end{proof}

\begin{corollary}[depth-3 Sub-exponential Lower Bound]\label{cor:d3anclb}
  Let $C$ be a single-output cleaned up depth-$3$ \QACZ~circuit $C$ on $n$ inputs with $m = \exp(o(\sqrt{n}))$ gates that computes the function $f(\x)$ on $n$ coordinates. Then, $f(\x)$ cannot be the \PARITY or \MAJORITY function.
\end{corollary}
\begin{proof}
First, we apply the clean-up step from \Cref{lem:l1cleanup} while keeping at least $n' = n/3$ coordinates. 
Then, by \Cref{thm:d3ac0}, $f(\x)$ is a function that can be computed by a depth-$3$ \ACZ circuit of size $O(m^4 \cdot n^4)$. Based on the known \ACZ lower-bound of \cite{hastad1986switch}, if $f(\x)$ is either \PARITY or \MAJORITY, it requires $m^4 n^4 \geq \exp(\Omega(\sqrt{n}))$, making $m \geq \exp(\Omega(\sqrt{n}))$.
\end{proof}
\subsection{Size-independent depth-3 Lower-Bound for \PARITY}\label{sec:d3unlim}
\newcommand{\calB}{\mathcal{B}} 
\renewcommand{\K}{\mathcal{K}}

Recall that if $C$ computes $f(\x)$ on target $t$ in $(\ket{\mu_0}, \ket{\mu_1})$ basis, $f_{\kb{C,\mu_u}}(\x) = f(\x)$ and \Cref{thm:d3ac0} gives an \ACZ circuit for this function for such depth-$\leq 3$ cleaned up circuit $C$. We will use the fact that \PARITY is invariant under any classical restriction to construct a specific $\leq 2n/3$ qubit restriction to simplify this \ACZ circuit.  

For a circuit $C$ and projector $\Pi$, we say that $\Pi$ \emph{kills} $C$ if $\Pi \cdot C(\x) = 0$ for all $\x \in \bin^n$. Similarly, we say that a classical restriction $R$ \emph{kills} a boolean function $f(\x)$ if $f|_R(\x) = 0$ for all $\x \in \bin^{n-|R|}$. For activation functions, note $\lr{f_{C,\Pi}}|_R(\x) = f_{C|_{R},\Pi}(\x) $. 

For any layer $2$ gate $G(S) = (I - 2\kb{\vth}_S)$ of a depth-$2$ circuit $C$, we use the notation $f_{C,G}(\x)$ to refer to the activation function $f_{C,\kb{\vth}_S}(\x)$. Then, since $f_{C,\kb{\vth}_S}(\x) = f_{C^1,\kb{\vth}_S}(\x)$, where $C^1$ is the depth-$1$ sub-circuit of $C$, it follows from \Cref{lem:d1ac0} that either $f_{C,G}(\x) = 1$ for all $\x$ (trivial) or it is the \AND function on a subset of the variables. Then, we have the following two key observations. 

\begin{restatable*}
[Limited non-monotonicity]{lemma}{onesidedmono}\label{lem:onesidefanout}
Let $C$ be a cleaned up depth-$2$ circuit on $n > 2$ inputs. Then, there exists a string $z\in \{0,1\}^n$, such that for each $i \in [n]$, the single-coordinate restriction $R = (i,z_i)$ kills $f_{C,G}$ for at most two layer-2 gates $G$ with non-trivial activation function. 
\end{restatable*}

\begin{restatable*}[Single-qubit activation functions at depth-$2$]{lemma}{singleqbadproj}\label{lem:d2badproj}
    Let $C$ be a cleaned-up depth-$2$ circuit and $\kb{\eta}_q$ be a projector that does not kill $C$. Let, $G = (I - 2\kb{\vth}_{S})$ be the layer $2$ gate containing $q$. Then, at least one of the below must hold. 

    \begin{enumerate}
   \item  Either there is a restriction $R$ on $|R| \leq 2n/3$ coordinates that fixes $f_{C|_R, \kb{\eta}_q}(\x) = 1$ for all $\x \in \bin^{n-|R|}$,
    
    \item Or $f_{C, \kb{\eta}_q}(\x) = f_{C,G}(\x)$ and is given by an \AND function (up to NOT gates) of width at least $2n/3$.
    \end{enumerate}
\end{restatable*}

These are sufficient to prove our main theorem and we defer their proofs to the next section. 

\begin{theorem}[Depth-$3$ Lower-Bound with Unlimited Ancillae] \label{thm:depth_3_inf_anc}
 Let $C$ be a depth-$3$ $\qac^0$ circuit on $n = \omega(1)$ coordinates and an arbitrary number of ancillae and size, that computes the function $f(\x)$.  Then $f(\x)$ cannot be the \PARITY function.
\end{theorem}

\begin{proof}
Suppose for contradiction that $f(\x)$ is \PARITY.  
First we will perform the cleanup step from \Cref{lem:l1cleanup}, which preserves at least $n/3 = \omega(1)$ coordinates. Now we will proceed to prove the bound against cleaned up circuits on $n$ coordinates. 
To do so, we will construct a restriction that preserves $\Theta(n)$ qubits and simplifies the final gate on the target qubit $t$. From \Cref{sec:pregates}, the final gate on $t$, $G(S,t) = (I - 2\kb{\vth_{S,t}})$ can be viewed as a controlled-U gate on target $t$ as below. 

\begin{align}\label{eq:d3gatedec2}
    G(S,t) &= \oth_S \tens I_t + \kb{\vth}_S \tens (I_t - 2\kb{\th_t})
\end{align} 
Then, we will argue that there is a restriction to either guarantee that we are always in the $\oth = (I - \kb{\vth}_S)$ subspace, in which $G$ is inactive, or guarantee that we are orthogonal to the subspace, which simplifies $G$ into a single qubit unitary. 

First we will remove all the redundant qubits from the gates, these are the qubits that are always in the state $\ket{\th}_q$ regardless of the input, making $\kb{\pth_q}$ kill $C$. Now WLOG $S$ contains no such qubits. We will use $C^2$ to denote the sub-circuits consisting of gates up to layer $2$. 

\paragraph{Restrict to identity subspace.} 
Suppose that there is a qubit $q \in S$ such that $\ket{\eta}_q = \ket{\pth_q}$ satisfies point (1) of \Cref{lem:d2badproj}, i.e., there exists a restriction $R$ on at most $2n/3$ coordinates to fix $f_{C^2|_R,\kb{\eta}_q}(\x) = 1$. Then, $C|_R$ still correctly computes $f(\x)$ on $\x \in \bin^{n-|R|}$ in some basis $\lr{\ket{\mu_0}, \ket{\mu_1}}$ on $t$. Since $\kb{\mu_0} + \kb{\mu_1} = I$, $f_{C^2|_R, \kb{\eta}_q \tens \kb{\mu_b}}(\x)$ is nonzero for some value of $b \in \bin$.
Additionally, for $b = f(\x) \oplus 1$, $f_{C|_R,\kb{\eta} \tens \kb{\mu_b}}(\x) = 0.$
Hence,  
\begin{align} 
f_{C^2|_R,\kb{\eta} \tens \kb{\mu_1}}(\x) &= f_{C|_R, \kb{\eta} \tens \kb{\mu_1}}  \ \ \ \ \ \ \  (\text{due to } \kb{\eta} \cdot G = \kb{\eta} \tens I)\\
&= f(\x).
\end{align} 
However, from \Cref{lem:d2ac0small}, $f_{C^2|_R, \kb{\eta}_q \tens \kb{\mu_1}}(\x) \in \aczm(O(n^3), 2)$. From the known bounds of \cite{hastad1986switch} for \ACZ circuits, this is a contradiction to $n = \omega(1)$. 

\paragraph{No restriction to identity subspace.}
This is when the previous case doesn't apply. Let $H_1 \dots H_m$ be the layer-$2$ gates containing at least one qubit from $S$. Then, for every qubit $q \in S$ belonging to the layer $2$ gate $H_j$, \Cref{lem:d2badproj} gives that 
$f_{C^2,\kb{\pth_q}}(\x) = f_{C^2,H_j}(\x)$ and is a width $\geq 2n/3$ \AND function. 

Now we will argue that there is a restriction $R$ on $\leq 0.9n$ inputs, that makes $f_{C^2|_R, H_j}(\x) = 0$ for all $\x$. This has the effect of
making $f_{C^2|_R, \kb{\pth_q}}(\x) = 0$ for all $q \in S$ forcing the state on $S$ after $C^2$ to be $\ket{\vth}_S$ regardless of input. Then, from \Cref{eq:d3gatedec2}, this simplifies $G(S,t)$ to a single qubit unitary acting only on $t$, giving us a depth-$\leq 2$ circuit computing parity on $n' \geq 0.1n$ coordinates. 

We will use \Cref{lem:onesidefanout} to construct $R$. The idea is that the activation function corresponding to all but a $O(1)$ number of these gates is monotone in the same direction for a large fraction of the inputs. This allows us to construct a restriction that kills all these gates simultaneously while handling $O(1)$ other gates separately. 

For each gate $H_j$ with $j \in [m]$, define $\K(H_j)$ to be the set of all single coordinate restrictions that kill $f_{C^2, H_j}$, i.e., 
\begin{align}
    \K(H_j) := \clr{(i,b) :  \forall \x \ \ f_{C^2|_{(i,b)}, H_j}(\x) = 0 }_{i \in [n], b\in \bin}
\end{align}
Recall that each $f_{C^2,H_j}(\x)$ is an \AND of $\geq 2n/3$ literals, which implies $\vlr{\K(H_j)} \geq 2n/3$.

First, let $\calB$ be the set of all the gates $H_j$ for $j \in [m]$ such that $\K(H_j)$ contains at least $0.1n$ elements of the form $(i,z_i)$, where $z_i$ is defined according to \Cref{lem:onesidefanout}. Then, from \Cref{lem:onesidefanout}, $\vlr{\calB} \leq 2n/0.1n \leq 20$, and each $\vlr{\K(H_j)} > 20$. By picking a coordinate $(i,z_i) \in \K(H_j)$ from each $H_j$ in $\calB$, this gives a restriction $R_0$ on $\leq 20$ qubits that kills all these gates in $\calB$ simultaneously.

Let $R_1$ be a restriction on arbitrary $3/4n$ inputs not in $R_0$, such that each $i$ in $R_1$ is set to $x_i = z_i \oplus 1$. We will argue that $R_1$ kills all the gates $H_j$ outside $\calB$.  
Then, $R = R_1 \cup R_0$.  
Every gate in $\calB$ is killed in $C^2|_R$ due to $R_0$. Now, for every gate $H_j \not \in \calB$, $\K(H_j)$ contains at least $2/3n - 0.1n > 0.5n$ elements of the form $(i, z_i \oplus 1)$. Since $R_1$ leaves only $0.25n$ elements unfixed, there is at least one such element from $\K(H_j)$ in $R_1$. 

Therefore, $R$ is such that $f_{C|_R, H_j}(\x) = 0$ for all $j \in [m]$ and $\x \in \bin^{n-|R|}$.  
Furthermore, $|R| \leq |R_0| + |R_1| \leq 0.75n + 20 \leq 0.8n$, for $n$ large enough. Hence, $C|_R$ is a depth-$2$ circuit that computes parity on $\geq 0.2n$ inputs and it must be that $0.2n \leq 6$ due to \Cref{cor:d2par}. 
\end{proof}

\subsubsection{Intermediate proofs}
Now we provide the proofs of the two lemmas.

\onesidedmono
\begin{proof}
For each $i$, let $\ket{\th^1_i}$ be the component of $i$ corresponding to its layer $1$ gate $G_i = (I - 2\kb{\vth^1}_{A_i,i})$ containing the ancillae $A_i$. 

We will set $z_i = \th^1_i$ if $\ket{\th^1_i} \in \clr{\ket{0}, \ket{1}}$ or $z_i = 0$ otherwise. For any layer $2$ gate $G = (I-2\kb{\vth^2_S})$ acting on qubits $S$, we will argue that $(i, z_i)$ kills $f_{C,G}(\x)$ \emph{only if} either $i \in S$ or $A_i \subseteq S$. Then, this implies that there can be at most two layer $2$ gates that are killed by $(i, z_i)$.

Note that since measurements corresponding to layer $2$ reflections commute with layer $2$ gates, we have that $f_{C,G}(\x) = f_{C^1,\kb{\vth^2}_S}$, where $C^1$ is the layer $1$ sub-circuit.  
Additionally, by the assumption that $f_{C,G}(\x)$ is non-trivial and from \Cref{lem:d1ac0}, we have that any $(i,b)$ kills $f_{C,G}(\x)$ if and only if it kills $f_{G_i, \kb{\vth^2}_{S}}(\x)$, where $G_i$ here denotes the circuit consisting only of the layer-1 gate acting on $i$. 

Consider the case when $i \not \in S$ and $A_i \not \subseteq S$. We consider the mixed state of $G_i(z_i)$ on $A_i \cap S$.
This is obtained by measuring qubits $i$ and $A_i \setminus S$ in any basis, so in particular, we measure them according to $\th^1$. We note that (i) if we measure $\ket{\th^1_a}$ for all $a\in A_i \setminus S$ and $\ket{\th^1_i}$ for $i$, then the remaining gate is a non-trivial reflection gate on $A_i \cap S$. Denote the result of this gate on the ancillae $A_i \cap S$ by $\ket{\psi}$. (ii) Otherwise, the remaining gate is identity on $A_i \cap S$. By definition of $z_i$ and the fact that we assumed that layer-1 gates are non-trivial, i.e., that $\braket{\th^1_a | 0} \neq 0$ for any ancilla $a$ starting in $\ket{0}$, we get that we are in case (i) with non-zero probability. By the fact that we assumed that layer-1 gates are non-trivial, i.e., that $\vlr{\braket{\th^1_a | 1}} \neq 1$ for any ancilla $a$, we get that we are in case (ii) with non-zero probability. We see that in the two cases we got different vectors, so the Schmidt rank of the mixed state is 2.

Since $f_{C,G}$ is non-trivial, $\kb{\th^2_{S}}$ does not kill $G_i$, and so it must be that for at least one of the two vectors $v\in \{\ket{0}, \ket{\psi}\}$ we have $\kb{\th^2_{S}} \cdot v \neq 0$. This means that with non-zero probability we get in $G_i(z_i)$ a state that is not killed by $\kb{\th^2_{S}}$, and thus $(i,z_i)$ does not kill $f_{G_i, \kb{\vth^2}_S}$. 

The remaining cases, $i \in S$ or $A_i \subseteq S$, can happen for at most two such layer-2 gates. Therefore, $(i, z_i)$ can kill $f_{C,G}$ for at most two layer-2 gates $G$.
\end{proof}

\singleqbadproj
\begin{proof}
We will prove the lemma by considering $f_{C^1,\kb{\eta}}(\x)$ where $C^1$ is the depth-$1$ sub-circuit.

\paragraph{The case when $\kb{\eta_q}$ kills $C^1$.}
It must be the case $[\kb{\eta}_q, G] \neq 0$, since $\kb{\eta}_q$ doesn't kill $C$, making $\braket{\eta_q | \th_q} \neq 0$. 
Then, on any input $\x$, 
  \begin{align}
    \kb{\eta}_q \cdot C(\x) &= \kb{\eta}_q \cdot (I - 2\kb{\vth^2_S}) \cdot C^1(\x)  \\
    &=\kb{\eta}_q \cdot (-2\kb{\vth^2_S}) \cdot C^1(\x)  \\
                            &= - 2 \ket{\eta}_q \ket{\vth^2}_{S \setminus q} \cdot \braket{\eta_q | \th_q} \cdot \lr{\bra{\vth^2_S} \cdot C^1(\x)} \label{2.9}
  \end{align}
  and thus, $f_{\kb{\eta}_q}(\x) = f_{C,\kb{\vth}_S}(\x) = f_{C,G}(\x)$. 

Since $f_{C,\kb{\vth}_S}(\x) = f_{C^1,\kb{\vth}_S}(\x)$, it follows from \Cref{lem:d1ac0} that $f_{C,\kb{\vth}_S}(\x)$ is an \AND gate of width $w \leq n$. If $w \leq 2n/3$ then we have our restriction $R$ by fixing all the coordinates to satisfy the gate. Otherwise, $w \geq 2n/3$ as required.

\paragraph{The case when $\kb{\eta}_q$ does not kill $C^1$.} 
If $f(\x) = f_{C,\kb{\eta}_q}(\x)$ is a $2$-junta, then since $\kb{\eta}_q$ doesn't kill $C$, there is a restriction on the $\leq 2$ coordinates that $f(\x)$ depends on that makes it always $1$.  

Otherwise, if $f(\x)$ is not a $2$-junta there are at least two input-dependent layer $1$ gates in the light-cone of $\kb{\eta}_q$. Then, there must be a qubit $r$ that doesn't share layer $1$ gates with $\kb{\eta}_q$ such that $\kb{\pth_{r}}$ does not kill $C^1$. This is because we can remove any ``redundant qubits", i.e., qubits that are always in the $\ket{\th}_r$ state before $G$, until we are either left with a small light-cone or find such an $r$.  

Then, for $\Pi = \kb{\pth_r} \tens \kb{\eta}_q$, $f_{C^1,\Pi}$ is given by \Cref{lem:d1ac0} and by assumption, is not the constant function $0$. Additionally, $f_{C^1,\Pi}(\x)$ is a 2-junta. Then, observe that for all $\x$,
\begin{align}
    \Pi \cdot G(S) \cdot C^1(\x) &= \Pi \cdot C^1(\x)
\end{align}
Therefore, $f_{C,\Pi}(\x) = f_{C^1,\Pi}(\x)$ and we can let $R$ be the restriction that fixes this 2-junta to $1$. 
\end{proof}

\section{Depth-2 \texorpdfstring{\QACZ}{QAC0} Circuits Have Small Total Influence}
\label{sec:depth_2_influence}
In this section, we will prove that any depth-2 \QACZ~circuit has small total influence, as defined in \Cref{sec:prelim_aobf}. In particular, we will show that, regardless of the number of ancillae, the total influence of an $n$-input depth-2 \QACZ~circuit is upper-bounded by $O(\log n)$.
To establish this, we will prove a stronger bound on the Fourier concentration (showing exponentially small Fourier tails) of the circuit's outcome as a function of the input qubits.

Since \PARITY~has total influence $n$, this implies that the functions computable by depth-2 \QACZ~circuits have small correlation with \PARITY. Therefore, this result offers a novel average-case depth-$2$ lower-bound against \PARITY~$\in$~\QACZ~with unlimited ancillae. More generally, the result implies an average-case depth-$2$ lower-bound against any Boolean function with large total influence.

\subsection{Main Proof}
We will now prove the main result of this section, showing that for any depth-$2$ \QACZ~circuit $C$, the Fourier tail of the function $f_C(x) = \Pr[\text{$C$ accepts $x$}]$ at level $k = c\log(1/\eps)\log(n/\eps)$ is at most $\eps$, i.e. $\W^{\ge k}[f_C] \leq \eps$. Along the way, we will introduce several intermediate claims and lemmata that will be proved in \Cref{sec:inf_proofs}. The formal theorem statement is as follows.
\begin{theorem} \label{thm:influence}
There exists a constant $c\ge 1$ such that the following holds. Let $C$ be a depth-$2$ \QACZ circuit with $n$ input qubits and any number $a$ of ancilla qubits.
Consider the function $f_C: \B^n \to [0,1]$ defined by $f_C(x) = \Pr[\text{$C$ accepts $x$}]$. Then,  for any $\eps>0$,
\begin{align}
    \W^{\ge k(\eps)}[f_C] \leq \eps, \quad \textnormal{ where } \; k(\eps) = c\log(1/\eps)\log(n/\eps)
\end{align}
\end{theorem}

\begin{proof}[Proof of \Cref{thm:influence}]
We will give an overview of the full proof with reference to several intermediate claims and lemmata. We defer the proofs of these intermediate results to \Cref{sec:inf_proofs}.

We begin the proof by simplifying the \CZ~gates at layer $1$ of the circuit, i.e., those closest to the inputs.
We replace each gate that depends on many input qubits with the identity gate (incurring a small error) and then apply a random restriction tailored to the circuit structure. The purpose of these two steps is to reduce to a more structured case, where each gate at layer $1$ depends on at most one input qubit (but potentially many other ancillae). Moreover, we show that these steps behave well with respect to the Fourier tails of the circuit.

\paragraph{Reduction to a Structured Circuit.}
We consider the gates at layer $1$. Recall that in \QACZ, gates at a certain layer are non-overlapping. This means that for each of the layer-$1$ gates, $g_1, \dots, g_\ell$, there is a set of incoming input bits, $S_1, \dots, S_{\ell} \subseteq[n]$, such that $S_1, \dots, S_{\ell}$ are disjoint (where possibly some of the $S_i$ are empty sets, as the corresponding gates only depend on ancillae).

We first show that any gate $g_i$ whose corresponding $S_{i}$ contains more than $b \triangleq \log(16n/\eps)$ inputs can be replaced with the identity gate, while incurring smalle error. In particular, \Cref{remove_large_layer_1_gates} shows that replacing each such gate with identity changes $f_C$ by at most $\eps/2n$ in Fourier weight above level $k$.
Since there are at most $n$ such gates (because each gate depends on at least one input qubit and the gates are non-overlapping), the Fourier weight above level $k$ differs by at most $\eps/2$. 

\begin{restatable}{claim}{RemoveLOne} \label{remove_large_layer_1_gates}
Let $g$ be a CZ-gate at layer $1$ with input qubits $S\subseteq[n]$ entering it. Let $C'$ be the circuit where $g$ is replaced with identity.
Then $\|f_{C}-f_{C'}\|_2^2 \le 4\cdot 2^{-|S|}$.
Furthermore, for any $k$, the Fourier weight above level $k$ of $f_C$ and $f_{C'}$ differs by at most
\begin{align}
    |\W^{\ge k}[f_C]-\W^{\ge k}[f_{C'}]| \leq 8 \cdot 2^{-|S|}.
\end{align}
\end{restatable}

\paragraph{Random Restrictions.} We replace $C$ with a circuit $C'$, replacing the above mentioned layer-$1$ gates that depend on more than $b$ input qubits with identity, and continue to analyze the Fourier tails of $f_{C'}$.
Each layer-$1$ gate of $f_{C'}$ depends on at most $b$ input qubits. We will apply a \emph{random-valued} restriction that randomly picks at most one input qubit per gate keeping it alive and sets the remaining inputs uniformly at random. The random restriction will sample a random subset $J \subseteq [n]$ of the variables to stay alive, based on the circuit structure, and a partial assignment $z\in \{0,1\}^{[n] \setminus J}$ sampled uniformly at random. 

We will now describe the explicit random restriction process for sampling $J$. We consider only the gates $g_i$ for which $1 \le |S_i| \le b$ (recall that $b=\log(16n/\eps)$). For each such gate, we pick exactly one of the bits in $S_i$ to be included in $J$ uniformly at random and independently of all other choices. Input qubits not involved in any gate always remain alive. \Cref{claim:Fourier_weights_under_random_restriction} demonstrates that Fourier tails behave nicely with respect to this random restriction. 
(Observe that this restriction always keeps at least $n/b$ variables alive.)

\begin{restatable}{claim}{RandRestrictFourier}\label{claim:Fourier_weights_under_random_restriction}
    Let $b\in \N$ and $S_1, \dots, S_{\ell}\subseteq[n]$ be disjoint sets of size between $1$ and $b$.
    Define $S_0 = [n]\setminus (S_1 \cup \dots \cup S_{\ell})$ so that $S_0, \ldots, S_{\ell}$ form a partition of $[n]$.
    Consider the random valued restriction $(J,z)$ that for each $i\in \{1,\ldots, \ell\}$, picks independently uniformly at random exactly one element from to $S_i$ to $J$, and furthermore surely picks all elements in $S_0$ to $J$.
    
    Then, for any $k\in \N$ and function $f:\B^n \to \R$, the above random restriction process satisfies
    \begin{align}
        \W^{\ge 4kb}[f] \le 2\cdot  \E_{J,z}[\W^{\ge k} [f|_{J,z}]].
    \end{align}
\end{restatable}

\paragraph{Fourier Tail Bounds for Structured Circuits.}
With \Cref{claim:Fourier_weights_under_random_restriction}, the proof of \Cref{thm:influence} is thus reduced to proving Fourier tail bounds of the structured depth-$2$ \QACZ~circuits, which we denote $C''$, with layer-$1$ gates each depending on at most one input qubit (and potentially many ancillae).
Specifically, it remains to prove that for any such circuit $C''$,  $\W^{\ge k}[f_{C''}] \le \eps/4$ for $k = \Theta(\log(1/\eps))$. 

While the circuit $C''$ depends only on the variables in $J$ that were kept alive, it will be convenient to think of it as a circuit with $n$ input qubits as well, where qubits outside $J$ are ignored. Without loss of generality, since the computation is single-output, we only need to consider the single layer-2 gate, $g$, containing the target qubit. Also, since the ancilla starting state is allowed to be an arbitrary separable state, without loss of generality, the gate $g$ is a $\CZ$ gate which flips the phase iff the input is $\ket{1^m}$, where $m$ is the number of input qubits to $g$.

We partition the input qubits of $g$ into disjoint sets $Q_0, Q_1, Q_2, \ldots, Q_n$ where for $i\in \{1,\ldots, n\}$, $Q_i$ is a set, potentially empty, of the qubits involved with the layer-1 gate containing $x_i$ and $Q_0$ contains the qubits coming from ancilla-only layer-1 gates. Note that for any fixed input $x \in \{0,1\}^n$, the states for these different subsets are separable.

We will now impose that the input to the layer-2 gate $g$ is a mixed state. Specifically, any layer-1 qubit which is not contained in $g$ will be traced out.
Let $\rho_{0}$ be the mixed state of qubits in $Q_0$. For $i\in [n]$, let $\rho_{i}^{b}$ be the mixed state of the qubits in $Q_i$ when $x_i = b$, for $b\in \{0,1\}$.
(If $Q_i= \emptyset$ then $\rho_i^b = (1)$, the trivial mixed state of dimension  one, for both $b\in \{0,1\}$.)
Then, on input $x\in \{0,1\}^n$, the state entering $g$ is 
\begin{align}
    \rho^x = \rho_0 \otimes \rho_{1}^{x_1} \otimes \dots \otimes \rho_n^{x_n}.    
\end{align}

We will denote the average state over the two options for $x_i$ as
\begin{align}
    \rho_i = \frac{1}{2}(\rho_{i}^{0} + \rho_{i}^{1})
\end{align}
and the average state over all $2^n$ options for $x$ as $\rho$, where we note that
\begin{align}
    \rho = \rho_0 \otimes \rho_{1} \otimes \dots \otimes \rho_n.
\end{align}

We will divide the remainder of the proof into two main cases. First, we consider the case in which the layer-$2$ gate is almost always inactive, meaning $g$ can be replaced by identity without notably affecting the acceptance probability on most inputs. In this case, the resultant function is close to a dictator function (i.e. depends only on one input qubit) and has extremely small Fourier tails. In the second remaining case, we argue that the ``entropy'' of the incoming state to the gate is small. This turns out to imply that the function's total influence and Fourier tails are small.

\paragraph{Case 1: $g$ is almost always inactive.} 
Suppose $\bra{1^m}\rho \ket{1^m}\le \eps/32$. 
For each fixed $x$, let $\eps_x := \bra{1^m}\rho^x \ket{1^m}$. Then, $\E_x[\eps_x] = \bra{1^m}\rho \ket{1^m} \le \eps/32$
and for each fixed $x$, \Cref{CZ-meaningless} implies that $\td{\rho^x}{\CZ_m \cdot  \rho^x \cdot \CZ_m} \le 2\sqrt{\eps_x}$.

\begin{restatable}{lemma}{CZmeaningless}
\label{CZ-meaningless}
Let $\rho$ be a quantum mixed state on $m$ qubits, and let $\delta = \bra{1^m}\rho \ket{1^m}$.
Then, 
\begin{align}
    \td{\rho}{\CZ_m \cdot \rho \cdot \CZ_m} \le 2\sqrt{\delta}.
\end{align}
\end{restatable}
\noindent Therefore, removing the \CZ gate, barely changes the circuit functionality. More formally, denote by $C'''$ the circuit $C''$ with gate $g$ replaced with identity.
Then, 
\begin{align}
    \|f_{C''}-f_{C'''}\|_2^2 &= \E_{x\sim \B^n}[(f_{C''}(x)-f_{C'''}(x))^2] \\
    &\le  \E_{x\sim \B^n}\left[\td{\rho^{x}}{\CZ_{m} \cdot \rho^x \cdot \CZ_{m}}^2\right] \\
    &\le \E_{x\sim \B^n}[\eps_x] \\
    &\le  \eps/8.
\end{align}
The circuit $C'''$ (without gate $g$ replaced with identity) has only one layer of \CZ~gates, and its output depends on only a single input qubit, i.e., it computes a dictator function. Such a circuit has zero Fourier weight above level $1$, let alone $k$. Thus, via \Cref{lemma:closeness_implies_closeness_in_tails}, our circuit $C''$ has at most $\eps/4$ Fourier weight above level $k$.
\begin{restatable}{lemma}{CloseTails}\emph{[Closeness in $\ell_2$ implies closeness in Fourier tails]}
\label{lemma:closeness_implies_closeness_in_tails}
Let $f, g: \B^n \to \R$, with $\|f\|_2, \|g\|_2\le 1$ and let $k \in \N$. Then, $|\W^{\ge k}[f]-\W^{\ge k}[g]| \le 2 \cdot \|f-g\|_2^2$
\end{restatable}

\paragraph{Case 2: $g$ is active with non-neglible probability.} Otherwise, $\bra{1^m}\rho \ket{1^m} \ge \eps/8$. Using the fact that $\rho = \rho_0 \otimes \rho_{1} \otimes \dots \otimes \rho_n$,
\begin{align}
	\bra{1^m}\rho \ket{1^m} = \prod_{i=0}^{n} \bra{1^{|Q_i|}}\rho_i \ket{1^{|Q_i|}}  \ge \eps/32\;.
\end{align}
Denoting $\delta_i = 1-\bra{1^{|Q_i|}}\rho_i \ket{1^{|Q_i|}}$, this implies that
\begin{align} \label{eqn:lnsum}
    \sum_{i=0}^{n} \delta_i = \sum_{i=0}^{n} \left(1-\bra{1^{|Q_i|}}\rho_i \ket{1^{|Q_i|}}\right) \le \sum_{i=0}^{n}\ln(1/\bra{1^{|Q_i|}}\rho_i \ket{1^{|Q_i|}}) \le \ln(32/\eps).
\end{align}
Therefore, for most $i\in [n]$, $\bra{1^{|Q_i|}}\rho_i \ket{1^{|Q_i|}}$ is rather close to $1$. 
Also note that both 
\begin{align}
    \bra{1^{|Q_i|}}\rho_{i}^0 \ket{1^{|Q_i|}} \geq 1-2\delta_i \quad \text{ and } \quad \bra{1^{|Q_i|}}\rho_{i}^1 \ket{1^{|Q_i|}} \geq 1-2\delta_i,
\end{align}
so \Cref{lemma:small-influence-in-direction-i} implies that $\td{\rho_{i}^0}{\rho_i^1}\leq8\sqrt{\delta_i}$.

\begin{restatable}{lemma}{SmallInfInDirectionI}
\label{lemma:small-influence-in-direction-i}
Let $0 \le \delta\le 1$.
Suppose $\rho'$ and $\rho''$ are two mixed states on $d$ qubits, such that $\bra{1^d} \rho' \ket{1^d} \geq 1-\delta$ and $\bra{1^d} \rho'' \ket{1^d} \geq 1-\delta$.
Then, $\td{\rho'}{\rho''} \le  2\delta + 2\sqrt{\delta}$.
\end{restatable}

Next, we prove a Fourier tail bound for this case.
We observe that 
\[
\rho_i^{x_i} = \rho_i + (-1)^{x_i} \cdot D_i\qquad\text{where}\qquad D_i = \frac{\rho_i^1 - \rho_i^{0}}{2}\]
is the \emph{derivative} of $\rho_i^{x_i}$ according to $x_i$.
This will allow us to write a nice ``matrix Fourier decomposition'' for $\rho$ and a bound on the Fourier spectrum of $f_{C''}$.
With this notation, 
\begin{align*}
\rho &= \rho_0 \otimes \bigg(\rho_1 + (-1)^{x_1} \cdot D_1\bigg) \otimes \dots \otimes \bigg(\rho_n + (-1)^{x_n} \cdot D_n\bigg)	
\\
&= \sum_{R \subseteq [n]} (-1)^{\sum_{i\in R} x_i} \cdot \rho_0 \otimes \bigotimes_{i=1}^{n} [(\rho_i)^{\one_{i\not\in R}} (D_i)^{\one_{i\in R}}]
\end{align*}
So, we can think of the coefficient of $(-1)^{\sum_{i\in R} x_i}$ as the $R$-Fourier coefficient of $\rho$. 
That is, \[\widehat{\rho}(R) := 
\rho_0 \otimes \bigotimes_{i=1}^{n} [(\rho_i)^{\one_{i\not\in R}} (D_i)^{\one_{i\in R}}].\]
Note that the $R$-Fourier coefficient in this case is a matrix (instead of a scalar in the standard case).
Since all density matrices have trace-norm 1, $\|\rho_i\|_1 =1$ for all $i \in [n]$. Therefore, since the trace is multiplicative with respect to tensor products,
\begin{align} \label{eqn:derivs}
    \|\widehat{\rho}(R)\|_1 = \prod_{i\in R} \|D_i\|_1.
\end{align}
Furthermore, observe that, by definition of trace distance, $\|D_i\|_1 = \td{\rho_i^1}{\rho_i^{0}}$ .

Now, $f_{C''}(x)$ is an application of a unitary $U$ and a projection $\Pi$ on $\rho^x$. As such, 
\begin{align*}
f_{C''}(x) = \Tr(\Pi U \rho^x U^{\dagger} )	 &=
\Tr\left(\Pi U \sum_{R\subseteq [n]} (-1)^{\sum_{i\in R}x_i} \widehat{\rho}(R) U^{\dagger}\right)
=\sum_{R\subseteq [n]} (-1)^{\sum_{i\in R}x_i} \Tr(\Pi U \widehat{\rho}(R) U^{\dagger} ).
\end{align*}
From hereon, we will denote $\widehat{f_{C''}}(R) =  \Tr(\Pi U \widehat{\rho}(R) U^{\dagger})$. By properties of the trace norm, $|\widehat{f_{C''}}(R)| \le \|\widehat{\rho}(R)\|_1$. Thus, by \Cref{eqn:derivs},
\begin{align}
    |\widehat{f_{C''}}(R)|\le \|\widehat{\rho}(R)\|_1 = \prod_{i\in R} \|D_i\|_1 = \prod_{i\in R}\td{\rho_i^{0}}{\rho_i^1} \le \prod_{i\in R} 8\sqrt{\delta_i}.
\end{align}
Therefore, for any $\ell \in \N$,
\begin{align*}\W^{\ell}[f_{C''}] = \sum_{\substack{R\subseteq [n] :\\|R|=\ell}} |\widehat{f}_{C''}(R)|^2 
&\le \sum_{\substack{R\subseteq [n] :\\|R|=\ell}} ~ \prod_{i\in R} (64 \delta_i)\\
&\le \frac{(\sum_{i=1}^n (64\delta_i))^\ell}{\ell!}\tag{Maclaurin's Inequality}\\
&\le \frac{(64\ln(32/\eps))^\ell}{\ell!} 
\le \left(\frac{e\cdot 64\ln(32/\eps)}{\ell}\right)^\ell\;.
\end{align*}
To obtain $\W^{\ge k}[f_{C''}] \le \eps/4$, it suffices to pick $k = 2e \cdot 64\ln(32/\eps) = \Theta(\ln(1/\eps))$.
\end{proof}

\subsection{Influence and Parity Correlation Bounds}

Before proving the lemmata and claims used in the proof of \Cref{thm:influence}, we prove three useful corollaries on the Fourier tails of $f_C$: (i) a bound on the tail $\W^{\ge k}[f_C]$ for $k$ in terms of $k$ and $n$, (ii) a bound on the total influence of $f_C$, and (iii) a bound on the correlation of $f_C$ with the \PARITY function.
All follow easily from \Cref{thm:influence}.

\begin{corollary} \label{thm:weigt_reex}
Let $c>1$ be the constant in \Cref{thm:influence}.
 For level $k\leq 2c\log^2 n$, we have $\W^{\ge k}[f_C] \le  \exp(-\Omega(k/\log n))$, and 
 for level $k\geq 2c\log^2 n$, we have 
 $\W^{\geq k}[f_C] \le \exp(-\Omega(\sqrt{k}))$.
 \end{corollary}

 \begin{proof}
 \Cref{thm:influence} implies that for any $s\ge 0$, if we want to get a Fourier tail bound of $\W^{\ge k}[f]\le  2^{-s}$ it suffices to take $k = c\log(1/2^{-s}) \log(n/2^{-s}) = cs(s +\log n)$.
 We consider two cases separately depending on whether $s \le \log n$ or not.
 Note that $s\le \log n$ if and only if $k \le 2c\log^2 n$, so we divide to cases based on this condition.
 
    \paragraph{The Case $k \le 2c \log^2 n$.} Let $s = k/(2c\log n) \le \log n$ and take $\eps = 2^{-s}$.
    We want to show that $\W^{\ge k}[f_C] \le \eps$.
    Indeed, we know that $\W^{\ge \kappa(\eps)}[f_C]\le \eps$ 
    for 
    \[\kappa(\eps) = c\log(1/\eps)(\log (n/\eps)) = cs(\log n + s) \le 2cs \log n \le k.\]
    By monotonicity of Fourier tails we get
    $\W^{\ge k}[f_C] \le W^{\ge\kappa(\eps)}[f_C] \le \eps = 2^{-k/2c\log n}$.
    
    \paragraph{The Case $k \ge 2c \log^2 n$.} 
    Let $s = \sqrt{k/2c} \ge \log n$ and take $\eps = 2^{-s}$.
    We want to show that $\W^{\ge k}[f_C] \le \eps$.
    Indeed, we know that $\W^{\ge \kappa(\eps)}[f_C]\le \eps$ for $\kappa(\eps) = cs(\log n + s) \le 2cs^2 \le k$.
    Thus, we get $\W^{\ge k}[f_C] \le W^{\ge\kappa(\eps)}[f_C] \le 2^{-\sqrt{k/2c}}$
 \end{proof}
 
\begin{corollary}
    $\Inf[f_C] \le O(\log n)$. 
\end{corollary}
\begin{proof}
    By the definition of total influence,
    \begin{align}
        \Inf[f_C] = \sum_{S \subseteq [n]}|S|\cdot \widehat{f}_C(S)^2 = \sum_{k>1} k \cdot \W^{=k}[f_C]=\sum_{k>1} \W^{\geq k}[f_C].
    \end{align}
    \Cref{thm:weigt_reex} shows that
    \begin{align}
        \W^{\geq k}[f_C] \leq \begin{cases}
            \exp(-\Omega(k/\log n)), & \text{if } k \leq 2c \log^2 n\\
            \exp(-\Omega(\sqrt{k})), & \text{if } k \geq 2c \log^2 n
        \end{cases}.
    \end{align}
    Thus, for some constants $\alpha$ and $\beta$, the total influence can be bounded as 
    \begin{align*}
        \Inf[f_C] \leq  \sum_{k\le 2c\log^2 n} \exp(-\alpha k/\log n) + \sum_{k> 2c\log^2 n} \exp(-\beta\sqrt{k}) = O(\log n).
    \end{align*}
\end{proof}
\begin{corollary} \label{thm:corr_proof}
    $f_C$ is weakly correlated with \PARITY, i.e. $\langle f_C, \chi_{[n]} \rangle \leq \exp(-\Omega(\sqrt{n}))$.
\end{corollary}
\begin{proof}
    The correlation of $f_C$ with the \PARITY~function $\chi_{[n]}$ is 
    \begin{align}
        \langle f_C, \chi_{[n]} \rangle = \sum_{S \subseteq [n]} \widehat{f}_C (S) \cdot \widehat{\chi}_{[n]} (S) = \widehat{f}_C ([n]) = \sqrt{\W^{=n}[f_C]}.
    \end{align}
    Via \Cref{thm:weigt_reex}, $\W^{=n}[f_C] \leq \exp\left(-\Omega(\sqrt{n})\right)$  for $n$ large enough (as it satisfies $n\ge 2c\log^2 n$), which implies the desired result. 
\end{proof}
 
subsection{Proofs of Claims and Lemmata} \label{sec:inf_proofs}

We now provide the proofs of the intermediate claims and lemmas used in the proof of \Cref{thm:influence}.
\CloseTails*
\begin{proof}[Proof of \Cref{lemma:closeness_implies_closeness_in_tails}]
\begin{align*}
\left|\W^{\ge k}[f]- \W^{\ge k}[g]\right| 
&=\left|\E_{x}[f^{\ge k} (x)^2 - g^{\ge k} (x)^2]\right|\\
&=
\left|\E_{x}[(f^{\ge k} (x) - g^{\ge k} (x))\cdot (f^{\ge k} (x) + g^{\ge k} (x))]\right|\\
&\le  
\|f^{\ge k} - g^{\ge k}\|_2 \cdot \|f^{\ge k} + g^{\ge k}\|_2 \tag{Cauchy-Schwarz}
\\
&\le \|f^{\ge k} - g^{\ge k}\|_2 \cdot (\|f^{\ge k}\|_2 + | g^{\ge k}\|_2)\tag{Triangle Inequality}\\
&\le \|f - g\|_2 \cdot (\|f\|_2 + | g\|_2)\\& \le 2\cdot \|f-g\|_2.\qedhere
\end{align*}
\end{proof}
We call a random restriction, $(J,z)$, ``random valued'' if $J\subseteq[n]$ is picked under some arbitrary distribution but given $J$, $z$ is sampled uniformly at random from $\B^{[n]\setminus J}$.

\RemoveLOne* 
\begin{proof}[Proof of \Cref{remove_large_layer_1_gates}]
Denote by $m$ the number of qubits on which the \CZ gate $g$ depends on.
Let $\ket{\psi^{x_S}_g}$ be the pure state entering the gate $g$, as a function of $x_{S}$. Let $\ket{\psi^{x_{\overline{S}}}_{rest}}$ be the rest of the state as a function of $x_{\overline{S}}$. So, for any $x\in \B^n$ we have that the state entering the layer $1$ gates is 
$\ket{\psi^x} = \ket{\psi^{x_S}_g} \otimes \ket{\psi^{x_{\overline{S}}}_{rest}}$.
Removing the \CZ~gate is the same as considering the behavior of the circuit on the state 
\[
\CZ_m \ket{\psi^x} 
= 
(\CZ_m \ket{\psi^{x_S}_g}) 
\otimes 
\ket{\psi^{x_{\overline{S}}}_{rest}}\;.
\]
So to show that the function $f_C$ associated with the original circuit $C$ and the function $f_{C'}$ associated with the circuit $C'$, where $g$ is replaced with identity, are close in $\ell_2$-distance it suffices to show that the states are $\CZ_m \ket{\psi^x}$ and $\ket{\psi^x}$ are close for most $x$.
\begin{align*}
\| f_{C} - f_{C'} \|_2^2 &= \E_{x\sim \B^n}[(f_C(x) - f_{C'}(x))^2] \\
&\le \E_{x\sim \B^n}[\td{\ket{\psi^x}}{\CZ_m \ket{\psi^x}}^2]\\
&= \E_{x\sim \B^n}[\td{\ket{\psi_{g}^{x_S}}}{ \CZ_m \ket{\psi_{g}^{x_S}}}^2]\\
&= \E_{x\sim \B^n}[1 - |\braket{\psi_{g}^{x_S} | \CZ_m |\psi_{g}^{x_S}}|^2]\\
&= \E_{x\sim \B^n}[1 - (1-2|\braket{\psi_{g}^{x_S} |0^m}|^2)^2]\\
&= \E_{x\sim \B^n}[4|\braket{\psi_{g}^{x_S} |0^m}|^2 - 4|\braket{\psi_{g}^{x_S} |0^m}|^4]\\
&\le 4\cdot \E_{x\sim \B^n}[|\braket{\psi_{g}^{x_S}|0^m}|^2]
\end{align*}
Since $\ket{\psi_g^{x_S}}$ is separable we can write it as $\ket{\psi_g^{x_s}} = \ket{\psi_g^{(0)}} \otimes \bigotimes_{i \in S} \ket{\psi_i^{x_i}}$ so that 
\[\E_{x\sim \B^n}[|\braket{\psi_{g}^{x_S}|0^m}|^2]
 = |\braket{\psi_g^{0}|\vec{0}}|^2 \cdot \prod_{i \in S} \E_x[|\braket{\psi_i^{x_i}|0}|^2]\]
To finish the proof, we claim that for any $i \in S$, $\E_x[\braket{\psi_i^{x_i}|0}|^2]=1/2$.
Indeed, this is the average of $|\braket{\psi_i^{0}|0}|^2$ and $|\braket{\psi_i^{1}|0}|^2$, and the vectors $\ket{\psi_i^{0}}, \ket{\psi_i^{1}}$ form an orthogonal basis over $\C^{2}$, so the average inner product squared with any fixed vector will be $1/2$.

The claim about Fourier tails follows from 
\Cref{lemma:closeness_implies_closeness_in_tails}.
\end{proof}

\RandRestrictFourier*
\begin{proof}[Proof of \Cref{claim:Fourier_weights_under_random_restriction}]
By the behavior of Fourier weight under random-valued restrictions, i.e., by \Cref{lemma:random_valued_restriction}, we have
\begin{align*}
	\E_{J,z}
	[\W^{\ge k}[f|_{J,z}]]
	&= \sum_{R \subseteq [n]} \widehat{f}(R)^2 \cdot \Pr[|R \cap J|\ge k]\ge \sum_{\substack{R \subseteq [n]: \\|R|\ge 4kb}} \widehat{f}(R)^2 \cdot \Pr[|R \cap J|\ge k]
\end{align*}
Thus, it suffices to prove that for any set $R$ of size at least $4kb$, the probability $\Pr[|R \cap J|\ge k] \ge 1/2$.
Partition $R$ according to the blocks $S_0, S_1, \ldots, S_\ell$, by taking $R_0 = R\cap S_0, \ldots, R_{\ell} = R\cap S_\ell$. For each $i\in \{0,1, \ldots, \ell\}$ the probability that $|R_i\cap J| =1$ is at least $|R_i|/b$, and these events are independent. In expectation we have $\E[\sum_{i=1}^{\ell} |R_i \cap J|] = \E[|R\cap J|] \ge \frac{|R|}{b} \ge 4k$. By Chernoff bound, we get that the probability that $\sum_{i=1}^{\ell} |R_i \cap J| \ge k$ is at least $1-\exp(-(4k\cdot (3/4)^2)/2)) \ge 1-\exp(-k) \ge 1/2$ as required to finish the proof.
\end{proof}

We move on to prove \Cref{CZ-meaningless}.  Before doing so, we will need the following additional lemma.
\begin{lemma}\label{lemma_rho_rho_squared}
    Let  $\rho$ be a quantum mixed state on finite domain $X$.
    Then for each $i \in X$, $\bra{i}\rho \ket{i} \ge \bra{i}\rho^2\ket{i}$.
\end{lemma}
\begin{proof}[Proof of \Cref{lemma_rho_rho_squared}]
    We will denote the eigenvectors of $\rho$ as $\{\ket{\psi_j}\}_{j\in X}$, with corresponding eigenvalue $\{\lambda_j\}_{j\in X}$ satisfying $\lambda_j\in [0,1]$. We can decompose the standard basis vectors in terms of these eigenvalues as
    \begin{align}
        \ket{i} = \sum_{j\in X} \braket{\psi_j|i} \cdot \ket{\psi_j} = \sum_{j\in X} \alpha_{i,j} \cdot \ket{\psi_j},
    \end{align}
    which we can use to decompose $\bra{i}\rho \ket{i}$ as
    \begin{align} \label{eqn:relate}
        \bra{i}\rho \ket{i} &= \left( \sum_{k\in X} \alpha_{i,k}^* \cdot \bra{\psi_k} \right) \rho \left( \sum_{j\in X} \alpha_{i,j} \cdot \ket{\psi_j} \right) = \sum_{k,j\in X} \alpha_{i,k}^*\alpha_{i,j}\bra{\psi_k} \rho\ket{\psi_j} \\
        &= \sum_{k,j\in X} \alpha_{i,k}^*\alpha_{i,j} \lambda_j \braket{\psi_k|\psi_j} = \sum_{j\in X} |\alpha_{i,j}|^2 \lambda_j.
    \end{align}
    We do the same for $\rho^2$ which has the same eigenvectors as $\rho$, but eigenvalues $\{\lambda_j^2\}_{j\in X}$ to get 
    \begin{align}
        \bra{i}\rho^2 \ket{i}=\sum_{j\in X} |\alpha_{i,j}|^2 \lambda_j^2.
    \end{align}
	Overall, this establishes that
	\[
	\bra{i}\rho \ket{i} = \sum_{j\in X} |\alpha_{i,j}|^2 \lambda_j \ge \sum_{j\in X} |\alpha_{i,j}|^2 \lambda_j^2 = \bra{i}\rho^2 \ket{i}.\qedhere
	\]
\end{proof}
\noindent With this, we can now prove \Cref{CZ-meaningless}.
\CZmeaningless*
\begin{proof}[Proof of \Cref{CZ-meaningless}]
Let $\rho' = \CZ_m \cdot \rho \cdot \CZ_m$.
Then $\bra{i}\rho'\ket{j} = \bra{i}\rho\ket{j}$ if both $i, j \neq 1^m$ or both are equal to $1^m$, and $\bra{i}\rho'\ket{j} = -\bra{i}\rho\ket{j}$ otherwise.
We get that $\td{\rho}{\CZ_m \cdot \rho \cdot \CZ_m} = \frac{1}{2} \|A\|_1$ where $A = \rho - \rho'$ and $\|A\|_1$ is the trace norm of $A$.
Observe that $A$ is non-zero only on the last row and the last column, and that it diagonal is all zeros. Thus, it is a Hermitian matrix with rank $2$ and trace $0$ that has two real non-zero eigenvalues that sum up to $0$, which we  denote by $\lambda$ and $-\lambda$.
Looking at $A^{\dagger} A = A^2$ we see that it is a block matrix composed of a $(2^{m}-1)\times(2^{m}-1)$ block and a $1\times 1$ block. The entry in the $1\times1$ block is one of the eigenvalues of $A^2$, so it equals $\lambda^2$ and also (by definition) equals $\sum_{i\in \{0,1\}^m} |\bra{1^m}A \ket{i}|^2$.
As the trace norm of $A$ is $2|\lambda|$ we get 
\begin{align*}
\|A\|_1 = 2|\lambda| = 2 \cdot \sqrt{\sum_{ \{0,1\}^m\setminus\{1^m\}} |\bra{1^m} A \ket{i}|^2} 
	&=  2 \cdot \sqrt{\sum_{i\in \{0,1\}^m\setminus\{1^m\}} |\bra{1^m}(2\rho) \ket{i}|^2} \\
    &\le 4 \cdot \sqrt{\sum_{i\in \{0,1\}^m} |\bra{1^m}\rho \ket{i}|^2}\\
	&= 4 \cdot \sqrt{|\bra{1^m}\rho^2 \ket{1^m}|^2} \\
	&\le 4 \cdot \sqrt{\bra{1^m}\rho \ket{1^m}}\tag{\Cref{lemma_rho_rho_squared}}\\
	&= 4\sqrt{\delta}.
\end{align*}
\end{proof}

Next, we recall \Cref{lemma:small-influence-in-direction-i} and prove it.
\SmallInfInDirectionI*
\begin{proof}[Proof of \Cref{lemma:small-influence-in-direction-i}]
We can assume $\delta < 1/3$ without loss of generality, as $\td{\rho'}{\rho''} \leq 1 \le 2\delta + 2\sqrt{\delta}$ otherwise.
By the triangle inequality (and symmetry), we have that 
\begin{align}
    \td{\rho'}{\rho''} \leq \td{\rho'}{ \ketbra*{1^d}} + \td{ \ketbra*{1^d}}{\rho''} = \td{\rho'}{ \ketbra*{1^d}} + \td{\rho''}{ \ketbra*{1^d}}.
\end{align}
Therefore, it suffices to show that $\bra{1^d} \rho \ket{1^d} \geq 1-\delta$ implies $\td{\rho}{ \ketbra*{1^d}} \le \delta + \sqrt{\delta}$, for any mixed state $\rho$.
	
	Let $\rho = \sum_{i} \lambda_i\ketbra{\psi_i}$, where $\sum_i \lambda_i =1$.
	We express $\ket{1^d} = \sum_{j}\alpha_{j} \ket{\psi_j}$ as a linear combination of the eigenvectors of $\rho'$.
	Then, as in \Cref{eqn:relate},
	\[1-\delta \le \bra{1^d} \rho \ket{1^d}  = \sum_{j} |\alpha_{j}|^2 \lambda_j\le \max_{j} |\alpha_{j}|^2 \cdot \sum_{j}\lambda_j =  \max_{j} |\alpha_{j}|^2, 
	\]
	so there exists a $j$ with $|\alpha_{j}|^2 \ge 1-\delta$, and since $\delta<1/3$, $j$ is unique.
	This means that \[|\braket{\psi_j|1^d}|^2 = |\alpha_{j}|^2 \ge 1-\delta.\]
Similarly,	
\[
1-\delta \le \bra{1^d} \rho \ket{1^d}  
= \sum_{k} |\alpha_{k}|^2 \lambda_jk \le \max_{k} \lambda_k \cdot \sum_{k} |\alpha_{k}|^2=  \max_{k} \lambda_k\;,
\]
so there exists a $k$ with $\lambda_k \ge 1-\delta$, and 
since $\delta<1/3$, $k$ is unique.

Furthermore, we will now show that it must be the case that 
\begin{align}\label{eq:argmaxs_are_same}
    j={\arg\max}_i |\alpha_{i}|^2= {\arg \max}_i \lambda_i=k,
\end{align}
which implies that there is a unique $j$ such that 
\begin{align*}
    1-\delta \le |\alpha_{j}|^2 \quad \text{and}\quad 1-\delta \le  \lambda_j.
\end{align*}
We will prove this by contradiction. For contradiction, assume that \Cref{eq:argmaxs_are_same} is false so the above $j$ and $k$ are different.
Since $\sum_{i}\lambda_i = 1$ and $\lambda_k \ge 1-\delta$ we get that  $\lambda_j \le \delta$.
Using the assumption $\delta\le 1/3$ this implies that
\[
\bra{1^d} \rho \ket{1^d} = \sum_{i}|\alpha_{i}|^2 \lambda_i  \le |\alpha_{j}|^2\cdot \delta  +\sum_{i:i\neq j}|\alpha_{i}|^2  = \left(\sum_{i}|\alpha_{i}|^2\right) - (1-\delta) |\alpha_{j}|^2 \leq 1-(1-\delta)^2 \le 5/9.
\]
However, this contradicts the assumption that $\bra{1^d} \rho \ket{1^d} \geq 1-\delta > 2/3$.

Overall, we have shown that there exists an eigenvector of $\rho$, denoted $\ket{\psi_j}$, such that, by the definition of trace distance for pure states,
\begin{align}
\td{\ketbra*{\psi_j}}{ \ketbra*{1^d}}=\sqrt{1-|\braket{\psi_j|1^d}|^2} \le  \sqrt{1-(1-\delta)} = \sqrt{\delta}
\end{align}
and, since the trace norm can be expressed as the sum of the absolute values of the eigenvalues, 
\begin{align}
\td{\rho}{\ketbra*{\psi_j}}=\frac{1}{2}\cdot \left(|1-\lambda_j| + \sum_{k\neq j}|\lambda_k|\right) \le \frac{1}{2}\cdot (\delta + \delta) = \delta 
\end{align}
By triangle inequality, this therefore implies that $\td{\rho}{\ketbra*{1^d}} \le \delta + \sqrt{\delta}$.	
\end{proof}

\section{Depth-2 Circuits Cannot Construct a Nekomata} \label{sec:generalized_neko}

We will now prove that a depth-$2$ \QACZ circuit cannot exactly synthesize a \emph{generalized $n$-nekomata} even with unlimited ancillae. We will refer to the definitions in \Cref{sec:nekodef}. 

The outline of our proof is as follows. Given a depth-$d$ circuit that output a $n$-qubit nekomata, we will construct a separable state on some subset of qubits, $\ket{\veta}_Q$, such that, inside the eigenspace of $\kb{\veta}_Q$, (1) the state is still a nekomata and (2) the final layer is simplified to a single-qubit layer. The main idea behind this is the same as the block diagonalization used in \Cref{lem:d2ac0small}. To complete our proof for $d= 2$, we show that a depth-$1$ \QAC~circuit cannot compute a state that looks like an $\Omega(1)$-qubit nekomata, even after post-selecting for such $\kb{\veta}_Q$ as in \Cref{def:nekopost} 

 Then, the two main components are given by the following lemmas.  
 \begin{restatable*}[Nekomata after one layer]{lemma}{nekotognsp}\label{lem:postlayer}
    Let $\ket{\psi}$ be a $n$-nekomata. Then, for any layer $L$ of separable reflection gates, the state $\ket{\varphi}$ given by $\ket{\varphi} = L \cdot \ket{\psi}$ is an \emph{$\lceil n/2 \rceil$-GNSP.}
 \end{restatable*}

\begin{restatable*}[No GNSP in depth-$1$]{lemma}{onelayergnsp}\label{lem:d1postlb}
Let $\ket{\psi}$ be a state constructed by a depth-$1$ \QACZ circuit $C$. Then, $\ket{\psi}$ cannot be a \emph{$n$-GNSP} for any $n > 2$.  
\end{restatable*}
Then, the bound from \Cref{thm:d2_lb} immediately follows as a consequence, 
\begin{corollary}\label{cor:nekolb}
Let $\ket{\psi}$ be a state constructed by a depth-$2$ \QACZ circuit $C$ starting from the $\ket{\0}$ state. Then, $\ket{\psi}$ cannot be a generalized $n$-nekomata for $n > 4$. 
\end{corollary}
\begin{proof}
Assuming for contradiction that $\ket{\psi}$ is a generalized $n$-nekomata for $n \geq 5$. Then, letting $\ket{\varphi}$ be the state after the first layer of $C$, $\ket{\psi}$ is given by $\ket{\psi} = L \cdot \ket{\varphi}$ for the second layer $L$. 
Since all reflection gates are Hermitian, $\ket{\varphi} = L \cdot \ket{\psi}$ and due to \Cref{lem:postlayer}, $\ket{\varphi}$ is an $3$-GNSP. This is a contradiction to \Cref{lem:d1postlb}.
\end{proof}

\subsection{Proofs of Lemmata} 
First, we show the following observation.
\begin{fact}\label{fact:nekoanc}
    Let $\ket{\psi}_{T,A}$ be a generalized $n$-nekomata on targets $t_1,t_2 \dots t_n$ and ancillae $A$. Then, for any unitary $U_A$ acting only on the ancillae, the state $U_A \cdot \ket{\psi}$ is still a generalized $n$-nekomata on targets $t_1,t_2 \dots t_n$.
\end{fact}
\begin{proof}
Suppose $\ket{\psi}$ is given by, 
\begin{align}
 \ket{\psi} &=  \alpha \cdot \ket{\mu_1}_{t_1} \ket{\mu_2}_{t_2} \dots \ket{\mu_n}_{t_n} \ket{\gamma_0}_A + \beta \cdot \ket{\mu^\perp_1}_{t_1} \ket{\mu^\perp_2}_{t_2} \dots \ket{\mu^\perp_n}_{t_n} \ket{\gamma_1}_A
\end{align}
Then,
\begin{align}
   U_A \cdot \ket{\psi} &= \alpha \ket{\mu_1}_{t_1} \ket{\mu_2}_{t_2} \dots \ket{\mu_n}_{t_n} U_A\ket{\gamma_0}_A + \beta \ket{\mu_1^\perp}_{t_1} \ket{\mu_2^\perp}_{t_2} \dots \ket{\mu_n^\perp}_{t_n} U_A \ket{\gamma_1}_A 
\end{align}
which is also a generalized $n$-nekomata on the same targets. 
\end{proof}

Now we provide the proofs of the lemmas. First, we will prove a special case of \Cref{lem:postlayer} for a single gate. 
\begin{claim}\label{cl:postgate} 
Let $\ket{\psi}$ be a generalized $n$-nekomata on targets $T$ and $\ket{\varphi} = G(S) \cdot \ket{\psi}$, where $G(S) = (I - 2\kb{\vth}_S)$ is a reflection gate on a subset of qubits $S$ containing $k = \vlr{S \cap T}$ targets. 
Then, there exists a separable state $\ket{\eta}_{Q}$ on qubits $Q \subseteq S$, such that, $\kb{\veta}_Q \cdot \ket{\varphi}$ is a generalized $n-\lfloor k/2 \rfloor$-nekomata. 
\end{claim}
\begin{proof}
Wlog $\ket{\psi}$ is given by, 
\begin{align}
    \ket{\psi} &=  \alpha \cdot \ket{0^n}_T \ket{\gamma_0}_A + \beta \cdot \ket{1^n}_T \ket{\gamma_1}_A
\end{align}
For every qubit $q \in S$, let $\ket{\pth_q}$ be such that $\braket{\th_q | \pth_q} = 0$, then, wlog $\bra{\pth_q} \cdot \ket{\psi} \neq 0$, otherwise $q$ is \emph{redundant} in $G(S)$ and we can remove it to get a smaller gate. If $S \cap T = \emptyset$, then $\ket{\varphi}$ is already a generalized $n$-nekomata due to \Cref{fact:nekoanc}. Otherwise we have two main cases. 
\paragraph{At least two targets in gate, $\vlr{S \cap T} \geq 2$:}
Suppose there is a target $t \in T \cap S$ such that, $0 < \vlr{\braket{\th_{t} | 0}} < 1$. Then, for $\ket{\eta} = \ket{\pth_{t}}$ we have $\braket{\pth | 0} \neq 0$ and $\braket{\pth | 1} \neq 0$. Therefore, the following state is a generalized $(n-1)$-nekomata on $T' = T \setminus t$, 
\begin{align}
\kb{\eta}_{t} \cdot \ket{\varphi}  &= \kb{\eta}_{t} \cdot \ket{\psi}      \qquad \text{(since $\kb{\eta}_t G(S) = \kb{\eta}_t \otimes I$)} \\
&= \alpha \cdot \braket{\eta|0} \cdot \ket{\eta}_{t} \ket{0^{n-1}}_{T'} \ket{\gamma_0}_A + \beta \cdot \braket{\eta| 1} \cdot \ket{\eta}_{t} \ket{1^{n-1}}_{T'} \ket{\gamma_1}_A 
\end{align}
If no such $t$ exists, then, observe that for every target $t \in T$, 
\begin{align} 
    & [\kb{0}_{t}, G(S)]= [\kb{1}_t, G(S)] = 0 \\
    &\Rightarrow [\kb{0^n}_T, G(S)] = [\kb{1^n}_T, G(S)] = 0 
\end{align}
Then, from \Cref{sec:pregates}, for unitaries $U,V$ acting only on the ancillae $S \setminus T$,
\begin{align}
    \kb{0^n}_T \cdot G(S) &= \kb{0^n} \tens U_{S \setminus T} \\
   \kb{1^n}_T \cdot G(S) &= \kb{1^n} \tens V_{S \setminus T} 
\end{align}
Hence, $\ket{\varphi}$ is already a generalized $n$-nekomata because, 
\begin{align}
    \ket{\varphi}  &= G(S) \cdot \ket{\psi} \\
    &= \alpha \cdot \ket{0^n}_T \tens U_{S'} \ket{\gamma_0}_A + \beta \cdot \ket{1^n}_T \tens V_{S'} \ket{\gamma_1}_A
\end{align}
\paragraph{One target in gate}
Let $t = T \cap S$ and, $S' = S \setminus t$ and $T' = T \setminus S$.
Then, we will prove by induction on $m = |S'|$ that for some product state $\ket{\veta}_Q$,  $\kb{\veta}_Q \cdot \ket{\varphi}$ is actually a generalized $n$-nekomata. In the case when $m = 0$, $G(S)$ acts as a single qubit unitary on $t$ and $\ket{\varphi}$ is already a generalized $n$-nekomata. Now for $m > 1$, let $q \in S'$ be any ancilla. First, if $\bra{\th} \cdot \ket{\psi} = 0$, or $\bra{\pth_q} \cdot \ket{\psi} = 0$, we can either replace the gate with $I$ or a smaller gate with $q$ a before and apply the inductive hypothesis on $m-1$.  
Wlog assume $\ket{\th_q} = \ket{1}_q$. Suppose that for some value $b \in \bin$, 
$\bra{b}_q \cdot \ket{\gamma_0}_A \neq 0$ and $\bra{b}_q \cdot \ket{\gamma_1} \neq 0$. Then, 
let $\ket{\psi'} \propto \kb{b} \cdot \ket{\psi}$ and $\ket{\varphi'} \propto \kb{b} \cdot \ket{\varphi}$. Observe that $\ket{\psi'}$ is a generalized $n$-nekomata because, 
\begin{align}
    \kb{b}_q \cdot \ket{\psi} &= \alpha \cdot \ket{0^n}_T \lr{\kb{b}_q \cdot \ket{\gamma_0}_A} + \beta \cdot \ket{1^n}_T \lr{\kb{b}_q \cdot \ket{\gamma_1}_A} 
\end{align}
Furthermore, from \Cref{sec:pregates}, if $b = 0$, $\ket{\varphi'} = \ket{\psi'}$ and is also a generalized $n$-nekomata and we are done. If $b = 1$, $\ket{\varphi'} = G(S \setminus q) \ket{\psi'}$ where $G(S \setminus q) = (I - 2\kb{\vth}_{S \setminus q})$, and let $\ket{\veta}_{Q'}$ be the state from the inductive hypothesis on $\ket{\varphi'}$ and $\ket{\psi'}$ such that $\kb{\veta}_{Q'} \cdot \ket{\varphi'}$ is a generalized $n-1$-nekomata. Then, by definition,
\begin{align}
    \kb{\veta}_{Q'} \cdot \ket{\varphi'} &= \kb{\veta}_{Q'} \tens \kb{1}_q  \cdot \ket{\varphi}
\end{align}
and therefore, for $\ket{\eta_q}_{q} = \ket{1}_q$, $\kb{\veta}_{Q',q} \cdot \ket{\varphi}$ is a generalized $n$-nekomata. 
\end{proof}
This allows us to easily extend to general version, \Cref{lem:postlayer}.
\nekotognsp
\begin{proof}
Let $T$ be the set of targets in $\ket{\psi}$ and $Q_0$ be the set of qubits belonging to the gates in $L$ and $n_0 = \vlr{T \cap Q_0}$ be the number of targets from $T$ in these gates. We will proceed by induction on $m$, the number of gates in $L$ and additionally argue that, (1) $Q \subseteq Q_0$, (2) $\kb{\veta}_{Q} \cdot \ket{\varphi}$ is a generalized nekomata on at least $n - \lfloor n_0/2 \rfloor$ targets.

For $m = 1$, this simply follows from \Cref{cl:postgate}. For $m > 1$, let $G(S) = (I - 2\kb{\vth}_S)$ be a gate in $L$ so that $L = G(S) \tens L_1$, where $L_1$ contains $m-1$ gates and let $n_1$ be the number of targets from $T$ in $L_1$ and $n_2 = n_0 - n_1$ be the number of targets in $S$. 
 Then, from the inductive hypothesis on $\ket{\varphi_1} = L_1  \cdot \ket{\psi}$, there is a state $\ket{\veta}_{Q_1}$ only on qubits in $L_1$ such that the state $\ket{\psi_1} \propto \kb{\veta}_{Q_1} \cdot \ket{\varphi_1}$ is a generalized $n - \lfloor n_1/2 \rfloor$-nekomata. 
Now for $\ket{\varphi_2} = G(S) \cdot \ket{\psi_1}$, from \Cref{cl:postgate}, there is a state $\ket{\veta}_{Q_2}$ on $Q_2 \subseteq S$ such that $\ket{\psi_2} \propto \kb{\veta}_{Q_2} \cdot \ket{\varphi_2}$ is a generalized $n'$ nekomata for, 
\begin{align}
 n' &= n - \lfloor n_1 / 2 \rfloor - \lfloor n_2 / 2 \rfloor \\
 &\geq n - \lfloor n_0 / 2 \rfloor
\end{align}
Then, for $Q = Q_1 \cup Q_2 \subseteq Q_0$ and $\ket{\veta}_{Q} = \ket{\veta}_{Q_1} \tens \ket{\veta}_{Q_2}$, 
\begin{align}
 \kb{\veta}_Q \cdot \ket{\varphi} &= \kb{\veta_Q} \cdot \lr{G(S) \tens L_1} \cdot \ket{\psi} \\
 &= \kb{\veta_{Q_2}} \cdot G(S) \cdot \lr{\kb{\veta_{Q_1}} \cdot L_1 \cdot \ket{\psi}} \\
 &\propto \kb{\veta_{Q_2}} \cdot G(S) \cdot \ket{\psi_1} \\
 &\propto \ket{\psi_2}
\end{align}
which is a generalized $n'$-nekomata. 
\end{proof}

Finally, we will prove the base case of the result for depth-$1$ restated below, 
\onelayergnsp
\begin{proof}
Note that $\ket{\psi}$ is separable across $S_1, S_2 \dots S_m$, the subset of qubits belonging to each of the $m$ gates of $C$. Furthermore, for any separable state $\ket{\veta}_Q$, $\kb{\veta}_Q \cdot \ket{\psi}$ is also separable across $S_1, S_2 \dots S_m$. Therefore, it suffices to prove the lemma for when $C$ consisting of a single gate. Then, $\ket{\psi} = (I - 2\kb{\vth^1}_S) \cdot \ket{\vth^0}_S$ for separable states $\ket{\vth^0}, \ket{\vth^1}$ such that $0 < \vlr{\braket{\th^1_q | \th^0_q}} < 1$ for every $q \in S$. 
Now suppose for contradiction that there exists $\ket{\veta}$ on $Q \subseteq S$, s.t $\kb{\veta}_Q \cdot \ket{\psi} \propto \ket{\varphi}_{S'} \tens \ket{\veta}_Q$ and $\ket{\varphi}$ is a generalized $n \geq 3$-nekomata as below, 
\begin{align}
 \ket{\varphi} = \alpha \cdot \ket{0^n}_{T} \ket{\gamma_1}_A  + \beta \ket{1^n}_T \ket{\gamma_2}_A
\end{align}
for some $\alpha \neq 0$ and $\beta \neq 0$ such that $S' = T \cup A$.  
For any target $t \in T$, since $0 < \vlr{\braket{\th^0_t |\th^1_t}} < 1$ there must be some $b \in \bin$ for which $0 < \vlr{\braket{\th^b_t | 0}} < 1$. Let $\ket{\varphi'} \propto \kb{\mu}_t \cdot \ket{\varphi}$ for $\ket{\mu}$ such that $\braket{\mu | \th^b_t} = 0$. Then, $\ket{\varphi'}$ is a is a generalized $n-1$-nekomata because,
\begin{align}
 \ket{\varphi'} &\propto \kb{\mu} \cdot \ket{\varphi} \\
 &= \alpha \cdot \braket{\mu | 0} \cdot \ket{0^{n-1}}_{T\setminus t} \ket{\gamma_1}_A \ket{\mu}_t + \beta \cdot \braket{\mu|1} \cdot \ket{1^{n-1}}_{T\setminus t} \ket{\gamma_2}_A \ket{\mu}_t.  
\end{align}
However, observe that $\ket{\psi} \in \spn \{ \ket{\vth^1}_S, \ket{\vth^0}_S \}$ and thus $\ket{\varphi} \in \spn\{ \ket{\vth^1}_{S'}, \ket{\vth^0}_{S'}\}$. Since $\kb{\mu}_t \tens I$ is orthogonal to $\ket{\vth^b}$, for $b' = b \oplus 1$ and $S'' = S \setminus t$, we have 
$\ket{\varphi'} \in \spn\{ \ket{\vth^{b'}_{S''}} \tens \ket{\mu} \} $. This is a contradiction to $n > 2$  because $\ket{\varphi'}$ is separable across all its qubits and cannot be a generalized $2$-nekomata.  
\end{proof}

\section{Acknowledgements}
M.J thanks Lucas Gretta for many insightful discussions about reversible circuits and restrictions. M.J also thanks Meghal Gupta for a very helpful discussion on subspaces.
A.T. thanks ChatGPT for coming up with the proof idea for \Cref{lemma:closeness_implies_closeness_in_tails}. F.V. thanks ChatGPT for helpful interactions throughout the course of this project. The authors are grateful to the anonymous reviewers for their careful reading and valuable feedback. The authors also thank Gregory Rosenthal for additional extensive feedback. 


\appendix

\bibliographystyle{alpha}
\bibliography{biblio}
\section{Approximate to Exact \PARITY Deferred Proofs}\label{sec:aproofs}

First we will formally prove the form of \cite{watts2019separation} that we require. 

\avgtoworst*
\begin{proof}
  We can implement the transformation $\ket{\x} \ket{0^n} \mapsto \ket{\x} \ket{x_1 \oplus x_2, x_2 \oplus x_3 \dots x_n \oplus x_1}$ in $2$ layers of $\cnot$ gates \cite{watts2019separation}. Observe that $ \ket{x_1 \oplus x_2, \dots x_n \oplus x_1 }$ has parity $0$ regardless of $\x$.  This allows us to synthesize the below superposition of even parity strings from $\ket{0^{2n}}$ in \QNCZ. 
  \begin{align}
    \ket{\nu^*} &= \frac{1}{\sqrt{2^n}} \sum_{\x \in \bin^n} \ket{\x}_A \ket{x_1 \oplus x_2, x_2 \oplus x_3 \dots x_n \oplus x_1}_B \\
                &= \frac{1}{\sqrt{2^{n-1}}} \sum_{\y \in \bin^{n} : \oplus \y = 0} \ket{\nu(\y)}_A \tens \ket{\y}_B 
  \end{align}
  Note that $\ket{\nu^*}$ is simply the ``Poor Man's Cat State'' (\cite{watts2019separation}) due to the entangled $A$ register. Uncomputing $\ket{\nu(\y)}$ would produce a $\ket{\Cat_{n-1}}$ which is not possible in \QNCZ. 

  WLOG suppose that $\gamma > 0$ so that the output register $t$ of $C$ measures to $\PARITY(\x)$ with probability $p = 1/2 + \gamma/2$ on a random input (otherwise flip the output of $C$). We will describe a circuit $C'$ that, on any input $\x$, outputs $\PARITY(\x)$ with probability at least $p = 1/2 + \vlr{\gamma}/2$. 

  On input $\x \in \bin^n$, using $2n$ additional ancillae, construct $\ket{\nu^*}$ and then apply a single layer of $\cnot$ gates from each input qubit to a corresponding qubit in $B$, to obtain,  

  \begin{align}
    \ket{\x} \ket{0^{2n}} \mapsto \frac{1}{\sqrt{2^{n-1}}} \sum_{\y \in \bin^{n} : \oplus \y = 0} \ket{\y} \ket{\nu(\y)}_A \tens \ket{x_1 \oplus y_1, \dots x_n \oplus y_n}_B 
  \end{align}

 Then, we can in feed the $B$ registers as inputs to $C$ instead. Observe that, measuring the $B$ register produces a uniformly random string with the same \PARITY as $\x$, therefore, the output of $C'$ measures to $\PARITY(\x)$ with probability at least $p$. 
\end{proof}

\approxtoexpar*
\begin{proof}
The proof follows through a series of reductions described below.  
\begin{enumerate}
  \item Apply \cref{lem:avgtoworst} to obtain a depth-$d_0 = d + O(1)$ circuit $C_0$ with $a_0 = a + O(n)$ ancillae that correctly computes \PARITY with probability $1/2 + |\rho|/2$ on every input $\x$.   

  \item  Using \FANOUT and \MAJORITY gates of size $k = \Theta(1/\rho^2) = \Theta(\log^{2\delta} n)$, repeat $C_0$ in parallel $k$ times and output the \MAJORITY of the runs. This gives $C_1$ that, on every $\x \in \bin^n$ correctly outputs $\PARITY(\x)$ correctly wp $\geq 0.999$. $C_1$ can be implemented in depth-$d_1 = d + O(1)$ using $a_1 = a_0 \cdot O(n^O(\delta))$ ancillae from $C_0$ due to \cite{rosenthal2021qac0, grier2026mathsfqac0containsmathsftc0with}.   

  \item Turn $C_1$ into a $n+1$-output circuit $C_2$ that preserves the input registers, by simply making a classical copy of each input coordinate $x_i$ in a single layer at the start and then running $C$ using on these qubits instead. This only increases the depth to $d_3 = d_2 + 1$ and the number of ancillae to $d_3 = a_3 + O(n)$. 

  \item Apply the approximate reduction of \cite{rosenthal2021qac0} to obtain a depth-$d_4 = O(d)$ and $a_4 = O(a)$ circuit $C_4$ such that the \emph{phase dependent fidelity} of $\ket{\psi} = C_4 \cdot \ket{0^{n+a_4}}$ with $\ket{\Cat_n}_T$ is at least $0.99$. In other words, there is some ancilla state $\ket{\alpha}_A$ such that, $$1 - \Vlr{\ket{\psi}_{T,A} - \ket{\Cat_n}_T \ket{\alpha}_A}_2^2 \geq 0.9.$$ Then, the qubits $T$ of $\ket{\psi}$ have at least $0.25$ probability of measuring to all $0$s and at least $0.25$ probability of measuring to all $1$s. 

  \item Apply \cref{lem:exactcat} to obtain $C_5$ for exactly computing $\ket{\Cat_n}$ in depth-$d_5 = O(d_4)$ using $a_5 = a_4 + 1$ ancillae. 
\end{enumerate}
This produces a depth-$O(d)$ circuit with $O(a \cdot n^{O(\delta)})$ ancillae to exctly compute $\ket{\Cat_n}$. This can be turned into $C'$ for computing exact \PARITY on $n$ coordinates in depth-$O(d)$ using $a' = O(d \cdot n^{O(\delta)})$ ancillae \cite{moore1999qac0}.

The depth of $C'$ we obtain is independent in the original error and, when $\rho = 1/\poly\log(n)$, $C'$ only requires $a' = \poly(n)$ ancillae.  
\end{proof}

\aprexactcat*
\begin{proof}
  Let $\alpha_b = \vlr{\bra{b^n}_T \cdot \ket{\psi}}$ for $b \in \bin$. Assume WLOG that $\alpha_1 \geq \alpha_0$. Using a fresh ancilla $q$, we will first apply a single gate $G(T,q)$ on $\ket{\psi} \tens \ket{\mu}_q$ to produce a state $\ket{\psi_1}$ with exactly $1/{2\sqrt{2}}$ amplitude on the two branches $\ket{0^n}_T \ket{0}_q$ and $\ket{1^n}_T \ket{0}_q$. To achieve this, let $\ket{\mu} = \beta \ket{0} + \sqrt{1-\beta^2} \ket{1}$ for $\beta = \frac{1}{\alpha_0 2\sqrt{2}}$
  and choose $G(T,q)$ to be a controlled-$U_q$ gate, controlled on $T$ being $\ket{1^n}$.  
  such that $U$ maps $\ket{0} \mapsto \gamma \ket{0} + \sqrt{1-\gamma^2} \ket{1}$ for $\gamma = \frac{\alpha_0}{\alpha_1}$. Note that we can set remaining degrees of freedom to make $U = U^\dag$ (valid reflection) as, 
 \begin{align}
     U := \begin{bmatrix}
            \gamma  & \sqrt{1-\gamma^2} \\
            \sqrt{1-\gamma^2} & -\gamma
         \end{bmatrix}
 \end{align} 
  This produces,
  \begin{align}
    \vlr{ \bra{0^n,0}_{T,q} \cdot \ket{\psi_0}} &= \vlr{\bra{0^n} \cdot \ket{\psi}} \cdot \beta = \alpha_0 \cdot \beta = \frac{1}{2 \sqrt{2}} \label{eq:br1} \\
    \vlr{\bra{1^n,0}_{T,q} \cdot \ket{\psi_0}} &= \vlr{\bra{1^n} \cdot \ket{\psi} \tens U_q \ket{\mu}_q} 
= \alpha_1 \cdot \gamma \cdot \beta  = \frac{1}{2 \sqrt{2}} \label{eq:br2} 
  \end{align}
Define $\Pi^* :=  \lr{\kb{0^n}_T + \kb{1^n}_T} \tens \kb{0}_q$ 
  and let $C_0$ be this depth-$d+1$ circuit to construct $\ket{\psi_0}$ from $\ket{0^{a+1}}$, by applying a single qubit unitary $\ket{0}_q \mapsto \ket{\mu}_q$ and then $G(T,q)\cdot C$. We claim that the following circuit synthesizes an $n$-nekomata, 
  \begin{align}
    C' &:=  (I - 2\kb{\psi_0}) \cdot (I - 2\Pi^*) \cdot C_0 \\
    &= 
    C^\dag_0 (I - 2\kb{0^{a+1}}) C_0 \cdot (I - 2\kb{0^n,0}_{T,q}) \cdot (I - 2\kb{1^n,0}_{T,q}) \cdot C_0, 
  \end{align}
 Then, $C'$ only requires depth-$d' = 3(d+1)+2 \leq 3(d+2)$ and $a+1$ ancillae.  

  We now proceed to prove the remaining claim. Due to \Cref{eq:br1} and \Cref{eq:br2}, there exists a $n$-nekomata $\ket{\nu_n}_{T,A}$ such that, 
  \begin{align}\label{eq:state0}
    \ket{\psi_0} &= \frac{1}{2} \cdot \ket{\nu_n} \ket{0}_q + \frac{\sqrt{3}}{2} \cdot \ket{\varphi'} 
  \end{align}
  for some state $\ket{\varphi'}$ satisfying $\Pi^* \cdot \ket{\varphi'} = 0$. Since $\ket{\nu_n} \ket{0}_q \in \eig{\Pi^*}$, 
  \begin{align}
    \ket{\psi_1} &:= (I - 2\Pi^*) \cdot C_0 \ket{0^{a+1}} \\
    &= (I - 2\Pi^*) \cdot \ket{\psi_0} \\
    &= - \frac{1}{2} \cdot \ket{\nu_n} \ket{0}_q + \cdot \frac{\sqrt{3}}{q} \cdot \ket{\varphi'}\\ 
    &= \ket{\psi_0} - \ket{\nu_n} \ket{0}
  \end{align} 
  Hence, 
  \begin{align}
    \ket{\psi_2} &:= C' \ket{0^{a+1}} \\ 
                 &= (I - 2\kb{\psi_0}) \cdot \ket{\psi_1} \\ 
                 &= \ket{\psi_1} - \ket{\psi_0} \\
                 &= -\ket{\nu_n} \ket{0} 
  \end{align}
  which is a $n$-qubit nekomata.
\end{proof}

\end{document}